\documentclass[11pt]{article}
\usepackage[margin=1in]{geometry}
\usepackage[T1]{fontenc}
\usepackage{lmodern,microtype,amsmath,amssymb,amsthm,mathtools,bm,booktabs}
\usepackage[colorlinks=true,linkcolor=blue,citecolor=blue,urlcolor=blue]{hyperref}
\usepackage{enumitem,graphicx,needspace}
\graphicspath{{figures/}}
\allowdisplaybreaks[2]
\numberwithin{equation}{section}
\hypersetup{pdftitle={Static classical-quantum-entanglement trade-offs: an entropic converse and single-letter characterizations},pdfauthor={Mark M. Wilde}}
\newtheorem{theorem}{Theorem}[section]
\newtheorem{lemma}[theorem]{Lemma}
\newtheorem{proposition}[theorem]{Proposition}
\newtheorem{corollary}[theorem]{Corollary}
\newtheorem{conjecture}[theorem]{Conjecture}
\theoremstyle{definition}
\newtheorem{definition}[theorem]{Definition}
\theoremstyle{remark}
\newtheorem{remark}[theorem]{Remark}
\newcommand{\Tr}{\operatorname{Tr}}
\newcommand{\id}{\operatorname{id}}
\newcommand{\diag}{\operatorname{diag}}
\newcommand{\supp}{\operatorname{supp}}
\newcommand{\conv}{\operatorname{conv}}
\newcommand{\EE}{\mathbb E}
\newcommand{\ket}[1]{\lvert#1\rangle}
\newcommand{\bra}[1]{\langle#1\rvert}
\newcommand{\proj}[1]{\ket{#1}\!\bra{#1}}

\newcommand{\cC}{\mathcal C}
\newcommand{\cR}{\mathcal R}
\newcommand{\cU}{\mathcal U}
\newcommand{\cN}{\mathcal N}
\newcommand{\cE}{\mathcal E}
\newcommand{\cT}{\mathcal T}

\newcommand{\sS}{\mathsf S}
\newcommand{\sJ}{\mathsf J}
\newcommand{\Cr}{C_{\mathrm r}}
\newcommand{\mc}{\mathrm{mc}}
\newcommand{\fp}{\mathrm{fp}}
\newcommand{\eff}{\mathrm{eff}}
\newcommand{\reg}{\mathrm{reg}}
\newcommand{\atanh}{\operatorname{atanh}}

\title{\textbf{Static classical--quantum--entanglement trade-offs:\\
an entropic converse and single-letter characterizations}}
\newcommand{\authorname}{Mark M. Wilde}
\newcommand{\authoremail}{wilde@cornell.edu}
\author{\texorpdfstring{\authorname\\[0.4em]
\small \textit{School of Electrical and Computer Engineering, Cornell University}\\
\small \textit{Ithaca, New York 14853, USA}\\
\small\texttt{\authoremail}}{\authorname}}

\date{\today}
\begin{document}
\maketitle
\begin{abstract}
The direct static capacity region of a bipartite quantum state describes its
asymptotic conversion into classical communication, quantum communication, and
entanglement, with all communication directed from Alice to Bob. This paper gives a unified
finite-block converse for this region. The proof treats generated and consumed
resources simultaneously and derives all three information inequalities for one
local instrument. Its main ingredients are no-signalling, the chain rule, strong
subadditivity, and continuity of conditional entropy. The achievability part follows by
composing instrument compression with quantum side information and state
redistribution, which establish the achievable rate region. A supporting-hyperplane
characterization gives the static capacity formula.
For a pure state of entanglement entropy $h$ subjected to erasure with probability~$p$,
the complete region is the convex hull of $0$ and $(0,-ph,(1-p)h)$, plus the
unit-resource cone. This statement holds for arbitrary collective instruments by
using subset-entropy inequalities related to the methods of Mamindlapally and
Winter. Complete characterizations also hold for symmetrically extendible states and locally
flagged pure-state mixtures. For maximally correlated states, the optimization
over local instruments admits an exact matrix formulation and an additive
outer bound. An explicit qubit example shows that a nonorthogonal discarded
quantum memory can outperform every efficient instrument and every instrument
with conditionally commuting discarded states.
For states obtained by qubit dephasing of Schmidt-aligned pure states, an exact
threshold characterizes when the static capacity formula vanishes, at every
blocklength and after regularization. Full single-letter characterizations
for general Hadamard states and for the remaining trade-offs for dephased states
remain open, with sufficient tensorization conditions identified.
\end{abstract}

\tableofcontents

\section{Introduction}

\subsection{Background and motivation}
A noisy bipartite state can supply entanglement and can improve communication
protocols when combined with noiseless communication. The rates at which these
resources can be generated are constrained by trade-offs. The direct static
capacity region quantifies these trade-offs
by keeping track of three net rates: classical communication~$C$, quantum
communication~$Q$, and entanglement~$E$. Positive rates indicate generation, and
negative rates indicate consumption. In the setting that we consider here, classical and quantum communication are
both from Alice to Bob, and backward communication is not free.

Ref.~\cite{HW10} established the regularized direct static capacity theorem by
combining classically assisted state redistribution with the three unit
protocols, and by proving a converse through reductions among resource
quadrants and octants. Ref.~\cite{WH12} subsequently developed a
unified entropy-based converse and a supporting-hyperplane formula for the
dynamic problem that involves the three aforementioned resources and a quantum
channel. Chapter~25 of~\cite{Wilde17} gives a detailed
exposition of that approach. A corresponding treatment of the static
problem should account for generated and consumed resources in a single
argument and identify the entropy optimization that determines every supporting
hyperplane of the region.

The dynamic trade-off theory also has applications to optical communication.
Ref.~\cite{WHG12a,WHG12b} applied it to bosonic channels,
including the pure-loss channel, and exhibited achievable trade-offs that
improve on time sharing between protocols for individual tasks.
Ref.~\cite{QW17} subsequently determined the corresponding
classical--quantum--entanglement capacity region of quantum-limited
bosonic amplifier channels under a mean photon-number constraint.
These results motivate the analogous static problem for shared bosonic
states, which is discussed among the future directions in
Section~\ref{sec:future}.

An additional question is when collective instruments can be replaced by
single-copy ones. Hadamard channels~\cite{KMNR07} have additive dynamic trade-off formulas
\cite{BHTW10,WH12}, and the corresponding state problem was explicitly proposed
in the conclusion of Ref.~\cite{BHTW10}. The static optimization includes a
quantum system that Alice may discard, and so the channel results do not immediately
resolve the state problem. One face is already determined by the result
of~\cite[Proposition~4]{LDS18} on the one-way distillable entanglement of degradable
states. Understanding the other faces requires
control of both the classical outcome and the discarded quantum system.

States resulting from a quantum erasure channel provide a useful setting in which to study these questions.
In the dynamic setting, the
subset-entropy methods of~\cite{MW23}, particularly
their Lemmas~3 and~4 and their erasure-channel analysis, control averages over
random subsets without assuming that the underlying state input to the channel is a product state.
Those lemmas build on the techniques of~\cite{GHW22}. For a static
resource, Alice's instrument must be chosen without knowledge of Bob's erasure
pattern. A proof must therefore apply to a fixed bipartite state and arbitrary
local instruments, rather than infer a state-capacity theorem from the dynamic
capacity of the associated channel.

\subsection{Summary of contributions}

This paper gives a unified entropic treatment of the direct static capacity theorem and characterizes the resulting trade-offs for several classes of bipartite states. The general treatment includes achievability, a converse that accounts simultaneously for generated and consumed resources, and a supporting-hyperplane capacity formula. Complete capacity regions follow for several state families, including erased states. For maximally correlated states and states obtained from qubit dephasing, structural reductions and exact boundary results give partial characterizations of the region. The analysis also demonstrates that allowing a discarded quantum system can strictly improve the trade-off, identifying an important consideration in the optimization over local instruments. These findings clarify which regions admit single-letter descriptions and isolate the remaining questions for broader classes of Hadamard states.

Theorems~\ref{thm:static} and~\ref{thm:finite-converse} give the established
static capacity theorem a unified entropic proof. The main step is an auxiliary
instrument that includes the reference of the transmitted quantum information
and both shares of any initially supplied maximally entangled state in its
mathematical dilation. No-signalling then removes these auxiliary correlations
from the first communication bound. The same instrument proves all three
facets of the capacity region, including for protocols that consume entanglement.
Theorem~\ref{thm:formula} derives the supporting-hyperplane capacity formula,
including its parameter domain and regularization.

For erased states, Theorems~\ref{thm:erased-support} and~\ref{thm:erased-region}
extend the erased Bell-state characterization of Ref.~\cite{HW10} to
arbitrary pure bipartite inputs. They give an additive static capacity
formula and the complete capacity region. The proof
uses strong subadditivity and weak monotonicity to establish the required
Bernoulli subset-entropy inequalities. Theorems~\ref{thm:extendible}
and~\ref{thm:flagged} determine two further complete families: symmetrically
extendible states and locally flagged pure-state mixtures.
Theorem~\ref{thm:had-face} records the coherent-information face of a Hadamard
state as a consequence of Ref.~\cite{LDS18}; an independent entropic proof is
included in Appendix~\ref{app:hadamard}.

For maximally correlated states, Theorem~\ref{thm:matrix-formula} gives an exact
matrix optimization at every blocklength, and Corollary~\ref{cor:mc-outer}
gives an additive outer bound using Theorem~\ref{thm:Jadd}. Theorem~\ref{thm:quantum-memory} exhibits a two-qubit state
for which a nonorthogonal discarded quantum memory strictly improves the
trade-off over all efficient instruments and all instruments with
conditionally commuting discarded states. For Schmidt-aligned qubit dephasing,
Theorem~\ref{thm:qubit-onecopy} reduces the one-copy problem to a concave
envelope of a two-variable binary-entropy expression.
Theorem~\ref{thm:threshold} determines exactly when the static capacity
formula vanishes, uniformly over collective instruments, and Corollary~\ref{cor:origin-slope} determines
the entanglement gain at vanishing quantum-communication cost. The proof uses
an entropy-contraction inequality valid with an arbitrary untouched auxiliary
system, established in Theorem~\ref{thm:contraction}. The intermediate
positive-support curve remains unresolved. We also explain why an arbitrary
pure input to a dephasing channel need not produce a maximally correlated state.

The direct proof uses instrument compression with quantum side information
\cite{WHBH12} and quantum state redistribution~\cite{DY08,YD09} as stated
coding theorems. Their underlying random-coding and decoupling proofs are not
reproduced. We give their composition, removal of shared randomness, resource
accounting, regularization, and the entire converse. Identification with
conventional resource accounting additionally invokes the established direct
coding theorem of Ref.~\cite[Theorem~2 and Section~VI-A]{HW10}, through the
comparison given after the converse.

Table~\ref{tab:scope} summarizes the precise status of the state-family
results. An exact one-copy optimization is distinguished from an operational
capacity formula, which additionally requires control of collective instruments.
\begin{table}[htbp]
\centering\small
\begin{tabular}{@{}p{0.29\textwidth}p{0.65\textwidth}@{}}
\toprule
Resource family & Result and remaining question \\
\midrule
Erased pure states & Complete region and additive support, including arbitrary collective instruments; Theorems~\ref{thm:erased-support} and~\ref{thm:erased-region}. \\[4pt]
Symmetrically extendible states & Complete region $\cU$; Theorem~\ref{thm:extendible}. \\[4pt]
Locally flagged pure mixtures & Complete region $(0,0,D)+\cU$; Theorem~\ref{thm:flagged}. \\[4pt]
General Hadamard states & Additive coherent-information face and exact penalty identity; full single-letterization remains open; Theorem~\ref{thm:had-face} and Eq.~\eqref{eq:had-penalty}. \\[4pt]
Maximally correlated states & Exact one-copy matrix optimization and additive outer bound; equality with the full capacity region remains open; Theorem~\ref{thm:matrix-formula} and Corollary~\ref{cor:mc-outer}. \\[4pt]
Schmidt-aligned qubit dephasing & Exact threshold for a vanishing static capacity formula and exact origin slope; intermediate trade-offs remain open; Theorem~\ref{thm:threshold} and Corollary~\ref{cor:origin-slope}. \\[4pt]
Nonaligned qubit dephasing & Exact coherent-information face and an achievable primitive point; Theorem~\ref{thm:had-face} and Eq.~\eqref{eq:nonaligned-primitive}; full region unresolved. \\
\bottomrule
\end{tabular}
\caption{Scope of the exact results and remaining optimizations. Here $\cU$
is the cone generated by the three unit-resource protocols, and $D$ is the
average pure-state entanglement of a locally flagged mixture.}
\label{tab:scope}
\end{table}

The main text contains the coding theorem, the complete entropic converse,
the derivation of the capacity formula, and the principal results for each
state family. Appendix~\ref{app:direct} gives the detailed direct construction.
Appendices~\ref{app:erased}--\ref{app:contraction} contain the longer
state-family proofs. Figures~\ref{fig:static-protocol} and~\ref{fig:auxiliary-instrument}
describe the task and its converse instrument, while
Figure~\ref{fig:erased-three-dimensional} shows the full erased-state region,
and Figures~\ref{fig:erased-slices} and~\ref{fig:dephasing-slices} display
capacity slices and bounds. Appendix~\ref{app:figures} specifies the numerical
construction used in the latter figure. All code and inputs needed to
reproduce the figures, together with reproduction instructions, are
available as ancillary files accompanying the arXiv posting of this paper
in the \texttt{anc/} directory.

\Needspace{12\baselineskip}
\section{Notation and the operational problem}\label{sec:notation}
This section fixes the entropy notation and defines the communication task,
including its error criterion and resource accounting. It also introduces
the unit-resource cone and the local instruments that appear in the
capacity theorem.

All quantum systems in this paper are finite dimensional, and every theorem
should be understood under this convention. We write
$\mathbb{N}\coloneqq\{1,2,\ldots\}$ and $[n]\coloneqq\{1,\ldots,n\}$ for $n\in\mathbb N$. We use $\log_2$ for the binary logarithm and $\ln$ for the natural logarithm. For a density operator $\omega_{ABD}$ on systems $A$, $B$, and $D$, define
\begin{align}
H(A)_\omega&\coloneqq-\Tr[\omega_{A}\log_2\omega_{A}],\\
H(A|B)_\omega&\coloneqq H(AB)_\omega-H(B)_\omega,\\
I(A;B)_\omega&\coloneqq H(A)_\omega-H(A|B)_\omega,\\
I(A;B|D)_\omega&\coloneqq H(A|D)_\omega-H(A|BD)_\omega,\\
I(A\rangle B)_\omega&\coloneqq-H(A|B)_\omega.
\end{align}
The relative entropy is defined as
\begin{equation}
D(\omega\|\tau)\coloneqq\Tr[\omega(\log_2\omega-\log_2\tau)]
\end{equation}
when
$\supp\omega\subseteq\supp\tau$, and it is equal to $+\infty$ otherwise.
We write $\pi_d\coloneqq I/d$ and
$\Phi_d\coloneqq\proj{\Phi_d}$, where
$\ket{\Phi_d}\coloneqq d^{-1/2}\sum_{i=1}^d\ket i\ket i$. 
Subsystem labels or the surrounding dimensions determine $d$ when it is omitted.
For a unit vector $\ket\xi$, we also write $\xi\coloneqq\proj\xi$ for its density
operator when the meaning is clear.
For a Hermitian operator $M$, the notation $M\geq0$ means that $M$ is
positive semidefinite. System labels appear as subscripts on states and channels.
For a system $A$, its copies are $A_1,\ldots,A_n$, and
$A^n\coloneqq A_1\cdots A_n$ denotes their joint system. If a system label
already contains a descriptive subscript, a separate numerical index denotes
the copy; for example, $A_{\mathrm f}^n\coloneqq
A_{\mathrm f,1}\cdots A_{\mathrm f,n}$.
The three net rates are denoted throughout by $C$, $Q$, and $E$. For classical $X$, conditional entropies are taken as averages over $x$.
For $t\in[0,1]$, the binary entropy is defined as $h_2(t)\coloneqq-t\log_2t-(1-t)\log_2(1-t)$,
with $0\log_2 0\coloneqq0$. The normalized trace distance of states $\omega$ and $\tau$ is given by $\tfrac12\left\|\omega-\tau\right\|_1$.

\subsection{Codes and net rates}
Fix $n\in\mathbb{N}$ copies of a bipartite state $\rho_{AB}$ and a purification $\psi_{ABE}$ of it.
A code has six nonnegative total resource amounts
\begin{equation}
(c_g,q_g,e_g;c_c,q_c,e_c),
\label{eq:resource-amounts}
\end{equation}
measured in bits, qubits, and ebits, respectively. The subscripts $g$ and $c$
indicate generation and consumption. Lowercase letters denote total resource
amounts for a block protocol, whereas uppercase letters denote rates per copy
of the shared state. The resource amounts may depend on~$n$; this dependence
is left implicit.

The information-processing task begins with
Alice receiving a uniform classical message $M$ of dimension $2^{c_g}$ and one share
$A_Q$ of $\Phi_{RA_Q}$ of Schmidt rank $2^{q_g}$. Alice and Bob also share
$\Phi_{T_AT_B}$ of Schmidt rank $2^{e_c}$. Alice sends a classical register
$V$, such that $\log_2|V|\leq c_c$, and a quantum register~$W$, such that
$\log_2|W|\leq q_c$. She retains a register $L$ of dimension $2^{e_g}$.
Bob decodes registers $\widehat M,\widehat Q,\widehat L$, with
$\widehat M$ classical; its dephasing in the message basis is included
in the decoder.
Integer rounding of logarithmic dimensions has no asymptotic effect.
The message, transmitted Bell pair, supplied entanglement, and resource
purification are initially independent. Their joint state is
\begin{equation}
\zeta^{\rm in}_{MRA_QT_AT_BA^nB^nE^n}
\coloneqq\pi_M\otimes\Phi_{RA_Q}\otimes\Phi_{T_AT_B}
\otimes\psi_{ABE}^{\otimes n}.
\label{eq:initial-code-state}
\end{equation}
An untouched classical copy of $M$ is retained by Alice for testing; its duplicate
on Alice's side is implicit in this notation.

The ideal final state on the tested registers is
\begin{equation}
\Omega\coloneqq\frac1{|M|}\sum_m\proj m_{M}\otimes\proj m_{\widehat M}
\otimes\Phi_{R\widehat Q}\otimes\Phi_{L\widehat L}.
\label{eq:ideal}
\end{equation}
The code has error $\varepsilon$ when its marginal on these registers is within
normalized trace distance $\varepsilon$ of~$\Omega$. This is the simultaneous
classical transmission, entanglement transmission, and entanglement generation
criterion. Local ancillas are assumed to be available to Alice and Bob for free.
All communication goes forward from Alice to Bob. Consequently
Bob's intermediate local operations can be deferred to a single final decoder:
they cannot affect a later operation of Alice, and the decoder can replay them
in their original order after storing all received registers. Alice's forward
messages can likewise be grouped into $V,W$. See Figure~\ref{fig:static-protocol}
for a visual depiction of the protocol.

\begin{figure}[t]
\centering
\includegraphics[width=\linewidth]{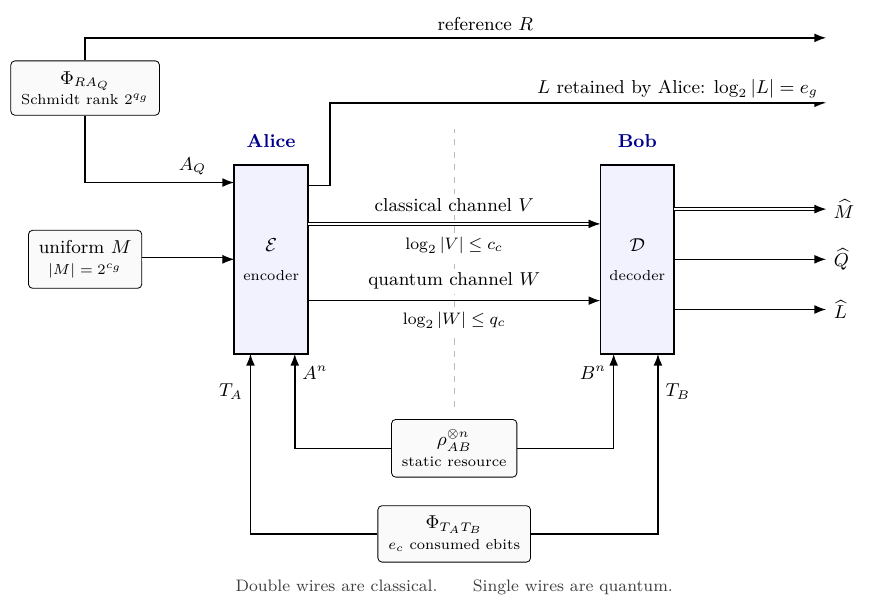}
\caption{A code for the direct static task, with resource amounts and net rates as
in~\eqref{eq:net}. Alice and Bob share $\rho_{AB}^{\otimes n}$ and $e_c$
supplied ebits in $\Phi_{T_AT_B}$. Alice encodes her shares, a uniform
classical message~$M$, and the input $A_Q$ entangled with reference $R$.
She transmits $V$ and $W$ and retains~$L$. Bob uses $B^nT_BVW$ to recover
$\widehat M,\widehat Q,\widehat L$. The error criterion~\eqref{eq:ideal}
tests the classical message and both maximally entangled pairs jointly.
An untouched classical copy of $M$ is implicit; local discarded systems
are omitted. Double wires are classical and single wires are quantum.}
\label{fig:static-protocol}
\end{figure}

The net rates are defined as
\begin{equation}
C\coloneqq\frac{c_g-c_c}{n},\qquad
Q\coloneqq\frac{q_g-q_c}{n},\qquad
E\coloneqq\frac{e_g-e_c}{n}.
\label{eq:net}
\end{equation}
We allow simultaneous consumption and generation of the same resource, with
all six resource amounts taken to be $O(n)$. Thus the definition includes finite-rate catalytic use
of the three ideal resources. A triple is achievable if it is a limit of such
net rates with error tending to zero. The closure of the set of achievable triples
is denoted $\cC_{\rm stat}(\rho)$. The conventional definition in Ref.~\cite{HW10} specifies consumed and
generated resources according to the signs of the net rates. Its codes
are included in the definition adopted here, which also permits
simultaneous consumption and generation of the same resource. After the
converse, we show that this additional freedom leaves the asymptotic
region unchanged. No cancellation of supplied and generated resources
is assumed in the finite-block argument.

\subsection{The unit-resource cone}
Teleportation, superdense coding, and entanglement distribution are specified by the following rate vectors:
\begin{equation}
u_{\rm TP}\coloneqq(-2,1,-1),\quad
u_{\rm SD}\coloneqq(2,-1,-1),\quad
u_{\rm ED}\coloneqq(0,-1,1),
\label{eq:unit-rays}
\end{equation}
respectively~\cite{Teleport,Dense}. Their nonnegative span is given by \cite{HW10}
\begin{equation}
\cU\coloneqq\{(C,Q,E):C+2Q\leq0,\ Q+E\leq0,\ C+Q+E\leq0\}.
\label{eq:unit}
\end{equation}
To verify the converse inclusion, write
\begin{equation}
(C,Q,E)=a\,u_{\rm TP}+b\,u_{\rm SD}+d\,u_{\rm ED},
\end{equation}
with $a,b,d\in\mathbb{R}$.
Since
$u_{\rm TP}=(-2,1,-1)$,
$u_{\rm SD}=(2,-1,-1)$, and
$u_{\rm ED}=(0,-1,1)$, each of the three defining linear
expressions isolates one coefficient:
\begin{equation}
C+Q+E=-2a,\qquad Q+E=-2b,\qquad C+2Q=-2d.
\end{equation}
Solving for the coefficients and substituting them back gives
\begin{equation}
(C,Q,E)
=-\frac{C+Q+E}{2}u_{\rm TP}
-\frac{Q+E}{2}u_{\rm SD}
-\frac{C+2Q}{2}u_{\rm ED}.
\end{equation}
This identity holds for every triple $(C,Q,E)$. If
$C+Q+E\leq0$, $Q+E\leq0$, and $C+2Q\leq0$, all three
coefficients are nonnegative. Thus every triple satisfying
these inequalities belongs to the cone generated by the
three unit-resource vectors.
Resource disposal is included in this cone.

\subsection{Local instruments}
A finite instrument on $A$ has a dilation by linear maps
\begin{equation}
V_x\colon A\longrightarrow SF,\qquad \sum_xV_x^\dagger V_x=I_A.
\label{eq:instrument}
\end{equation}
Here $S$ is its quantum output, $X$ its classical outcome, and $F$ its
Stinespring environment. For a fixed purification $\psi_{ABE}$ of $\rho_{AB}$, define the outcome
probability $p_x\coloneqq\Tr[(V_x^\dagger V_x\otimes I_{BE})\psi_{ABE}]$.
The corresponding state is given by
\begin{equation}
\sigma_{XSFBE}\coloneqq\sum_xp_x\proj x_{X}\otimes\psi^{x}_{SFBE},
\qquad
\ket{\psi^x}_{SFBE}\coloneqq p_x^{-1/2}(V_x\otimes I_{BE})\ket\psi_{ABE}.
\label{eq:instrument-state}
\end{equation}
Only positive-probability outcomes need to be included. Each conditional state $\psi^{x}_{SFBE}$ is pure.
A classical-output instrument is not a single global isometry after the outcome
has been dephased.
Every finite instrument admits such a representation: if its map for outcome \(x\) has Kraus operators \(\left(K_{x,j}\right)_j\), define \(V_x\coloneqq\sum_j K_{x,j}\otimes|j\rangle_F\). Tracing out (discarding) \(F\) recovers the original map for each outcome.
 We impose no restriction on its finite
auxiliary dimensions. In the direct protocol, $F$ may remain available to Alice;
``discarded'' refers to the reduced instrument used in the information formula.

We define the one-copy information region $\cR^{(1)}(\rho)$ as being equal to the union, over all such instruments, of the rate triples satisfying \cite{HW10}
\begin{align}
C+2Q&\leq-I(SX;E|B)_{\sigma},\label{eq:facet1}\\
Q+E&\leq I(S\rangle BX)_{\sigma},\label{eq:facet2}\\
C+Q+E&\leq I(S\rangle BX)_{\sigma}-I(X;E|B)_{\sigma}.\label{eq:facet3}
\end{align}
All three quantities are evaluated on the state $\sigma_{XSFBE}$ in
\eqref{eq:instrument-state}.

\section{The static capacity theorem}\label{sec:theorem}
The static capacity theorem identifies the operational region with a
regularized union of the information regions defined in
\eqref{eq:facet1}--\eqref{eq:facet3}. After stating the theorem, we recall
the protocol that achieves its generating rate point. The complete
entropic converse is given in Section~\ref{sec:converse}.

\begin{theorem}[Direct static capacity theorem~\cite{HW10}]\label{thm:static}
For every bipartite state $\rho_{AB}$,
\begin{equation}
\cC_{\rm stat}(\rho)
=\overline{\bigcup_{n\in\mathbb{N}}\frac1n\cR^{(1)}(\rho^{\otimes n})}.
\label{eq:regularized-region}
\end{equation}
Here $\cR^{(1)}$ is the information region in~\eqref{eq:facet1}--\eqref{eq:facet3},
and $\cC_{\rm stat}$ is the operational region defined in
Section~\ref{sec:notation}. The same region is obtained under the
conventional resource-accounting definition of Ref.~\cite{HW10}.
\end{theorem}
The overline in~\eqref{eq:regularized-region} denotes closure in $\mathbb R^3$.
It includes boundary rate triples approached by code sequences with vanishing error
and arbitrarily small rate slack. The union over finite blocklengths and
finite instruments need not itself contain all such limits.
The direct construction is recalled below, and its detailed proof is given in Appendix~\ref{app:direct}. The converse in
Section~\ref{sec:converse} proves all three inequalities for a single induced
instrument, without splitting the rate space into octants, as was done previously in \cite{HW10}.

\subsection{Achievability and the generating rate point}
Classically assisted state redistribution combines a local instrument with
classical communication to convey its outcome and quantum communication
to transfer its quantum output. This protocol was introduced in
Ref.~\cite[Lemma~1 and Section~V]{HW10}. In the formulation used here,
instrument compression with quantum side information first makes the
classical outcome available to Alice and Bob~\cite[Section~5.6.1]{WHBH12},
and conditional state redistribution then transfers the quantum output to
Bob~\cite{DY08,YD09}. The environment of the local instrument remains
available to Alice during the redistribution step.

For the state $\sigma_{XSFBE}$ in~\eqref{eq:instrument-state}, classically
assisted state redistribution achieves the net rates
\begin{align}
C_0&\coloneqq-I(X;E|B)_\sigma,\label{eq:CASR-C}\\
Q_0&\coloneqq-\tfrac12 I(S;E|BX)_\sigma,\label{eq:CASR-Q}\\
E_0&\coloneqq\tfrac12\big[I(S;B|X)_\sigma-I(S;F|X)_\sigma\big].
\label{eq:CASR-E}
\end{align}
Instrument compression first makes the classical outcome available to both
parties. Conditional state redistribution then transfers $S$ to Bob, using
$F$ and $B$ as the respective side information at Alice and Bob. Purity of the overall state
conditioned on $X$ and the chain rule imply that
\begin{align}
C_0+2Q_0&=-I(SX;E|B)_\sigma,\\
Q_0+E_0&=I(S\rangle BX)_\sigma,\\
C_0+Q_0+E_0&=I(S\rangle BX)_\sigma-I(X;E|B)_\sigma.
\label{eq:corner-facets}
\end{align}
Thus adding teleportation, superdense coding, and entanglement distribution
gives exactly the inequalities in~\eqref{eq:facet1}--\eqref{eq:facet3}.
Appendix~\ref{app:direct} gives the complete composition argument, including
fixed resource dimensions, conditioning on typical outcome sequences, removal
of shared randomness, and regularization over instruments acting on several
copies. The proof uses instrument compression and state redistribution as
stated coding theorems; their underlying random-coding proofs are not repeated.

\section{A unified finite-block entropic converse}\label{sec:converse}
We now prove the converse to~\eqref{eq:regularized-region} directly from the
code definition in Section~\ref{sec:notation}. Throughout this section, $B$
and $E$ abbreviate $B^n$ and $E^n$, respectively. The argument applies to an
arbitrary resource state; the tensor-power assumption is
needed only when we take the asymptotic limit.

\subsection{The induced instrument and the no-signalling identity}
Let $\omega_{MVRLWT_BFBE}$ denote the joint state after Alice's encoder and
before Bob's decoder, as represented in Figure~\ref{fig:auxiliary-instrument}. We choose a dilation of the encoder for which the
conditional state on $RLWT_BFBE$, given $M=m$ and $V=v$, is pure. The system
$F$ contains every other encoder output. Thus $\omega$ has the form required
of an instrument state in~\eqref{eq:instrument-state}. To keep track of its
outputs and the registers used in the decoding argument, define
\begin{equation}
\begin{aligned}
K&\coloneqq RL, & D&\coloneqq WT_B, & S&\coloneqq KD=RLWT_B,\\
X&\coloneqq MV, & Z&\coloneqq BDV.
\end{aligned}
\label{eq:converse-registers}
\end{equation}
In the instrument description, $B$ is Bob's original resource system $B^n$,
and $S$ and $X$ are the quantum and classical outputs defined in
\eqref{eq:converse-registers}. The additional notation $Z=BDV=B^nWT_BV$
denotes the complete input to Bob's decoder: it includes $B^n$, his share
$T_B$ of the supplied entanglement, and the received messages $W,V$.
Thus $Z$ contains $B$ but is not identified with it. In particular,
$ZM=DBX$ after regrouping the registers. The entropy calculations below
use this regrouping to express bounds involving the decoder input $Z$
in terms of the instrument systems $S,B,X$. All three information
quantities in~\eqref{eq:facet1}--\eqref{eq:facet3} are evaluated on this
same state $\omega$.

To verify that this is a legitimate instrument on Alice's original resource
system $A^n$, prepare $\Phi_{RA_Q}$ and $\Phi_{T_AT_B}$ locally as auxiliary
states, choose $m$ uniformly, and apply the encoder to $A^nA_QT_A$. Declare
$RLWT_B$ to be the quantum output and $m,v$ to be the classical output.
The resulting map is trace preserving when we sum over $m,v$ and trace all
of its quantum outputs, so its instrument operators obey the completeness
condition in~\eqref{eq:instrument}. This construction reproduces
$\omega_{XSFBE}$. It is an auxiliary description used in the information
formula: in the physical code, $T_B$ is at Bob and $R$ is the untouched
reference of the quantum message. The $BE$ marginal remains the marginal
of the initial resource purification.

Figure~\ref{fig:auxiliary-instrument} shows this construction.
\begin{figure}[t]
\centering
\includegraphics[width=\linewidth]{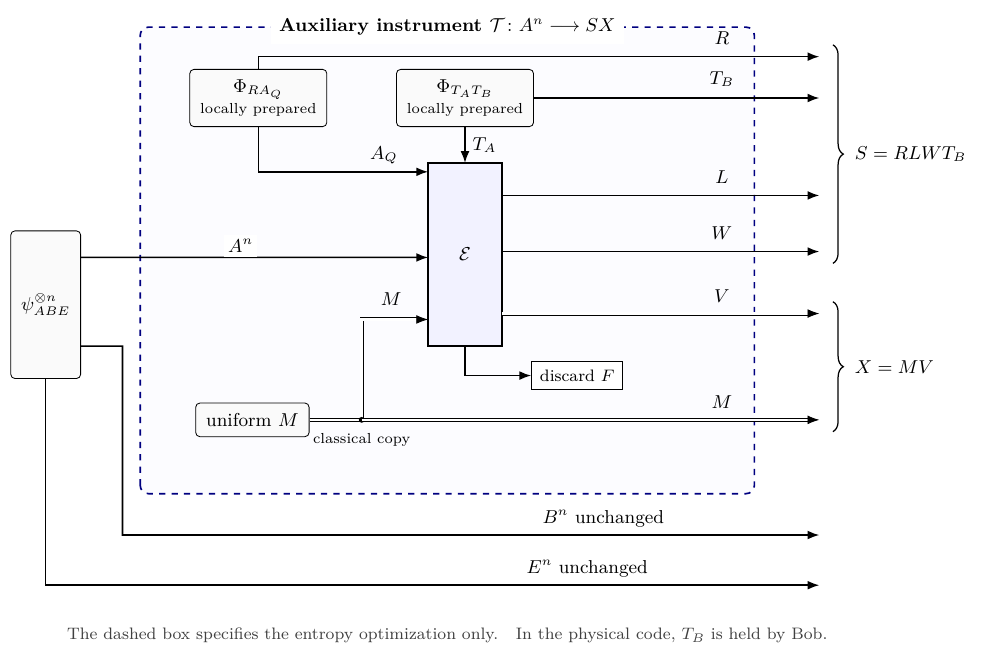}
\caption{The instrument used in the converse proof. The dashed box
specifies the auxiliary map $\mathcal T\colon A^n\to SX$ in
\eqref{eq:converse-registers}. Both $\Phi_{RA_Q}$ and $\Phi_{T_AT_B}$ are
prepared locally in this auxiliary construction. The encoder is the
one in Figure~\ref{fig:static-protocol}, with $M$ chosen uniformly.
Its quantum output is grouped as $S=RLWT_B$, and its classical output is
$X=MV$. Only the classical register $M$ is copied. The environment $F$
purifies each conditional state and is discarded in the reduced map.
The original $B^nE^n$ systems are untouched. In the physical protocol,
$T_B$ remains at Bob; placing it in the auxiliary output supplies no
additional communication or entanglement to the code.}
\label{fig:auxiliary-instrument}
\end{figure}

The following exact no-signalling identity holds:
\begin{equation}
\omega_{MRT_BBE}
=\pi_M\otimes\pi_R\otimes\pi_{T_B}\otimes\psi_{BE}.
\label{eq:nosignal}
\end{equation}
Here $\psi_{BE}$ denotes the marginal of the purification of the entire input
resource block. For each fixed~$m$, summing over the encoder outcomes gives
a trace-preserving map on $A^nA_QT_A$. Such a local map leaves the marginal
on $RT_BBE$ unchanged. The product structure of~\eqref{eq:initial-code-state} and averaging over
the uniform message $m$ therefore give~\eqref{eq:nosignal}. Consequently,
\begin{equation}
\begin{aligned}
I(MR;BT_B)_\omega&=0, & I(MRT_B;E|B)_\omega&=0,\\
I(M;BE)_\omega&=0, & H(MR)_\omega&=c_g+q_g.
\end{aligned}
\label{eq:nosignal-consequences}
\end{equation}
These are exact identities, including for a code with nonzero decoding error.

For use below, let $U_{\mathcal D}\colon Z\to\widehat M\widehat Q\widehat L G$
be a Stinespring isometry of Bob's decoder, and define its output state by
\begin{equation}
\widetilde\omega_{MRLFE\widehat M\widehat Q\widehat L G}
\coloneqq(I_{MRLFE}\otimes U_{\mathcal D})\omega_{MRLFEZ}
(I_{MRLFE}\otimes U_{\mathcal D}^{\dagger}).
\label{eq:decoded-converse-state}
\end{equation}
Tracing out $G$ gives the actual decoder output. The code
criterion mentioned after~\eqref{eq:ideal} is
\begin{equation}
\tfrac12\left\|
\widetilde\omega_{MRL\widehat M\widehat Q\widehat L}
-\Omega_{MRL\widehat M\widehat Q\widehat L}
\right\|_1\leq\varepsilon.
\label{eq:converse-error}
\end{equation}
Thus $\omega$, $\widetilde\omega$, and $\Omega$ denote, respectively, the
state before decoding, its isometric decoder output, and the ideal tested
output. Keeping these states distinct will make every use of data processing
and continuity explicit.

\subsection{Four consequences of a small decoding error}

It suffices to consider $0\leq\varepsilon\leq1/2$. Define
\begin{equation}
g(\varepsilon)\coloneqq(1+\varepsilon)
 h_2\!\left(\frac{\varepsilon}{1+\varepsilon}\right).
\end{equation}
For states $\xi_{UY}$ and $\zeta_{UY}$ at normalized trace distance at most $\varepsilon$,
the following conditional-entropy continuity bound holds~\cite{Winter16}:
\begin{equation}
\left|H(U|Y)_\xi-H(U|Y)_\zeta\right|
\leq2\varepsilon\log_2|U|+g(\varepsilon).
\label{eq:AFW}
\end{equation}
The sharp entropy-continuity inequality of Audenaert~\cite[Theorem~1]{Audenaert07}
also gives the ordinary-entropy bound
\begin{equation}
\left|H(U)_\xi-H(U)_\zeta\right|
\leq\varepsilon\log_2|U|+h_2(\varepsilon).
\label{eq:ordinary-continuity-converse}
\end{equation}
Indeed, for $d\coloneqq|U|\geq2$ and normalized trace distance
$t\leq\varepsilon\leq1/2$, Audenaert's bound is
$t\log_2(d-1)+h_2(t)$. Both terms are nondecreasing on this interval,
and $\log_2(d-1)\leq\log_2d$, giving~\eqref{eq:ordinary-continuity-converse}. 
Both continuity estimates depend only on the dimension of $U$;
the conditional-entropy bound is independent of the dimension of $Y$.

For the gross resource amounts specified in Section~\ref{sec:notation}, set
\begin{equation}
\begin{aligned}
\delta_K&\coloneqq2\varepsilon(q_g+e_g)+g(\varepsilon),\\
\delta_{MR}&\coloneqq2\varepsilon(c_g+q_g)+g(\varepsilon),\\
\delta_L&\coloneqq3\varepsilon e_g+h_2(\varepsilon)+g(\varepsilon),\\
\delta_M&\coloneqq\varepsilon c_g+h_2(\varepsilon).
\end{aligned}
\label{eq:error-terms}
\end{equation}
The first three consequences of~\eqref{eq:converse-error} are
\begin{align}
q_g+e_g&\leq I(K\rangle ZM)_\omega+\delta_K,
\label{eq:K-bound}\\
c_g+2q_g&\leq I(MR;Z)_\omega+\delta_{MR},
\label{eq:MR-bound}\\
c_g&\leq I(M;Z)_\omega+\delta_M.
\label{eq:M-bound}
\end{align}
We prove each in turn. By the ideal-state definition~\eqref{eq:ideal}, the
registers $R\widehat Q$ and $L\widehat L$ are independent maximally entangled
pairs. Since $K=RL$, continuity, data processing through Bob's decoder, and
strong subadditivity give
\begin{align}
q_g+e_g
&=-H(K|\widehat Q\widehat L)_\Omega\\
&\leq-H(K|\widehat Q\widehat L)_{\widetilde\omega}+\delta_K\\
&\leq-H(K|Z)_\omega+\delta_K\\
&\leq-H(K|ZM)_\omega+\delta_K\\
&=I(K\rangle ZM)_\omega+\delta_K.
\end{align}
This proves~\eqref{eq:K-bound}. For~\eqref{eq:MR-bound}, the ideal state
satisfies
\begin{equation}
H(MR)_\Omega=c_g+q_g,\qquad
H(MR|\widehat M\widehat Q)_\Omega=-q_g.
\end{equation}
Moreover, the $MR$ marginals of $\omega$, $\widetilde\omega$, and $\Omega$
are identical by~\eqref{eq:nosignal}. Thus only the conditional entropy needs
a continuity estimate, and
\begin{align}
c_g+2q_g
&=I(MR;\widehat M\widehat Q)_\Omega\\
&\leq I(MR;\widehat M\widehat Q)_{\widetilde\omega}+\delta_{MR}\\
&\leq I(MR;Z)_\omega+\delta_{MR}.
\end{align}
Finally,~\eqref{eq:converse-error} bounds the probability that
$M\ne\widehat M$ by $\varepsilon$. Fano's inequality gives
\begin{equation}
H(M|\widehat M)_{\widetilde\omega}
\leq\varepsilon c_g+h_2(\varepsilon)=\delta_M.
\end{equation}
Together with $H(M)_{\widetilde\omega}=c_g$ and data processing, this proves
\begin{equation}
c_g\leq I(M;\widehat M)_{\widetilde\omega}+\delta_M
\leq I(M;Z)_\omega+\delta_M,
\end{equation}
which is~\eqref{eq:M-bound}.

The fourth consequence controls a conditional mutual information involving
the entire decoding workspace:
\begin{equation}
I(L;E|MRZ)_\omega\leq\delta_L.
\label{eq:residual}
\end{equation}
Set $Y\coloneqq MR\widehat M\widehat QG$ in the decoded state
\eqref{eq:decoded-converse-state}. Isometric invariance, the chain rule, and
strong subadditivity imply
\begin{align}
I(L;E|MRZ)_\omega
&=I(L;E|\widehat L Y)_{\widetilde\omega}\\
&\leq I(L;EY|\widehat L)_{\widetilde\omega}\\
&=I(L;EY\widehat L)_{\widetilde\omega}
  -I(L;\widehat L)_{\widetilde\omega}\\
&\leq2H(L)_{\widetilde\omega}
  -I(L;\widehat L)_{\widetilde\omega}\\
&\leq2e_g-I(L;\widehat L)_{\widetilde\omega}.
\label{eq:residual-chain}
\end{align}
The fourth line uses $I(U;Y)_\nu\leq2H(U)_\nu$ for every state $\nu_{UY}$,
which follows from the Araki--Lieb inequality. The last line uses
$\log_2|L|=e_g$. To lower-bound the remaining mutual information, apply
\eqref{eq:ordinary-continuity-converse} and~\eqref{eq:AFW} to the marginal on
$L\widehat L$, whose ideal state is $\Phi_{L\widehat L}$:
\begin{align}
H(L)_{\widetilde\omega}
&\geq e_g-\varepsilon e_g-h_2(\varepsilon),\\
H(L|\widehat L)_{\widetilde\omega}
&\leq-e_g+2\varepsilon e_g+g(\varepsilon).
\end{align}
Subtracting the second inequality from the first gives
\begin{equation}
I(L;\widehat L)_{\widetilde\omega}
\geq2e_g-\delta_L.
\end{equation}
Substitution into~\eqref{eq:residual-chain} proves~\eqref{eq:residual}.
No dimension bound on $E,F$, or $G$ has been used.

\begin{theorem}[Finite-block simultaneous converse]\label{thm:finite-converse}
For every code described in Section~\ref{sec:notation} with normalized trace-distance
error at most $\varepsilon\in[0,1/2]$, the single induced instrument in
\eqref{eq:converse-registers} gives the predecoder state $\omega_{XSFBE}$ described above, satisfying
\begin{align}
c_g-c_c+2(q_g-q_c)
&\leq-I(SX;E|B)_\omega+\delta_L+\delta_{MR},\\
q_g-q_c+e_g-e_c
&\leq I(S\rangle BX)_\omega+\delta_K,\\
c_g-c_c+q_g-q_c+e_g-e_c
&\leq I(S\rangle BX)_\omega-I(X;E|B)_\omega+\delta_M+\delta_K.
\end{align}
The error terms are those of~\eqref{eq:error-terms}. The conclusion places
no sign restriction on the net rates in~\eqref{eq:net}.
\end{theorem}
We prove the three inequalities below, using the same state $\omega$ in each.

\subsection{The first facet: classical communication plus twice quantum communication}
The instrument identifications~\eqref{eq:converse-registers}, the chain rule,
and the exact no-signalling identities in~\eqref{eq:nosignal-consequences} give
\begin{align}
I(SX;E|B)_\omega
&=I(MRLWT_BV;E|B)_\omega\\
&=I(MRT_B;E|B)_\omega+I(VW;E|BMRT_B)_\omega+I(L;E|MRZ)_\omega\\
&\leq I(VW;E|BMRT_B)_\omega+\delta_L.
\label{eq:ell-split}
\end{align}
The last step also uses~\eqref{eq:residual}. Similarly,
\begin{equation}
I(MR;Z)_\omega=I(MR;VW|BT_B)_\omega,
\end{equation}
because $Z=B T_B W V$ and $I(MR;BT_B)_\omega=0$.
Adding this equality to~\eqref{eq:ell-split} completes a chain-rule expansion:
\begin{align}
I(SX;E|B)_\omega+I(MR;Z)_\omega&\leq I(VW;EMR|BT_B)_\omega+\delta_L\\
&=I(V;EMR|BT_B)_\omega
       +I(W;EMR|BT_BV)_\omega+\delta_L\\
&\leq H(V)_\omega+2\log_2|W|+\delta_L\\
&\leq c_c+2q_c+\delta_L.
\label{eq:first-joint-bound}
\end{align}
For clarity, the classical-register bound used here follows from
\begin{equation}
I(V;Y|U)_\nu
=H(V|U)_\nu-H(V|UY)_\nu
\leq H(V|U)_\nu\leq H(V)_\nu
\end{equation}
for every state $\nu_{VUY}$ that is classical on $V$; its conditional entropy
$H(V|UY)_\nu$ is nonnegative. The quantum-register bound follows from
$-\log_2|W|\leq H(W|T)_\nu\leq\log_2|W|$ for every state $\nu_{WT}$, applied
to the two conditional entropies in the conditional mutual information.
Finally, the resource constraints imply $H(V)_\omega\leq\log_2|V|\leq c_c$ and
$\log_2|W|\leq q_c$.

Combining~\eqref{eq:first-joint-bound} with~\eqref{eq:MR-bound} gives
\begin{equation}
c_g-c_c+2(q_g-q_c)
\leq-I(SX;E|B)_\omega+\delta_L+\delta_{MR}.
\label{eq:converse-first}
\end{equation}
The consumed entanglement amount $e_c$ does not appear in this bound. The
cancellation is justified by including $T_B$ in the exact no-signalling
identity in~\eqref{eq:nosignal}.

\subsection{The second facet: quantum communication plus entanglement}
The identifications $Z=BDV$, $X=MV$, and $S=KD$ in
\eqref{eq:converse-registers} give the coherent-information chain rule
\begin{align}
I(K\rangle ZM)_\omega
&=-H(K|DBX)_\omega\\
&=-H(KD|BX)_\omega+H(D|BX)_\omega\\
&=I(S\rangle BX)_\omega+H(D|BX)_\omega.
\label{eq:coherent-converse-chain}
\end{align}
Furthermore, subadditivity and the input dimensions imply
\begin{equation}
H(D|BX)_\omega\leq H(D)_\omega
\leq\log_2|W|+\log_2|T_B|\leq q_c+e_c.
\label{eq:quantum-entanglement-bound}
\end{equation}
Substituting these two expressions into~\eqref{eq:K-bound} proves
\begin{equation}
q_g-q_c+e_g-e_c\leq I(S\rangle BX)_\omega+\delta_K.
\label{eq:converse-second}
\end{equation}

\subsection{The third facet: the sum of all three rates}
First, the classical penalty obeys
\begin{align}
I(X;E|B)_\omega
&=I(MV;E|B)_\omega\\
&=I(M;E|B)_\omega+I(V;E|BM)_\omega\\
&=I(V;E|BM)_\omega\\
&\leq H(V|BM)_\omega.
\label{eq:classical-environment-bound}
\end{align}
The third line follows from~\eqref{eq:nosignal}, and the final line follows
from the nonnegative conditional entropy of the classical register $V$.
The same no-signalling identity gives $I(M;B)_\omega=0$. Consequently,
\begin{align}
I(M;Z)_\omega
&=I(M;V|B)_\omega+I(M;D|BV)_\omega\\
&=H(V|B)_\omega-H(V|BM)_\omega+H(D|BV)_\omega-H(D|BMV)_\omega.
\end{align}
Since $BMV=BX$, adding $H(D|BX)_\omega+I(X;E|B)_\omega$ and
using~\eqref{eq:classical-environment-bound} yields the joint bound
\begin{align}
I(M;Z)_\omega+H(D|BX)_\omega+I(X;E|B)_\omega&\leq H(V|B)_\omega+H(D|BV)_\omega\\
&\qquad\leq c_c+q_c+e_c.
\label{eq:joint-bound}
\end{align}
The last step uses $H(V|B)_\omega\leq\log_2|V|\leq c_c$ and
$H(D|BV)_\omega\leq\log_2|W|+\log_2|T_B|\leq q_c+e_c$.

We can now combine the two decoding estimates~\eqref{eq:M-bound}
and~\eqref{eq:K-bound}, the coherent-information identity in~\eqref{eq:coherent-converse-chain}, and~\eqref{eq:joint-bound}:
\begin{align}
c_g+q_g+e_g
&\leq I(M;Z)_\omega+I(K\rangle ZM)_\omega+\delta_M+\delta_K\\
&=I(S\rangle BX)_\omega+I(M;Z)_\omega+H(D|BX)_\omega+\delta_M+\delta_K\\
&\leq I(S\rangle BX)_\omega-I(X;E|B)_\omega+c_c+q_c+e_c+\delta_M+\delta_K.
\end{align}
Rearranging gives
\begin{equation}
c_g-c_c+q_g-q_c+e_g-e_c
\leq I(S\rangle BX)_\omega-I(X;E|B)_\omega+\delta_M+\delta_K.
\label{eq:converse-third}
\end{equation}
This completes the proof of Theorem~\ref{thm:finite-converse}.

\subsection{Completion of the converse}
The inequalities in~\eqref{eq:converse-first},~\eqref{eq:converse-second},
and~\eqref{eq:converse-third} concern one instrument on $A^n$, with its state
$\omega$ of the form~\eqref{eq:instrument-state}. Divide them by $n$ and use
the net-rate definitions~\eqref{eq:net}. Because the six gross resource amounts
are $O(n)$ and $\varepsilon\to0$, every term in~\eqref{eq:error-terms}
divided by $n$ tends to zero. Let
\begin{equation}
\eta_n\coloneqq\frac1n\max\{\delta_L+\delta_{MR},\,
\delta_K,\,\delta_M+\delta_K\}.
\end{equation}
Subtracting $\eta_n$ from each coordinate of $(C,Q,E)$ reduces the three
left sides in~\eqref{eq:facet1}--\eqref{eq:facet3} by $3\eta_n$,
$2\eta_n$, and $3\eta_n$, respectively. Thus
\begin{equation}
(C-\eta_n,Q-\eta_n,E-\eta_n)
\in\frac1n\cR^{(1)}(\rho_{AB}^{\otimes n}).
\end{equation}
Since $\eta_n\to0$, taking the closure proves the converse
in~\eqref{eq:regularized-region} under the resource-accounting convention of Section~\ref{sec:notation}.

Every code under the conventional resource-accounting definition of
Ref.~\cite{HW10} is also a code under the definition adopted here.
Conversely, the preceding converse bounds the region defined here by the
regularized information region in~\eqref{eq:regularized-region}. The
original direct coding theorem~\cite[Theorem~2 and Section~VI-A]{HW10}
establishes achievability of that region under the conventional
definition. The two operational definitions therefore give the same
capacity region. In particular, simultaneous consumption and generation
with all six resource amounts $O(n)$ does not enlarge the region.

The direct construction in Section~\ref{sec:theorem} and
Appendix~\ref{app:direct} proves achievability under our definition from
instrument compression and state redistribution. The comparison with
conventional accounting uses the earlier direct theorem, without
assuming cancellation of approximate communication resources.

\section{The static capacity formula}\label{sec:formula}
Theorem~\ref{thm:static} characterizes the static capacity region through
an optimization over local instruments acting on arbitrarily many copies
of the shared state. For each instrument, the three
inequalities in~\eqref{eq:facet1}--\eqref{eq:facet3} specify a region of
achievable rate triples. Although this characterization is operationally
complete, evaluating the union of these regions can be difficult.
A useful alternative is to determine the supremum of each weighted
combination of the rates.

In this section, we express these suprema as entropy optimizations and
show how they characterize the entire capacity region. We first take
nonnegative linear combinations of the three rate inequalities and
identify the weights for which the resulting optimization is finite.
After a normalization, two parameters suffice to describe these weights,
apart from the direction corresponding to $C+2Q$. Regularizing the
entropy optimizations then gives the necessary and sufficient
inequalities stated in Theorem~\ref{thm:formula} below. The derivation uses linear-programming duality, following
Section~25.4.2 of Ref.~\cite{Wilde17}.

This formulation also provides a way to establish single-letter
capacity formulas. If the relevant entropy optimizations are additive
on tensor powers of the shared state, their regularization can be
removed. The final part of this section states this criterion
precisely, preparing for the analysis of particular state families
in the sections that follow.

\subsection{Derivation of the static capacity formula}
For a nonempty rate region $\mathcal R\subseteq\mathbb R^3$ and a real
weight vector $w=(w_C,w_Q,w_E)$, its \emph{support function} is defined as
\begin{equation}
h_{\mathcal R}(w)\coloneqq
\sup_{(C,Q,E)\in\mathcal R}(w_C C+w_Q Q+w_E E).
\label{eq:support-function-definition}
\end{equation}
This value, which may be infinite, is the least upper bound on the
weighted sum of rates. Whenever it is finite, the corresponding
inequality defines a half-space containing $\mathcal R$. A closed convex
region is the intersection of all these half-spaces. We now express the
support function of $\cR^{(1)}(\rho)$ as an entropy optimization; its
regularization will give the support function of $\cC_{\rm stat}(\rho)$.

Fix an instrument $\cT$ and its state $\sigma_{XSFBE}$ from
\eqref{eq:instrument-state}. Let $\mathcal P_{\cT}$ denote the polyhedron
specified by the three inequalities in~\eqref{eq:facet1}--\eqref{eq:facet3}.
For a weight vector $w\coloneqq(w_C,w_Q,w_E)\in\mathbb R^3$, computing its
support function is the linear program
\begin{equation}
\begin{aligned}
P_{\cT}(w)&\coloneqq\sup_{C,Q,E\in\mathbb R}(w_C C+w_Q Q+w_E E)\\
\text{subject to}\qquad C+2Q&\leq-I(SX;E|B)_\sigma,\\
Q+E&\leq I(S\rangle BX)_\sigma,\\
C+Q+E&\leq I(S\rangle BX)_\sigma-I(X;E|B)_\sigma.
\end{aligned}
\label{eq:static-primal}
\end{equation}
The rates are unrestricted real variables because they are net rates.
The instrument is fixed throughout this linear program; the optimization
over instruments will be taken afterward. We follow the duality argument
in the proof of Theorem~25.4.1 of Ref.~\cite{Wilde17}.

Associate nonnegative multipliers $a_1,a_2,a_3$ with the three constraints
in~\eqref{eq:static-primal}, respectively, and set
$a\coloneqq(a_1,a_2,a_3)$. Adding each multiplier times the corresponding
right-hand side minus left-hand side to the objective gives the Lagrangian
\begin{equation}
\begin{aligned}
\mathcal L_{\cT}(w;C,Q,E;a)\coloneqq{}&w_C C+w_Q Q+w_E E
+a_1[-I(SX;E|B)_\sigma-C-2Q]\\
&+a_2[I(S\rangle BX)_\sigma-Q-E]\\
&+a_3[I(S\rangle BX)_\sigma-I(X;E|B)_\sigma-C-Q-E].
\end{aligned}
\label{eq:static-lagrangian}
\end{equation}
For a feasible rate triple, every added term is nonnegative. Thus the
supremum of the Lagrangian over all real $C,Q,E$ is an upper bound on
$P_{\cT}(w)$. The coefficients of these variables in the Lagrangian are,
respectively, $w_C-a_1-a_3$, $w_Q-2a_1-a_2-a_3$, and $w_E-a_2-a_3$.
If one is nonzero, the unrestricted supremum is infinite. Otherwise all
three rate variables disappear. The dual linear program is consequently
\begin{equation}
\begin{aligned}
D_{\cT}^{\rm dual}(w)\coloneqq
\inf_{a_1,a_2,a_3\geq0}\ \big\{&-a_1 I(SX;E|B)_\sigma
+(a_2+a_3)I(S\rangle BX)_\sigma-a_3 I(X;E|B)_\sigma\big\}\\
\text{subject to}\quad
&a_1+a_3=w_C,\quad 2a_1+a_2+a_3=w_Q,\quad a_2+a_3=w_E.
\end{aligned}
\label{eq:static-dual}
\end{equation}
An infimum over an empty feasible set is understood as $+\infty$.

For comparison with the geometry of the rate region, define its facet normals
\begin{equation}
n_1\coloneqq(1,2,0),\qquad n_2\coloneqq(0,1,1),\qquad n_3\coloneqq(1,1,1).
\label{eq:static-facet-normals}
\end{equation}
The dual equality constraints in~\eqref{eq:static-dual} say precisely that
\begin{equation}
w=a_1n_1+a_2n_2+a_3n_3
=(a_1+a_3,\;2a_1+a_2+a_3,\;a_2+a_3).
\label{eq:polar}
\end{equation}
This system has the unique solution
\begin{align}
a_1&=\tfrac12(w_Q-w_E)=-\tfrac12w\cdot u_{\rm ED},\label{eq:dual-a1}\\
a_2&=\tfrac12(-2w_C+w_Q+w_E)=-\tfrac12w\cdot u_{\rm SD},\label{eq:dual-a2}\\
a_3&=\tfrac12(2w_C-w_Q+w_E)=-\tfrac12w\cdot u_{\rm TP}.
\label{eq:dual-a3}
\end{align}
Thus the dual is feasible exactly when these three numbers are nonnegative.

The primal is always feasible: the classically assisted state-redistribution
point $v_{\cT}\coloneqq(C_0,Q_0,E_0)$ from
\eqref{eq:CASR-C}--\eqref{eq:CASR-E} simultaneously saturates all three
constraints in~\eqref{eq:static-primal}.
Suppose that the dual is feasible, so that the uniquely determined
multipliers in~\eqref{eq:dual-a1}--\eqref{eq:dual-a3} are nonnegative.
The dual equality constraints and saturation of the primal constraints
at $v_{\cT}$ give
\begin{equation}
w\cdot v_{\cT}
=-a_1 I(SX;E|B)_\sigma
+(a_2+a_3)I(S\rangle BX)_\sigma
-a_3 I(X;E|B)_\sigma
=D_{\cT}^{\rm dual}(w).
\label{eq:static-duality-attainment}
\end{equation}
We can therefore verify strong duality directly:
\begin{equation}
w\cdot v_{\cT}
\leq P_{\cT}(w)
\leq D_{\cT}^{\rm dual}(w)
=w\cdot v_{\cT}.
\end{equation}
The first inequality follows from primal feasibility of $v_{\cT}$,
the second follows from weak duality, and the final equality follows
from~\eqref{eq:static-duality-attainment}.
Since the endpoints coincide, all inequalities are equalities.
Thus
$P_{\cT}(w)=D_{\cT}^{\rm dual}(w)$,
with the primal optimum attained at $v_{\cT}$ and the dual optimum
attained at the multipliers in~\eqref{eq:dual-a1}--\eqref{eq:dual-a3}.

If instead one of the numbers in~\eqref{eq:dual-a1}--\eqref{eq:dual-a3}
is negative, the corresponding unit-resource vector has a positive inner
product with $w$. Adding arbitrarily large nonnegative multiples of this
vector to $v_{\cT}$ preserves feasibility and makes the objective unbounded.
Geometrically, $\mathcal P_{\cT}=v_{\cT}+\cU$, with $\cU$ defined in
\eqref{eq:unit}; finite support requires $w\cdot u\leq0$ for every
$u\in\cU$. Thus the dual-feasibility conditions identify exactly the
directions with finite support.

Taking the supremum over instruments now motivates the definition
\begin{equation}
D_{a_1,a_2,a_3}(\rho)\coloneqq\sup_{\cT}
\big\{-a_1 I(SX;E|B)_\sigma+(a_2+a_3)I(S\rangle BX)_\sigma
-a_3 I(X;E|B)_\sigma\big\},\qquad a_1,a_2,a_3\geq0.
\label{eq:unnormalized}
\end{equation}
For dual-feasible weights $w$, the multipliers in
\eqref{eq:dual-a1}--\eqref{eq:dual-a3} depend only on $w$, and the
fixed-instrument duality calculation gives
\begin{equation}
\sup_{(C,Q,E)\in\cR^{(1)}(\rho)}(w_C C+w_Q Q+w_E E)
=\sup_{\cT}P_{\cT}(w)=D_{a_1,a_2,a_3}(\rho).
\label{eq:unnormalized-support-interpretation}
\end{equation}
No interchange between an instrument optimization and a dual minimization
is needed: the rate linear program was solved for each instrument first.

This supremum is finite even though no uniform auxiliary dimension has
been imposed. Indeed, $I(S\rangle BX)_\sigma\leq H(B)_\rho$, both
conditional mutual informations in~\eqref{eq:unnormalized} are nonnegative,
and the trivial instrument gives zero. Hence
\begin{equation}
0\leq D_{a_1,a_2,a_3}(\rho)\leq(a_2+a_3)H(B)_\rho.
\label{eq:unnormalized-bound}
\end{equation}

If $s\coloneqq a_2+a_3>0$, divide by $s$ and set
\begin{equation}
\mu\coloneqq a_1/s,\qquad
\lambda\coloneqq(a_1+a_3)/s.
\end{equation}
The normalized dual multipliers are $a_1/s=\mu$,
$a_2/s=1+\mu-\lambda$, and $a_3/s=\lambda-\mu$. Their
nonnegativity is equivalent to
\begin{equation}
\mu\geq0,\qquad\mu\leq\lambda\leq\mu+1,
\label{eq:parameters}
\end{equation}
and $w/s=(\lambda,1+2\mu,1)$. The chain rule
\begin{equation}
I(SX;E|B)_\sigma=I(X;E|B)_\sigma+I(S;E|BX)_\sigma
\end{equation}
then gives the quantum static capacity formula
\begin{equation}
\sS_{\lambda,\mu}(\rho)\coloneqq\sup_{\cT}
\big\{I(S\rangle BX)_\sigma-\lambda I(X;E|B)_\sigma
-\mu I(S;E|BX)_\sigma\big\}.
\label{eq:support}
\end{equation}
Thus, for the weights in~\eqref{eq:parameters},
\begin{equation}
\sS_{\lambda,\mu}(\rho)
=h_{\cR^{(1)}(\rho)}(\lambda,1+2\mu,1).
\label{eq:onecopy-support-interpretation}
\end{equation}
The static capacity formula therefore evaluates the support function of
the one-copy information region in these normalized directions.
When $s=0$, only a nonnegative multiple of $n_1$ remains. Its support is zero,
since $I(SX;E|B)_{\sigma}\geq0$ and the trivial instrument has $I(SX;E|B)_{\sigma}=0$.

\subsection{Regularization and the exact half-space representation}
The following inequality holds for all bipartite states $\rho$ and $\tau$:
\begin{equation}
\sS_{\lambda,\mu}(\rho\otimes\tau)
\geq\sS_{\lambda,\mu}(\rho)+\sS_{\lambda,\mu}(\tau),
\label{eq:superadd}
\end{equation}
by setting the instrument in the optimization on the left-hand side to be a tensor product of instruments.
The trivial instrument gives zero. Since $X$ is classical,
$H(SB|X)_\sigma\geq0$, and thus
$I(S\rangle BX)_{\sigma}\leq H(B|X)_{\sigma}\leq H(B)_{\sigma}$.
The two conditional mutual information penalties are nonnegative, so
\begin{equation}
0\leq\sS_{\lambda,\mu}(\rho)\leq H(B)_\rho.
\end{equation}

Fix $\mu\geq0$ and $\mu\leq\lambda\leq\mu+1$.
Applying the preceding bound to $\rho_{AB}^{\otimes n}$ gives,
for all $n\in\mathbb{N}$,
\begin{equation}
0\leq \sS_{\lambda,\mu}(\rho^{\otimes n})
\leq H(B^n)_{\rho^{\otimes n}}
=nH(B)_\rho.
\end{equation}
The last equality follows from additivity of entropy on product states.
Thus the normalized quantities
$\sS_{\lambda,\mu}(\rho^{\otimes n})/n$ are uniformly bounded.
Together with superadditivity in~\eqref{eq:superadd},
this bound allows us to apply Fekete's lemma, giving the finite limit
\begin{equation}
\sS_{\lambda,\mu}^{\reg}(\rho)
\coloneqq
\lim_{n\to\infty}\frac1n
\sS_{\lambda,\mu}(\rho^{\otimes n})
=
\sup_{n\in\mathbb{N}}\frac1n
\sS_{\lambda,\mu}(\rho^{\otimes n}).
\label{eq:reg-support}
\end{equation}

To see this in detail, fix $k\in\mathbb{N}$.
For $n\geq k$, write $n=mk+r$, where
$m\coloneqq\lfloor n/k\rfloor$ and $0\leq r<k$.
Repeated superadditivity, together with nonnegativity of the
remainder contribution when $r>0$, gives
\begin{equation}
\frac1n\sS_{\lambda,\mu}(\rho^{\otimes n})
\geq
\frac{m}{n}\sS_{\lambda,\mu}(\rho^{\otimes k}).
\end{equation}
When $r=0$, the same inequality follows directly from
superadditivity, without a remainder term.
Since $m/n\to1/k$ as $n\to\infty$, it follows that
\begin{equation}
\liminf_{n\to\infty}\frac1n
\sS_{\lambda,\mu}(\rho^{\otimes n})
\geq
\frac1k\sS_{\lambda,\mu}(\rho^{\otimes k}).
\end{equation}
This holds for every $k\in\mathbb{N}$, so
\begin{equation}
\liminf_{n\to\infty}\frac1n
\sS_{\lambda,\mu}(\rho^{\otimes n})
\geq
\sup_{k\in\mathbb{N}}\frac1k
\sS_{\lambda,\mu}(\rho^{\otimes k}).
\end{equation}
Conversely, every term of the normalized sequence is at most
this supremum, so its limit superior is also at most the
supremum. The limit therefore exists and equals the finite
supremum in~\eqref{eq:reg-support}.

The regularized information region is convex. Indeed, time-sharing codes on
separate blocks realizes rational mixtures of their rate triples, and limits
realize arbitrary mixtures. It is closed by definition and is invariant under
addition of $\cU$. Its support is the supremum of the supports of its block
regions; taking a closure does not change the supremum of a continuous linear
functional. In the notation of~\eqref{eq:support-function-definition},
\begin{equation}
h_{\cC_{\rm stat}(\rho)}(\lambda,1+2\mu,1)
=\sS_{\lambda,\mu}^{\reg}(\rho).
\label{eq:operational-support-interpretation}
\end{equation}
The regularized static capacity formula is therefore the support function
of the operational region in these directions. We have thus proved the
following characterization.
\begin{theorem}[Static capacity formula]\label{thm:formula}
For every bipartite state $\rho_{AB}$, its static capacity region
$\cC_{\rm stat}(\rho)$ from Section~\ref{sec:notation} consists
exactly of the triples $(C,Q,E)\in\mathbb R^3$ satisfying
\begin{equation}
C+2Q\leq0,\qquad
\lambda C+(1+2\mu)Q+E\leq\sS_{\lambda,\mu}^{\reg}(\rho)
\quad(\mu\geq0,\ \mu\leq\lambda\leq\mu+1).
\label{eq:halfspace}
\end{equation}
The second inequality is required for every indicated pair $(\lambda,\mu)$,
and $\sS_{\lambda,\mu}^{\reg}$ is defined in~\eqref{eq:reg-support}.
\end{theorem}
For the converse implication in this geometric statement, a point outside a
closed convex set can be strictly separated from it. The support function is finite in the separating direction, so this
direction has the decomposition in~\eqref{eq:polar} with nonnegative
coefficients. It is either
$n_1$ up to scaling or one of the normalized directions above. The corresponding
inequality in~\eqref{eq:halfspace} excludes the point. This proves sufficiency
as well as necessity of the displayed family.

\subsection{The precise single-letterization criterion}\label{sec:single-letterization}

If, for every admissible $\lambda,\mu$ and all $n\in\mathbb{N}$,
\begin{equation}
\sS_{\lambda,\mu}(\rho^{\otimes n})=n\sS_{\lambda,\mu}(\rho),
\label{eq:selfadd}
\end{equation}
then~\eqref{eq:halfspace} has a one-copy right side. Moreover,
$\cR^{(1)}(\rho)$ is convex: choose an independent classical coin to select
between instruments and include it in $X$. Its marginal is independent of $BE$,
so its additional $I(X;E|B)_{\sigma}$ contribution is zero and all three facet values
are the corresponding averages. The closure of this one-copy region has the
same support as the operational region under~\eqref{eq:selfadd}; hence
\begin{equation}
\cC_{\rm stat}(\rho)=\overline{\cR^{(1)}(\rho)}.
\label{eq:single-letter-region}
\end{equation}
Strong additivity with every second state implies~\eqref{eq:selfadd};
the argument here only requires self-additivity.

\section{State families with an exact capacity region}\label{sec:exact-families}

In this section, we evaluate the static capacity formula of
Theorem~\ref{thm:formula} for three classes of bipartite states:
symmetrically extendible states, erased states, and locally flagged
pure-state mixtures. For each family, the optimization can be
controlled at every blocklength, giving a single-letter
characterization of the entire capacity region. The resulting
regions have simple descriptions in terms of the unit-resource
cone~\eqref{eq:unit} and the additional resources obtainable from
the shared state. Each characterization follows by combining
achievable protocols with converse bounds that apply to arbitrary
collective instruments.

\subsection{Symmetrically extendible states}
We first consider states that admit a second copy of Bob's marginal
correlations with Alice. These are precisely the antidegradable states
of Ref.~\cite[Definition~2 and Lemma~30]{LDS18}.
More explicitly, if $\psi_{ABE}$ purifies $\rho_{AB}$, antidegradability
means that a channel $\mathcal A\colon E\to B'$ satisfies
$(\id_A\otimes\mathcal A)(\psi_{AE})=\rho_{AB'}$, where $B'\simeq B$.
Applying this channel to $E$ gives an extension with equal $AB$ and
$AB'$ marginals. Conversely, every such extension is obtained from
$\psi_{ABE}$ by a channel on the purifying system $E$: purify the
extension and use the isometric equivalence of purifications, then
discard the additional purifying system. Averaging an extension with
its image under the swap of $B$ and $B'$ makes it symmetric.
This equivalence relates the present family to the degradable states
considered in Section~\ref{sec:hadamard}.

\begin{theorem}\label{thm:extendible}
Suppose a bipartite state $\rho_{AB}$ has an extension $\xi_{ABB'}$ with
$\xi_{AB'}=\rho_{AB}$ under an identification $B'\simeq B$.
Then
\begin{equation}
\cC_{\rm stat}(\rho)=\cU.
\end{equation}
Here $\cC_{\rm stat}$ is the region of Section~\ref{sec:notation}, and
$\cU$ is the unit-resource cone in~\eqref{eq:unit}.
\end{theorem}
\begin{proof}
Apply an arbitrary instrument on $A$ to the extension and denote its output by
$\sigma_{XSFBB'}$. Its $SBX$ and $SB'X$ marginals
are equal. Weak monotonicity, averaged over the classical $X$, gives
\begin{equation}
0\leq H(S|BX)_{\sigma}+H(S|B'X)_{\sigma}=2H(S|BX)_{\sigma}.
\end{equation}
Thus $I(S\rangle BX)_{\sigma}\leq0$. The conditional mutual information penalties in \eqref{eq:facet1}--\eqref{eq:facet3} are nonnegative by strong subadditivity, so each information
polyhedron is contained in $\cU$. Tensor powers have the same extension
property. The static theorem proves the upper inclusion, and ignoring the
state proves the lower inclusion.
\end{proof}
Every separable state has such an extension: in a decomposition
$\rho_{AB}=\sum_zq_z\alpha^{z}_{A}\otimes\beta^{z}_{B}$, use
$\sum_zq_z\alpha^{z}_{A}\otimes\beta^{z}_{B}\otimes\beta^{z}_{B'}$.
States produced by antidegradable channels also qualify. If $E\to B'$ simulates
the channel output from the complementary output, applying this map to a
purification produces the required extension.

\subsection{Erased pure states}\label{sec:erased-states}
The static capacity region of an erased maximally entangled Bell state
was characterized in Ref.~\cite[Section~VIII-A]{HW10}. Here the input
may be an arbitrary pure bipartite state, and we prove the converse directly
from subset-entropy inequalities that apply to every collective
instrument, following the methods of Refs.~\cite{MW23,GHW22}.

Let $\phi_{RA}$ be a pure bipartite state. Set
$\tau_A\coloneqq\phi_A$ and $h\coloneqq H(A)_\phi$. Define the quantum
erasure channel $\cE_p\colon A\to B$ with erasure probability
$p\in[0,1]$ as
\begin{equation}
\cE_p(\xi)\coloneqq(1-p)\xi\oplus p\Tr[\xi]\proj{\perp}.
\label{eq:erasure-channel}
\end{equation}
Here $B$ is the direct sum of a copy of $A$ and an orthogonal erasure
flag $\ket\perp_B$, and $\xi$ is an operator on $A$. Define an isometric
extension $U_p\colon A\to BE$ by its action on a vector $\ket v_A$:
\begin{equation}
U_p\ket{v}_{A}
\coloneqq\sqrt{1-p}\ket{v}_{B}\ket{\perp}_{E}
+\sqrt p\ket{\perp}_{B}\ket{v}_{E}.
\label{eq:erasure-isometry}
\end{equation}
The environment $E$ has the same direct-sum structure as $B$. Define
\begin{equation}
\psi^p_{RBE}\coloneqq(I_R\otimes U_p)\phi_{RA}(I_R\otimes U_p^\dagger),
\qquad\rho^p_{RB}\coloneqq\Tr_E[\psi^p_{RBE}].
\label{eq:erased-state-definition}
\end{equation}
Equivalently,
\begin{equation}
\rho^{p}_{RB}=(1-p)\phi_{RA\cong B}\oplus p\phi_{R}\otimes\proj{\perp}_{B}.
\end{equation}
The one-party output entropies are given by
\begin{align}
H(B)_{\psi^p}&=h_2(p)+(1-p)h,\\
H(E)_{\psi^p}&=h_2(p)+ph.
\end{align}
By isometric invariance, the \emph{joint} entropy is
\begin{equation}
H(BE)_{\psi^p}=H(A)_\phi=h.
\end{equation}
Thus
\begin{equation}
H(E|B)_{\psi^p}=ph-h_2(p).
\label{eq:fixed-conditional}
\end{equation}

\subsubsection{Evaluation of the static capacity formula}\label{sec:erased-support}
\begin{theorem}[Static capacity formula for erased states]
\label{thm:erased-support}
Let $\phi_{RA}$ be a pure bipartite state, let $p\in[0,1]$, and let
$\rho^p_{RB}$ be the corresponding erased state defined
in~\eqref{eq:erased-state-definition}.
Set $h\coloneqq H(A)_\phi$.
For all $n\in\mathbb{N}$ and every pair $(\lambda,\mu)$ satisfying
$\mu\geq0$ and $\mu\leq\lambda\leq\mu+1$, the static capacity formula
defined in~\eqref{eq:support} satisfies
\begin{equation}
\sS_{\lambda,\mu}\bigl((\rho^p)^{\otimes n}\bigr)
=nh\max\{0,1-2p(1+\mu)\}.
\label{eq:erased-support}
\end{equation}
\end{theorem}
\begin{proof}
The proof is given in Appendix~\ref{app:erased-support}. It uses the
subset-entropy inequalities in Lemmas~\ref{lem:subset-concavity}--\ref{lem:subset-interpolation}
to control every collective instrument, including its discarded quantum
system. The identity instrument and the instrument that discards the entire
input attain the two possible support values.
\end{proof}

\subsubsection{Erased-state capacity region}\label{sec:erased-region}
\begin{theorem}[Erased-state static capacity]\label{thm:erased-region}
Let $\rho^p_{RB}$ be the erased state in~\eqref{eq:erased-state-definition},
obtained from a pure bipartite state $\phi_{RA}$ with $p\in[0,1]$, and
set $h\coloneqq H(A)_\phi$. With the forward-only convention of
Section~\ref{sec:notation} and the unit-resource cone $\cU$ in~\eqref{eq:unit},
\begin{equation}
\cC_{\mathrm{stat}}(\rho^p)
=\conv\{0,(0,-ph,(1-p)h)\}+\cU.
\label{eq:erased-region}
\end{equation}
Equivalently, a rate triple belongs to the region if and only if there exists
$t\in[0,1]$ such that
\begin{align}
C+2Q&\leq-2pth,\label{eq:region-t-first}\\
Q+E&\leq(1-2p)th,\\
C+Q+E&\leq(1-2p)th.
\label{eq:region-t}
\end{align}
For $p\geq1/2$ this region is simply $\cU$.
\end{theorem}
\begin{proof}
\noindent\textbf{Achievability.}
Take the identity instrument on $R$ (in possession of Alice), with quantum output $S=R$ and
trivial classical output and discarded system. The rate point in
\eqref{eq:CASR-C}--\eqref{eq:CASR-E} is then
\begin{equation}
\left(0,-\tfrac12 I(R;E)_{\psi^p},\tfrac12 I(R;B)_{\psi^p}\right).
\label{eq:erased-identity-point}
\end{equation}
The corresponding protocol transfers Alice's share to Bob using quantum
communication and generates entanglement.
By \eqref{eq:erasure-isometry} and purity, $\frac12I(R;E)_{\psi^p}=ph$. Similarly,
\begin{align}
\tfrac12 I(R;B)_{\psi^p}
&=\tfrac12[H(R)_{\psi^p}+H(B)_{\psi^p}-H(E)_{\psi^p}]\\
&=\tfrac12[h+h_2(p)+(1-p)h-h_2(p)-ph]\\
&=(1-p)h.
\end{align}
Hence this state-redistribution protocol achieves $(0,-ph,(1-p)h)$. Ignoring the state achieves
zero. Deterministic time sharing achieves their convex hull, and adding the
three unit protocols achieves its sum with $\cU$.

\textbf{Converse through support functions.}
The static capacity theorem (Theorem~\ref{thm:static}) reduces the
operational converse to bounding the regularized information region
in~\eqref{eq:regularized-region}. By \eqref{eq:erased-support}, its support value
in direction $(\lambda,1+2\mu,1)$ is
\begin{equation}
h\max\{0,1-2p(1+\mu)\}.
\end{equation}
The convex hull of the origin and the rate point $(0,-ph,(1-p)h)$ has exactly this support:
\begin{align}
\max_{0\leq t\leq1}
[\lambda\cdot0+(1+2\mu)(-pth)+(1-p)th]
&=\max_{0\leq t\leq1}th[1-2p(1+\mu)]\\
&=h\max\{0,1-2p(1+\mu)\}.
\end{align}
The support value in direction $(1,2,0)$ is zero for both regions. These are all
finite support directions after including nonnegative rescaling and the unit
cone. Thus the two closed convex regions coincide. An algebraic verification of the finite-facet description appears in Appendix~\ref{app:erased-facets}.

Finally, if $p\geq1/2$, the rate point $(0,-ph,(1-p)h)$ itself belongs to $\cU$, since
\begin{align}
0+2(-ph)&=-2ph\leq0,\\
-ph+(1-p)h&=(1-2p)h\leq0.
\end{align}
Therefore its convex hull with zero, added to $\cU$, is just $\cU$.
\end{proof}

The formula in~\eqref{eq:erased-region} has two useful consequences.
For fixed $p$, the capacity region depends on the pure input only through
its entanglement entropy $h$: inputs with different Schmidt coefficients
and local dimensions can therefore give the same region. Moreover,
time sharing between the state-redistribution protocol and ignoring the
resource, followed by the unit protocols, achieves the entire region.
The collective-instrument converse shows that more general processing
of many copies cannot enlarge it.

\subsubsection{Five-facet form}
Suppose $0<p<1/2$ and $h>0$, and set $k\coloneqq1-2p$.
The erased-state static capacity region
$\cC_{\mathrm{stat}}(\rho^p)$ in~\eqref{eq:erased-region}
consists exactly of the triples $(C,Q,E)\in\mathbb{R}^3$ satisfying
\begin{align}
C+2Q&\leq0,\label{eq:five-first}\\
Q+E&\leq kh,\label{eq:five-second}\\
C+Q+E&\leq kh,\label{eq:five-third}\\
kC+2(1-p)Q+2pE&\leq0,\label{eq:five-fourth}\\
C+2(1-p)Q+2pE&\leq0.
\label{eq:five}
\end{align}
Appendix~\ref{app:erased-facets} derives these inequalities by
eliminating the time-sharing parameter $t$ from the characterization
in Theorem~\ref{thm:erased-region}.
For the limiting cases, \eqref{eq:erased-region} gives
$\cC_{\mathrm{stat}}(\rho^p)=\cU$ when $h=0$ and
$\cC_{\mathrm{stat}}(\rho^0)=(0,0,h)+\cU$ when $p=0$.
Figure~\ref{fig:erased-three-dimensional} displays the five facets for
$p=1/4$ and $h=1$. The two vertices are the origin and the achievable
rate point $(0,-ph,(1-p)h)$ in~\eqref{eq:erased-identity-point};
the edge joining them corresponds to time sharing. The coordinate axes
meet at the origin of the region, so the signs and relative magnitudes
of the three resource rates can be read directly.

\begin{figure}[!t]
\centering
\includegraphics[width=\linewidth]{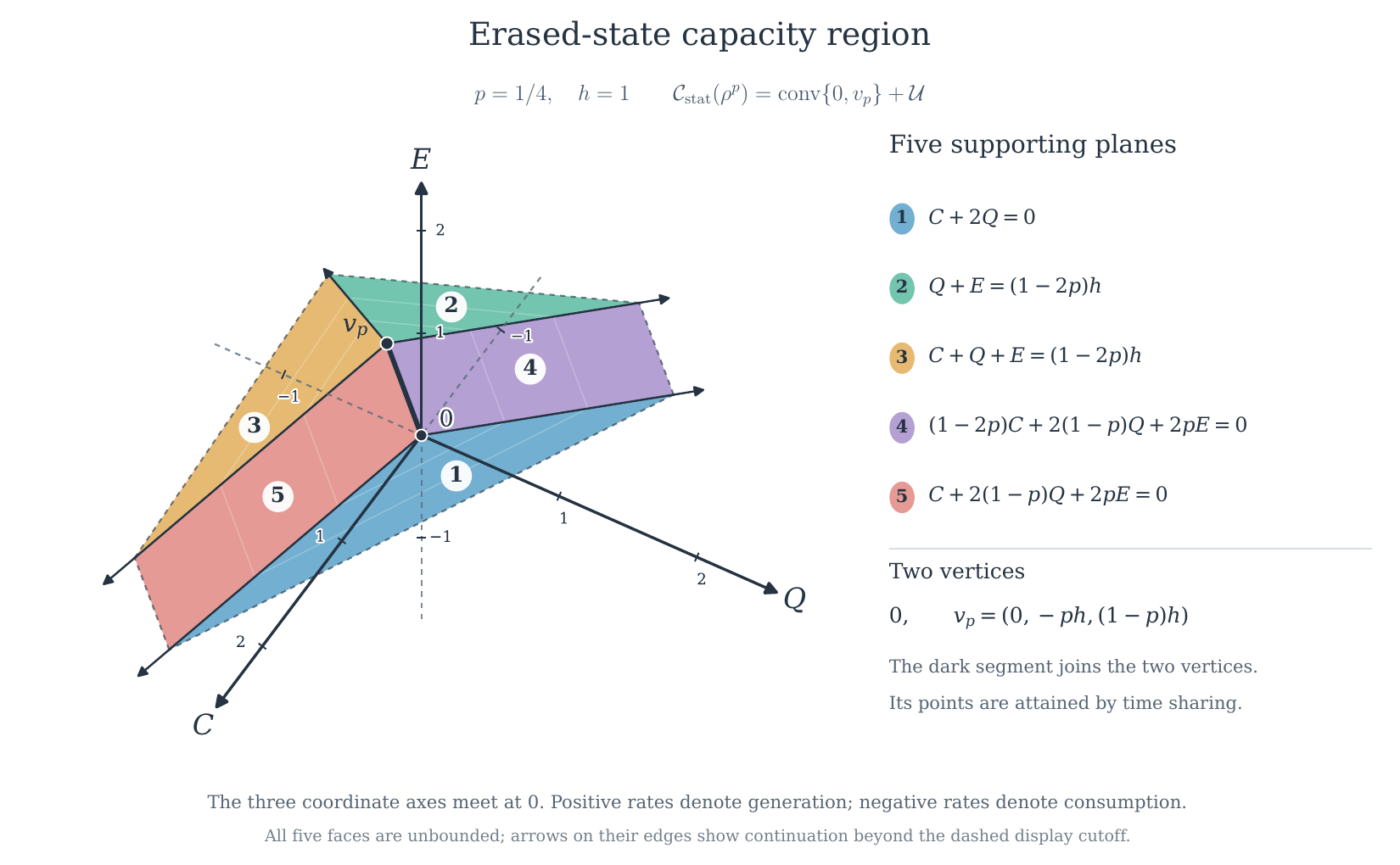}
\caption{The erased-state static capacity region in~\eqref{eq:erased-region},
for $p=1/4$ and $h=1$. The five colored faces lie on the supporting planes
listed at right, corresponding to the five inequalities in~\eqref{eq:five-first}--\eqref{eq:five}.
The two vertices are $0$ and $v_p\coloneqq(0,-ph,(1-p)h)$, and their
connecting edge is attained by time sharing. All five faces are unbounded:
arrows on their edges indicate continuation, and dashed segments mark the
boundary of the displayed portion. All three coordinate axes pass through
the origin. Solid axis arrows indicate positive coordinate directions;
lighter dashed extensions indicate negative directions. Tick values are
in bits, qubits, and ebits per copy for $C$, $Q$, and $E$, respectively.}
\label{fig:erased-three-dimensional}
\end{figure}

\subsubsection{Two capacity slices}\label{sec:erased-slices}
For $0<p<1/2$, set $k\coloneqq1-2p$. With no net classical communication, set
$C=0$ and $Q=-Q_c$, where $Q_c\geq0$ is the supplied quantum communication
rate. The largest achievable entanglement rate is
\begin{equation}
E_{\max}(0,-Q_c)=\min\!\left\{kh+Q_c,\frac{1-p}{p}Q_c\right\}.
\label{eq:erased-Q-slice}
\end{equation}
For $0\leq Q_c\leq ph$, time sharing with the rate point $(0,-ph,(1-p)h)$ attains the boundary.
For $Q_c\geq ph$, this state-redistribution protocol followed by entanglement distribution
attains it.

With no net quantum communication, set $Q=0$ and $C=-C_c$, where $C_c\geq0$
is the supplied classical communication rate.
Then
\begin{equation}
E_{\max}(-C_c,0)=k\min\!\left\{h,\frac{C_c}{2p}\right\}.
\label{eq:erased-C-slice}
\end{equation}
Indeed, the inequalities in~\eqref{eq:five-first}--\eqref{eq:five} imply $E\leq kh$ and $2pE\leq kC_c$.
The last of these inequalities gives $2pE\leq C_c$, which is weaker because $0<k<1$.
Replacing the quantum communication in this protocol by teleportation consumes
classical communication at rate $2ph$ and entanglement at rate $ph$,
giving the net rate point $(C,Q,E)=(-2ph,0,kh)$.
Time sharing with the origin gives the initial line
segment in \eqref{eq:erased-C-slice}.
Figure~\ref{fig:erased-slices} displays both slices.

\begin{figure}[t]
\centering
\includegraphics[width=\linewidth]{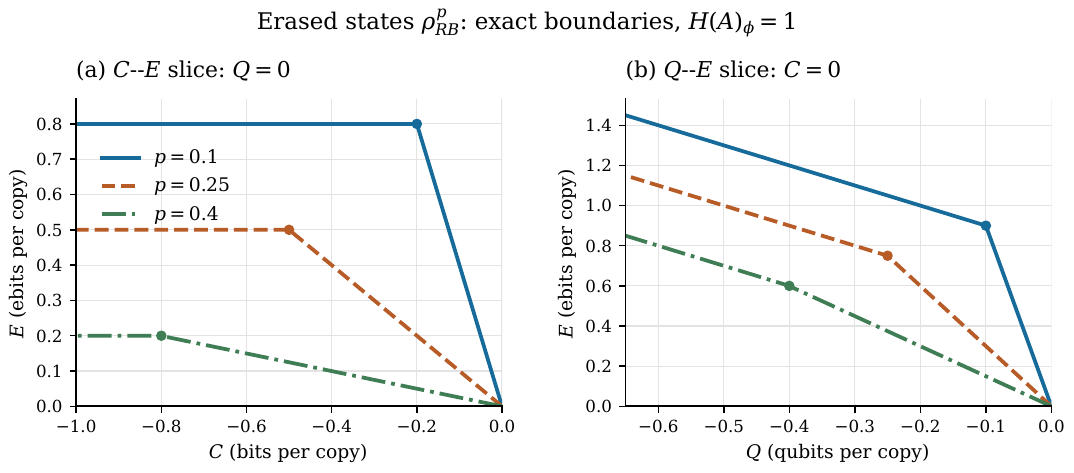}
\caption{Exact slices of the erased-state capacity region in
Theorem~\ref{thm:erased-region}, for $h=H(A)_\phi=1$ and
$p\in\{0.1,0.25,0.4\}$. The boundaries are given by
\eqref{eq:erased-C-slice} and \eqref{eq:erased-Q-slice}.
Markers indicate the rate point $(0,-ph,(1-p)h)$ in (b) and its teleportation conversion
in (a). Negative communication rates represent consumption. Only the
portions with $E\geq0$ are displayed; points below each boundary are
achievable. The increase beyond one ebit in (b) uses additional supplied
quantum communication for entanglement distribution.}
\label{fig:erased-slices}
\end{figure}

\begin{remark}[Why the erasure flag does not provide free feedback]
Bob can locally distinguish successful transmission from erasure. Alice's
instrument must nevertheless be chosen without being told the actual erasure
pattern. Giving her this pattern would change the operational problem. For
example, free backward communication would allow them to keep the surviving
pure states and concentrate entanglement at rate $(1-p)h$. The forward-only
answer above must not be confused with this different task.
\end{remark}

\subsection{Locally flagged pure-state mixtures}\label{sec:flagged-states}
Suppose that Alice and Bob have matching classical flags specifying a
shared pure state. They can each read the flag without communication and
apply a protocol adapted to that state. The following theorem shows that
the complete static region is determined by the average entanglement
entropy. Ordinary pure bipartite states are included by taking a single
flag value.

\begin{theorem}[Locally flagged pure-state mixtures]\label{thm:flagged}
Consider the state
\begin{equation}
\rho_{A_{\mathrm f}A_{\mathrm q}B_{\mathrm f}B_{\mathrm q}}
\coloneqq\sum_zq_z\proj z_{A_{\mathrm f}}\otimes\proj z_{B_{\mathrm f}}
\otimes\phi^z_{A_{\mathrm q}B_{\mathrm q}},
\label{eq:flagged-pure-state}
\end{equation}
where $\left(q_z\right)_z$ is a probability distribution, every $\phi^z_{A_{\mathrm q}B_{\mathrm q}}$
is a pure state, and the flag vectors form orthonormal families on $A_{\mathrm f}$
and $B_{\mathrm f}$. The subscripts $\mathrm f$ and $\mathrm q$ distinguish
the classical flag registers from the systems carrying the conditional pure
state. Set $A\coloneqq A_{\mathrm f}A_{\mathrm q}$, $B\coloneqq B_{\mathrm f}B_{\mathrm q}$, and
\begin{equation}
D\coloneqq\sum_zq_zH(A_{\mathrm q})_{\phi^z}.
\label{eq:flagged-average-entanglement}
\end{equation}
For all $n\in\mathbb{N}$, $\mu\geq0$, and $\mu\leq\lambda\leq\mu+1$,
\begin{equation}
\sS_{\lambda,\mu}(\rho_{AB}^{\otimes n})=nD,
\qquad
\cC_{\rm stat}(\rho_{AB})=(0,0,D)+\cU.
\label{eq:flagged-region}
\end{equation}
The static capacity formula, operational region, and unit-resource cone are defined
in~\eqref{eq:support}, Section~\ref{sec:notation}, and~\eqref{eq:unit},
respectively.
\end{theorem}
\begin{proof}
Fix $n\in\mathbb N$. For a flag sequence
$z^n\coloneqq(z_1,\ldots,z_n)$, define
\begin{equation}
q_{z^n}\coloneqq\prod_{i=1}^nq_{z_i},
\qquad
\phi^{z^n}_{A_{\mathrm q}^nB_{\mathrm q}^n}
\coloneqq\bigotimes_{i=1}^n\phi^{z_i}_{A_{\mathrm q,i}B_{\mathrm q,i}}.
\label{eq:flagged-product-data}
\end{equation}
Here $A_{\mathrm q}^n\coloneqq A_{\mathrm q,1}\cdots A_{\mathrm q,n}$, and similarly for the
other systems. The tensor-power input state therefore has the form
\begin{equation}
\rho_{AB}^{\otimes n}
=\sum_{z^n}q_{z^n}\proj{z^n}_{A_{\mathrm f}^n}
\otimes\proj{z^n}_{B_{\mathrm f}^n}\otimes\phi^{z^n}_{A_{\mathrm q}^nB_{\mathrm q}^n}.
\label{eq:flagged-tensor-power}
\end{equation}
Entropy additivity on each conditional product state in
\eqref{eq:flagged-product-data} gives
\begin{align}
\sum_{z^n}q_{z^n}H(B_{\mathrm q}^n)_{\phi^{z^n}}
&=\sum_{i=1}^n\sum_{z^n}
\left(\prod_{j=1}^nq_{z_j}\right)H(B_{\mathrm q,i})_{\phi^{z_i}}\\
&=\sum_{i=1}^n\sum_{z_i}q_{z_i}H(B_{\mathrm q,i})_{\phi^{z_i}}
=nD.
\label{eq:flagged-product-entanglement}
\end{align}
The second equality uses $\sum_{z_j}q_{z_j}=1$ for $j\ne i$, and the
last equality uses purity of each $\phi^{z_i}$ and the definition of $D$
in~\eqref{eq:flagged-average-entanglement}.

Consider an arbitrary collective instrument on $A^n=A_{\mathrm f}^nA_{\mathrm q}^n$,
with classical output $X$, quantum output $S$, and discarded system $F$.
Let $\sigma_{XSFB^nE^n}$ be its state as in
\eqref{eq:instrument-state}, applied to a tensor-power purification of
$\rho_{AB}^{\otimes n}$. Bob's register $B_{\mathrm f}^n$ remains classical in
its flag basis, so its marginal with $X,S,B_{\mathrm q}^n$ is
\begin{equation}
\sigma_{XSB_{\mathrm f}^nB_{\mathrm q}^n}
=\sum_{x,z^n}p_{x,z^n}\proj x_X\otimes\proj{z^n}_{B_{\mathrm f}^n}
\otimes\sigma^{x,z^n}_{SB_{\mathrm q}^n},
\label{eq:flagged-arbitrary-instrument}
\end{equation}
where $\left(p_{x,z^n}\right)_{x,z^n}$ is a probability distribution and the
conditional states are normalized whenever $p_{x,z^n}>0$.
Zero-probability terms are omitted. Averaging over all instrument
outcomes is a trace-preserving operation on Alice's systems. It leaves
Bob's marginal unchanged separately for every flag sequence, giving
\begin{equation}
\sum_xp_{x,z^n}=q_{z^n},
\qquad
\sum_xp_{x,z^n}\sigma^{x,z^n}_{B_{\mathrm q}^n}
=q_{z^n}\phi^{z^n}_{B_{\mathrm q}^n}.
\label{eq:flagged-fixed-marginals}
\end{equation}
The conditional states $\sigma^{x,z^n}_{SB_{\mathrm q}^n}$ need not be product
states across the $n$ copies.

Since $B^n=B_{\mathrm f}^nB_{\mathrm q}^n$, the entropy formula for conditioning on the
classical registers $B_{\mathrm f}^n,X$, nonnegativity of ordinary entropy, and
concavity of entropy imply
\begin{align}
I(S\rangle B^nX)_\sigma
&=\sum_{x,z^n}p_{x,z^n}
\bigl[H(B_{\mathrm q}^n)_{\sigma^{x,z^n}}
-H(SB_{\mathrm q}^n)_{\sigma^{x,z^n}}\bigr]\\
&\leq\sum_{x,z^n}p_{x,z^n}H(B_{\mathrm q}^n)_{\sigma^{x,z^n}}\\
&\leq\sum_{z^n}q_{z^n}H(B_{\mathrm q}^n)_{\phi^{z^n}}
=nD.
\label{eq:flagged-coherent-bound}
\end{align}
For the second inequality, apply concavity separately for each sequence
with $q_{z^n}>0$, using the normalized probabilities
$p_{x,z^n}/q_{z^n}$ and~\eqref{eq:flagged-fixed-marginals}.
The last equality follows from~\eqref{eq:flagged-product-entanglement}.
Both conditional mutual information penalties in~\eqref{eq:support}
are nonnegative. Since $\lambda,\mu\geq0$, their subtraction can only
decrease the objective. Thus
\begin{equation}
\sS_{\lambda,\mu}(\rho_{AB}^{\otimes n})\leq nD.
\label{eq:flagged-support-upper}
\end{equation}

For the matching lower bound, it suffices to consider a single-copy
instrument. Choose the purification
\begin{equation}
\ket\psi_{ABE}\coloneqq\sum_z\sqrt{q_z}\ket z_{A_{\mathrm f}}\ket z_{B_{\mathrm f}}
\ket{\phi^z}_{A_{\mathrm q}B_{\mathrm q}}\ket z_E.
\label{eq:flagged-purification}
\end{equation}
Alice measures $A_{\mathrm f}$ in its flag basis, records the outcome in $X$,
and retains $S=A_{\mathrm q}$. The discarded system $F$ is trivial.
The resulting state is
\begin{equation}
\widehat\sigma_{XSBE}
\coloneqq\sum_zq_z\proj z_X\otimes\proj z_{B_{\mathrm f}}
\otimes\phi^z_{SB_{\mathrm q}}\otimes\proj z_E.
\label{eq:flagged-instrument-state}
\end{equation}
Here the system $A_{\mathrm q}$ of $\phi^z$ is relabelled as $S$.
Conditioned on $X=z$, the state on $SB_{\mathrm q}$ is pure and the
environment is fixed. Moreover, Bob's flag $B_{\mathrm f}$ already determines
$X$. These observations and the definition of $D$
in~\eqref{eq:flagged-average-entanglement} give
\begin{equation}
I(S\rangle BX)_{\widehat\sigma}=D,
\qquad I(X;E|B)_{\widehat\sigma}=0,
\qquad I(S;E|BX)_{\widehat\sigma}=0.
\label{eq:flagged-attaining-information}
\end{equation}
Thus this instrument attains the value $D$ in~\eqref{eq:support}.
Applying this single-copy instrument independently to each copy
attains the value $nD$ for $\rho_{AB}^{\otimes n}$.
Together with~\eqref{eq:flagged-support-upper}, this proves
$\sS_{\lambda,\mu}(\rho_{AB}^{\otimes n})=nD$
for all $n\in\mathbb N$, and consequently
$\sS_{\lambda,\mu}^{\reg}(\rho_{AB})=D$.
It remains to identify the operational region. By the chain rule
and~\eqref{eq:flagged-attaining-information},
$I(SX;E|B)_{\widehat\sigma}=0$.
Consequently, the inequalities in~\eqref{eq:facet1}--\eqref{eq:facet3}
for this single-copy instrument are precisely
\begin{equation}
C+2Q\leq0,\qquad Q+E\leq D,\qquad C+Q+E\leq D.
\label{eq:flagged-region-facets}
\end{equation}
By~\eqref{eq:unit}, these inequalities describe
$(0,0,D)+\cU$. The direct part of Theorem~\ref{thm:static}
therefore establishes achievability of this entire region.

Conversely, substituting
$\sS_{\lambda,\mu}^{\reg}(\rho_{AB})=D$
into~\eqref{eq:halfspace} and taking
$(\lambda,\mu)=(0,0)$ and $(1,0)$ gives
$Q+E\leq D$ and $C+Q+E\leq D$.
Together with the universal inequality $C+2Q\leq0$,
these are exactly the constraints
in~\eqref{eq:flagged-region-facets}.
This proves the region statement in~\eqref{eq:flagged-region}.
The point $(0,0,D)$ can also be achieved directly by pure-state
entanglement concentration conditioned on the locally available
flags~\cite{Concentration}.
\end{proof}

\section{Hadamard states}\label{sec:hadamard}
This section characterizes states obtained from Hadamard channels and
recalls the exact coherent-information face of their static capacity
region. A decomposition of the remaining entropy optimization into
nonnegative terms then identifies the obstacle to proving a complete
single-letter characterization.

A channel is Hadamard if its complementary channel is entanglement
breaking. This class was introduced in Ref.~\cite{KMNR07}; its dynamic
trade-off capacities were studied in Refs.~\cite{BHTW10,WH12}.
A Hadamard-generated state has the form
\begin{equation}
\rho_{AB}\coloneqq(\id_A\otimes\cN_{A'\to B})(\phi_{AA'}),
\end{equation}
where $\phi_{AA'}$ is pure and $\cN_{A'\to B}$ is Hadamard.
An entanglement-breaking channel admits rank-one Kraus operators. Thus an
isometric dilation of $\cN_{A'\to B}$ can be written, up to an output
isometry, as
\begin{equation}
U_{A'\to BE}
\coloneqq\sum_y\ket y_B\ket{e_y}_E\bra{v_y}_{A'},\qquad
\sum_y\proj{v_y}_{A'}=I_{A'},
\label{eq:had-dilation}
\end{equation}
with every vector $\ket{e_y}_E$ normalized. Fix the purification
\begin{equation}
\psi_{ABE}\coloneqq
(I_A\otimes U_{A'\to BE})\phi_{AA'}
(I_A\otimes U_{A'\to BE})^\dagger.
\label{eq:had-purification}
\end{equation}
Measuring $B$ in the basis displayed in~\eqref{eq:had-dilation} and preparing
$\proj{e_y}_E$ simulates the complementary channel. Consequently $\cN$ is
degradable, with an entanglement-breaking degrading channel, and $\psi_{AE}$
is separable. These structural facts follow from the measure-and-prepare
characterization of entanglement-breaking channels~\cite{HSR03,BHTW10}.

\begin{proposition}\label{prop:had-characterization}
A bipartite state $\rho_{AB}$ with purification $\psi_{ABE}$ admits a
Hadamard realization as defined in Section~\ref{sec:hadamard}, up to local isometries and pure local ancillas, if and only if
$\psi_{AE}$ is separable.
\end{proposition}
\begin{proof}
One direction follows from~\eqref{eq:had-purification} and the definition of
an entanglement-breaking complementary channel. For the converse, restrict
$A$ to the support of $\tau_A\coloneqq\psi_A$ and diagonalize it as
$\tau_A=\sum_i t_i\proj i_A$, where $t_i>0$. Define the positive semidefinite operator
\begin{equation}
J_{AE}\coloneqq
(\tau_A^{-1/2}\otimes I_E)\psi_{AE}
(\tau_A^{-1/2}\otimes I_E).
\label{eq:had-filtered-choi}
\end{equation}
Local filtering preserves separability, and $\Tr_E[J_{AE}]=I_A$.
We can therefore write
\begin{equation}
J_{AE}=\sum_y L_A^y\otimes\eta_E^y,
\qquad L_A^y\geq0,\qquad \sum_yL_A^y=I_A,
\end{equation}
where every $\eta_E^y$ is a state. Let $A'$ be a copy of $A$, identified
through the chosen eigenbasis, and define
\begin{equation}
\mathcal M_{A'\to E}(\xi_{A'})
\coloneqq\Tr_A[(\xi^{T}\otimes I_E)J_{AE}]
=\sum_y\Tr[(L^y)^T\xi_{A'}]\eta_E^y.
\label{eq:had-prepare-channel}
\end{equation}
In the partial trace, $\xi^T$ is the transpose of $\xi_{A'}$ transported to
$A$; in the final expression, $(L^y)^T$ is transported to $A'$.
The operators $(L^y)^T$ form a positive operator-valued measure, so
$\mathcal M_{A'\to E}$ is a measure-and-prepare channel.
For
\begin{equation}
\ket{\phi^\tau}_{AA'}
\coloneqq\sum_i\sqrt{t_i}\ket i_A\ket i_{A'},
\end{equation}
the Choi identity gives
\begin{equation}
(\id_A\otimes\mathcal M_{A'\to E})(\phi^\tau_{AA'})
=(\tau_A^{1/2}\otimes I_E)J_{AE}(\tau_A^{1/2}\otimes I_E)
=\psi_{AE}.
\label{eq:had-choi-recovery}
\end{equation}
A complementary channel of $\mathcal M$ is Hadamard. An isometric dilation
of $\mathcal M$ acting on $\phi^\tau_{AA'}$ purifies the same $AE$ state as
$\psi_{ABE}$. Uniqueness of purification, up to an isometry on the purifying
support, identifies its other output with $B$. This proves the assertion.
\end{proof}

\subsection{An additive coherent-information face}
\begin{theorem}[Coherent-information face of a Hadamard state~\cite{LDS18}]
\label{thm:had-face}
Let $\rho_{AB}$ be a Hadamard-generated bipartite state as defined
in Section~\ref{sec:hadamard}, and let $\sS_{\lambda,\mu}$ be the
static capacity formula in~\eqref{eq:support}. For all $n\in\mathbb{N}$ and
$0\leq\lambda\leq1$,
\begin{equation}
\sS_{\lambda,0}(\rho_{AB}^{\otimes n})
=nI(A\rangle B)_\rho.
\label{eq:had-mu0}
\end{equation}
\end{theorem}
A Hadamard state is degradable in the state sense: the degrading channel
$\mathcal D\colon B\to E$ described after~\eqref{eq:had-purification} satisfies
$ (\id_A\otimes\mathcal D)(\rho_{AB})=\psi_{AE}$.
Ref.~\cite[Proposition~4]{LDS18} proved that the
one-way distillable entanglement of a degradable state equals its coherent
information. Their result also identifies the optimization over Alice's
instruments at every blocklength. In the notation of~\eqref{eq:instrument-state},
it gives
\begin{equation}
\sup_{\mathcal T} I(S\rangle B^nX)_\sigma
=I(A^n\rangle B^n)_{\rho^{\otimes n}}=nI(A\rangle B)_\rho.
\label{eq:degradable-instrument-bound}
\end{equation}
The state $\sigma$ here is induced by an instrument on $A^n$ applied to
$\psi_{ABE}^{\otimes n}$. Dropping the nonpositive term
$-\lambda I(X;E^n|B^n)_\sigma$ from~\eqref{eq:support} gives the upper
bound in~\eqref{eq:had-mu0}. The identity instrument attains it, since its
classical output is trivial. Thus Theorem~\ref{thm:had-face} is a direct
consequence of Ref.~\cite{LDS18}. Appendix~\ref{app:hadamard} gives an
independent entropy proof for Hadamard states.

\subsection{The remaining general Hadamard optimization}

We now express the objective function in the static capacity formula of \eqref{eq:support} as a fixed term, depending only on
the initial state $\psi_{ABE}$, minus an
instrument-dependent nonnegative quantity.
This identifies what must be controlled to extend
Theorem~\ref{thm:had-face} beyond $\mu=0$.

Let $\rho_{AB}$ be a Hadamard-generated state with purification
$\psi_{ABE}$. For an instrument $V\coloneqq\left(V_x\right)_x$
with output state $\sigma_{XSFBE}$ from~\eqref{eq:instrument-state},
and parameters $\lambda,\mu$ satisfying~\eqref{eq:parameters},
define its score by
\begin{equation}
s_{\lambda,\mu}(V;\rho_{AB})
\coloneqq I(S\rangle BX)_\sigma
-\lambda I(X;E|B)_\sigma
-\mu I(S;E|BX)_\sigma.
\label{eq:had-score}
\end{equation}

\begin{lemma}[Decomposition of the Hadamard score]
\label{lem:had-score-decomposition}
Let $\rho_{AB}$ be a Hadamard-generated state with purification
$\psi_{ABE}$, and let $V\coloneqq\left(V_x\right)_x$ be a local
instrument with associated state $\sigma_{XSFBE}$ as
in~\eqref{eq:instrument-state}.
For $\mu\geq0$ and $\mu\leq\lambda\leq\mu+1$, define
\begin{equation}
\begin{aligned}
\Gamma_{\lambda,\mu}(V;\rho_{AB})
\coloneqq{}&(1+\mu-\lambda)
\big[I(X;B)_\sigma-I(X;E)_\sigma\big]
+(\lambda-\mu)I(X;B|E)_\sigma\\
&+H(F|EX)_\sigma+\mu H(E|FX)_\sigma.
\end{aligned}
\label{eq:had-gamma}
\end{equation}
Then the score in~\eqref{eq:had-score} satisfies
\begin{equation}
s_{\lambda,\mu}(V;\rho_{AB})
=I(A\rangle B)_\rho-\mu H(E|B)_\psi
-\Gamma_{\lambda,\mu}(V;\rho_{AB}),
\label{eq:had-penalty}
\end{equation}
and $\Gamma_{\lambda,\mu}(V;\rho_{AB})\geq0$.
\end{lemma}

\begin{proof}
The chain rule for conditional mutual information gives
\begin{equation}
s_{\lambda,\mu}(V;\rho_{AB})
=I(S\rangle BX)_\sigma
-\mu I(SX;E|B)_\sigma
-(\lambda-\mu)I(X;E|B)_\sigma.
\label{eq:had-score-chain}
\end{equation}
We evaluate the first two information quantities separately.

First, conditioned on each outcome $x$ of nonzero probability, the state
on $SFBE$ is pure. Consequently,
\begin{align}
I(S\rangle BX)_\sigma
&=H(B|X)_\sigma-H(SB|X)_\sigma\\
&=H(B|X)_\sigma-H(FE|X)_\sigma\\
&=H(B|X)_\sigma-H(E|X)_\sigma-H(F|EX)_\sigma.
\label{eq:had-entropy-split}
\end{align}
The first equality follows from the definition of coherent information,
the second follows from conditional purity, and the third follows from
the entropy chain rule.

The instrument acts only on $A$, so its averaged output satisfies
$\sigma_{BE}=\psi_{BE}$.
Expanding the first two conditional entropies
in~\eqref{eq:had-entropy-split} therefore gives
\begin{align}
H(B|X)_\sigma-H(E|X)_\sigma
&=H(B)_\psi-H(E)_\psi-I(X;B)_\sigma+I(X;E)_\sigma\\
&=I(A\rangle B)_\rho
-\big[I(X;B)_\sigma-I(X;E)_\sigma\big].
\label{eq:had-classical-entropy-split}
\end{align}
For the last equality, purity of $\psi_{ABE}$ implies
$H(AB)_\rho=H(E)_\psi$, and hence
$I(A\rangle B)_\rho=H(B)_\psi-H(E)_\psi$.
Substituting~\eqref{eq:had-classical-entropy-split}
into~\eqref{eq:had-entropy-split} yields
\begin{equation}
I(S\rangle BX)_\sigma
=I(A\rangle B)_\rho
-\big[I(X;B)_\sigma-I(X;E)_\sigma\big]
-H(F|EX)_\sigma.
\label{eq:had-coherent-identity}
\end{equation}

Next, consider $I(SX;E|B)_\sigma$.
Conditional purity also gives
\begin{align}
H(E|SBX)_\sigma
&=H(ESB|X)_\sigma-H(SB|X)_\sigma\\
&=H(F|X)_\sigma-H(EF|X)_\sigma\\
&=-H(E|FX)_\sigma.
\end{align}
Indeed, for each conditional pure state on $SFBE$, the complementary
systems $ESB$ and $F$ have the same entropy, as do $SB$ and $EF$.
Averaging these equalities over $x$ proves the second equality above.
Using this identity and $\sigma_{BE}=\psi_{BE}$, we obtain
\begin{equation}
I(SX;E|B)_\sigma
=H(E|B)_\psi-H(E|SBX)_\sigma
=H(E|B)_\psi+H(E|FX)_\sigma.
\label{eq:had-quantum-identities}
\end{equation}

Substituting~\eqref{eq:had-coherent-identity}
and~\eqref{eq:had-quantum-identities}
into~\eqref{eq:had-score-chain} now gives
\begin{equation}
\begin{aligned}
s_{\lambda,\mu}(V;\rho_{AB})
={}&I(A\rangle B)_\rho-\mu H(E|B)_\psi
-\big[I(X;B)_\sigma-I(X;E)_\sigma\big]\\
&-(\lambda-\mu)I(X;E|B)_\sigma
-H(F|EX)_\sigma-\mu H(E|FX)_\sigma.
\end{aligned}
\label{eq:had-expanded-score}
\end{equation}

To obtain the expression in~\eqref{eq:had-penalty}, we regroup the terms
involving the classical outcome $X$.
Expanding $I(X;BE)_\sigma$ in the two possible orders gives
\begin{equation}
I(X;B|E)_\sigma
=I(X;B)_\sigma-I(X;E)_\sigma+I(X;E|B)_\sigma.
\label{eq:had-classical-chain}
\end{equation}
Consequently,
\begin{equation}
\begin{aligned}
&I(X;B)_\sigma-I(X;E)_\sigma
+(\lambda-\mu)I(X;E|B)_\sigma\\
&\qquad=(1+\mu-\lambda)
\big[I(X;B)_\sigma-I(X;E)_\sigma\big]
+(\lambda-\mu)I(X;B|E)_\sigma.
\end{aligned}
\label{eq:had-classical-regrouping}
\end{equation}
Substituting~\eqref{eq:had-classical-regrouping}
into~\eqref{eq:had-expanded-score} and using the definition
of $\Gamma_{\lambda,\mu}$ in~\eqref{eq:had-gamma}
proves~\eqref{eq:had-penalty}.

It remains to prove that $\Gamma_{\lambda,\mu}(V;\rho_{AB})\geq0$.
The parameter constraints imply that all coefficients
in~\eqref{eq:had-gamma} are nonnegative.
The degrading channel from $B$ to $E$ reproduces the $XE$ marginal
from the $XB$ marginal after Alice's instrument.
Data processing therefore gives
$I(X;B)_\sigma\geq I(X;E)_\sigma$.
Strong subadditivity gives $I(X;B|E)_\sigma\geq0$.

For the remaining terms, $\psi_{AE}$ is separable because
$\rho_{AB}$ is Hadamard generated.
Since the instrument acts only on $A$, the normalized conditional
state $\sigma^x_{FE}$ is separable for every outcome $x$ of nonzero
probability.
Write
\begin{equation}
\sigma^x_{FE}
=\sum_z q_{z|x}\alpha_F^{xz}\otimes\beta_E^{xz},
\end{equation}
where $(q_{z|x})_z$ is a probability distribution and
$\alpha_F^{xz}$ and $\beta_E^{xz}$ are density operators.
Concavity of conditional entropy and the product form of each
term give
\begin{equation}
H(F|E)_{\sigma^x}
\geq\sum_z q_{z|x}H(F)_{\alpha^{xz}}\geq0,
\qquad
H(E|F)_{\sigma^x}
\geq\sum_z q_{z|x}H(E)_{\beta^{xz}}\geq0.
\label{eq:had-conditional-entropy-main}
\end{equation}
Averaging over $x$ yields
$H(F|EX)_\sigma\geq0$ and $H(E|FX)_\sigma\geq0$.
Thus every summand in~\eqref{eq:had-gamma} is nonnegative,
which completes the proof.
\end{proof}

We now return to the single-letterization criterion in
Section~\ref{sec:single-letterization}.
As explained there, additivity of the static capacity formula on tensor
powers, for every parameter pair in~\eqref{eq:parameters}, removes the
regularization in~\eqref{eq:reg-support} and determines the capacity
region from the one-copy information region.
Lemma~\ref{lem:had-score-decomposition} expresses this additivity
question entirely in terms of the optimization of
$\Gamma_{\lambda,\mu}$.
Indeed, taking the supremum over instruments
in~\eqref{eq:had-penalty} gives
\begin{equation}
\sS_{\lambda,\mu}(\rho_{AB})
=I(A\rangle B)_\rho-\mu H(E|B)_\psi
-\inf_V\Gamma_{\lambda,\mu}(V;\rho_{AB}).
\label{eq:had-optimized-score}
\end{equation}
The first two terms on the right are additive on tensor-product
states. Applying~\eqref{eq:had-optimized-score} to
$\rho_{AB}^{\otimes n}$ with purification $\psi_{ABE}^{\otimes n}$
therefore shows that the self-additivity condition
in~\eqref{eq:selfadd} is equivalent to
\begin{equation}
\inf_{V_n}\Gamma_{\lambda,\mu}
(V_n;\rho_{AB}^{\otimes n})
=n\inf_V\Gamma_{\lambda,\mu}(V;\rho_{AB})
\qquad\text{for all }n\in\mathbb{N}.
\label{eq:had-gamma-selfadd}
\end{equation}
Here $V_n$ ranges over all local instruments on $A^n$, including
collective instruments.
The same reasoning applies to products of two different
Hadamard-generated states: additivity of the static capacity formula
for such a product is equivalent to additivity of the corresponding
infimum of $\Gamma_{\lambda,\mu}$.

Nonnegativity alone does not show that the infimum of
$\Gamma_{\lambda,\mu}$ is additive. For two Hadamard-generated states
$\rho_{AB}$ and $\tau_{A'B'}$, product instruments only prove
\begin{align}
\inf_{V_{12}}\Gamma_{\lambda,\mu}
(V_{12};\rho_{AB}\otimes\tau_{A'B'})&\leq
\inf_{V_1}\Gamma_{\lambda,\mu}(V_1;\rho_{AB})
+\inf_{V_2}\Gamma_{\lambda,\mu}(V_2;\tau_{A'B'}).
\end{align}
The reverse inequality would require an additional argument controlling
collective instruments.

To examine the effect of retaining $F$ as part of Alice's quantum output,
let $V^{\mathrm{ret}}$ be the instrument with the same maps and classical
output as $V$, but with quantum output $SF$ and a trivial discarded system.
Evaluating both scores on the original state $\sigma_{XSFBE}$ gives
\begin{align}
&s_{\lambda,\mu}(V^{\mathrm{ret}};\rho_{AB})
-s_{\lambda,\mu}(V;\rho_{AB})\nonumber \\
&=I(SF\rangle BX)_\sigma-I(S\rangle BX)_\sigma
-\mu\big[I(SF;E|BX)_\sigma-I(S;E|BX)_\sigma\big]\\
&=-H(F|SBX)_\sigma-\mu I(F;E|SBX)_\sigma\\
&=H(F|EX)_\sigma-\mu I(F;E|X)_\sigma.
\label{eq:retain}
\end{align}
The classical term cancels because $XBE$ is unchanged. Conditional purity
gives $H(F|SBX)_\sigma=-H(F|EX)_\sigma$ and
$I(F;E|SBX)_\sigma=I(F;E|X)_\sigma$, justifying the last equality.
The two quantities in the last line are nonnegative individually, but their
difference need not have a fixed sign when $\mu>0$. The following sections
give a more explicit reduction and a strict example of this obstruction.

\Needspace{10\baselineskip}
\section{Maximally correlated states: an exact matrix reduction}\label{sec:mc}
Maximally correlated states have a distinguished role in entanglement
distillation~\cite{Rains01,DW05}. We study their full three-resource
trade-off while retaining Alice's discarded quantum system in the optimization.
The main results are an exact matrix formulation at every blocklength,
a strict example of an advantage from a discarded quantum system, and
an additive outer bound obtained from an information-bottleneck
optimization. The full single-letter capacity region remains open.
Let
\begin{equation}
\rho_{\mc,AB}\coloneqq\sum_{i,j=1}^d r_{ij}\ket{ii}\bra{jj}_{AB},
\qquad r\geq0,\quad\Tr r=1.
\label{eq:mcstate}
\end{equation}
We regard the coefficient matrix $r$ as a density operator on a
$d$-dimensional system $\widetilde A$, with distinguished basis
$\left\{\ket i_{\widetilde A}\right\}_{i=1}^d$. Its $n$ copies are
$\widetilde A_1,\ldots,\widetilde A_n$, and their joint system is denoted by
$\widetilde A^n\coloneqq\widetilde A_1\cdots\widetilde A_n$.
The tilde is part of the system label; a numerical subscript denotes a copy.
Set
\begin{equation}
p_i\coloneqq r_{ii}. \label{eq:p_i-def-max-corr}
\end{equation}
Coordinates with $p_i=0$
can be deleted: positive semidefiniteness gives $|r_{ij}|^2\leq p_ip_j$, so the corresponding
rows and columns vanish. Henceforth $p_i>0$.

Define the correlation matrix
\begin{equation}
G_{ij}\coloneqq\frac{r_{ij}}{\sqrt{p_ip_j}}.
\label{eq:G}
\end{equation}
Then $G\geq0$ and $G_{ii}=1$. Choose unit vectors $\ket{e_i}_E$ satisfying
$\langle e_j|e_i\rangle=G_{ij}$, and write $e_{i,E}\coloneqq\proj{e_i}_E$.
A purification of the maximally correlated state is
\begin{equation}
\ket\psi_{ABE}\coloneqq\sum_i\sqrt{p_i}\ket i_A\ket i_B\ket{e_i}_E.
\label{eq:purification}
\end{equation}
Tracing out $A$ gives
\begin{equation}
\psi_{BE}=\sum_i p_i\proj i_B\otimes e_{i,E}.
\label{eq:cqBE}
\end{equation}
It is useful to introduce at this point the associated classical--quantum source
\begin{equation}
\omega_{YE}\coloneqq\sum_i p_i\proj i_Y\otimes e_{i,E}.
\label{eq:cqsource}
\end{equation}
The register $Y$ records a classical input label. It is a separate register from
$B$, although $\omega_{YE}$ and $\psi_{BE}$ have the same matrix representation.
Since each conditional state $e_{i,E}$ is pure,
\begin{align}
H(BE)_\psi&=H(Y)_\omega+\sum_i p_iH(E)_{e_i}=H(Y)_\omega,\\
H(B)_\psi&=H(Y)_\omega,\\
H(E|B)_\psi&=0.
\label{eq:EgivenBzero}
\end{align}
The isometry $\ket i_{\widetilde A}\mapsto\ket{ii}_{AB}$ maps $r_{\widetilde A}$ to
$\rho_{\mc,AB}$, so their nonzero spectra coincide. Purity of $\psi_{ABE}$ then gives
\begin{align}
H(E)_\psi&=H(AB)_{\rho_{\mc}}=H(\widetilde A)_r,\\
I(A\rangle B)_{\rho_{\mc}}&=H(Y)_\omega-H(\widetilde A)_r.
\label{eq:cohinfo}
\end{align}

The Schur-product channel on $\widetilde A$ is
\begin{equation}
\cN_G(\theta)\coloneqq G\circ\theta,
\label{eq:Schur}
\end{equation}
where $\circ$ denotes entrywise multiplication. Its output is identified with
$\widetilde A$ when we evaluate the matrix functional below. An isometric extension is
$\ket i_{\widetilde A}\mapsto\ket i_{\widetilde A}\ket{e_i}_E$, with complementary channel
\begin{equation}
\cN_G^c(\theta)=\sum_i\theta_{ii}e_{i,E}.
\end{equation}
This complement measures the input basis and prepares a state, so it is
entanglement breaking. Thus $\cN_G$ is Hadamard. Applying $\cN_G$ to one share
of $\sum_i\sqrt{p_i}\ket{ii}$, with its output relabeled as $B$, produces
$\rho_{\mc,AB}$. No orthogonality assumption is made on the vectors $e_i$.

\subsection{Canonical instruments and an exact matrix formula}
We first reduce the general instrument optimization in
\eqref{eq:instrument-state} and~\eqref{eq:facet1}--\eqref{eq:facet3} to a
form adapted to the distinguished basis in~\eqref{eq:mcstate}.
The reduction preserves the full one-copy region and applies to a whole
block of copies. It is the starting point for the matrix optimization below.

For an instrument $V_x\colon A\to SF$, write
\begin{equation}
\sigma_{XSFBE}\coloneqq\sum_x\proj x_X\otimes
(V_x\otimes I_{BE})\psi_{ABE}(V_x^\dagger\otimes I_{BE}).
\end{equation}
Thus each conditional quantum state is pure. The following normal form
retains an orthogonal basis label and restricts the discarded states to be pure
conditional on that label and the instrument outcome.

\begin{theorem}[Canonical instrument normal form]
\label{thm:normal-form}
Let $\rho_{\mc,AB}$ be the maximally correlated state in~\eqref{eq:mcstate},
with a $d\times d$ coefficient density matrix $r$ and distinguished basis
$\left\{\ket i\right\}_{i=1}^d$. The instrument union defining $\cR^{(1)}$
in~\eqref{eq:facet1}--\eqref{eq:facet3} is unchanged if restricted to
instruments $V_x\colon A\to SF$ of the form
\begin{equation}
V_x\ket i_A=\sqrt{a(x|i)}\ket i_S\ket{f_{xi}}_F,
\qquad a(x|i)\geq0,\quad\sum_xa(x|i)=1,
\label{eq:canonical}
\end{equation}
where every $\ket{f_{xi}}_F$ is a unit vector. One can take $\dim S=d$
and $\dim F\leq d$. With the purification in~\eqref{eq:purification},
the induced state $\sigma_{XSFBE}$ of~\eqref{eq:instrument-state} satisfies
$I(X;E|B)_\sigma=0$. For every blocklength $n\in\mathbb{N}$, the same
statement holds with multi-indices $i^n$, $\dim S=d^n$, and $\dim F\leq d^n$.
\end{theorem}
\begin{proof}
See Appendix~\ref{app:mc-normal-form}.
\end{proof}

\begin{remark}[What the normal form does not say]
The vectors $f_{xi}$ need not be orthogonal or equal. At blocklength $n$, the
vectors $f_{x,i^n}$ need not factor across copies. The pure-state refinement is
not an arbitrary rank-one measurement of the discarded system of the original
instrument. It constructs a new instrument and uses a tuple independent of the
input label conditional on $X$, after exploiting the distinguished basis.
It therefore does not contradict the nonmonotonicity of arbitrary instrument
refinement in the general Hadamard problem.
\end{remark}

\subsection{The finite-dimensional matrix optimization}

Theorem~\ref{thm:normal-form} reduces the instrument optimization to
conditional probabilities and discarded pure-state vectors.
We now encode these quantities in an ensemble of density operators,
whose matrix entries contain the probabilities and the overlaps of
the discarded vectors.
This gives an exact expression for the static capacity formula
in Theorem~\ref{thm:matrix-formula} and a description of the complete
one-copy information region in~\eqref{eq:exact-onecopy-region}.
The ensemble is constrained by its average diagonal, which is fixed
by the resource state.
The same construction applies at every blocklength, allowing us to
state precisely the remaining additivity question and a sufficient
condition for removing regularization.
Finally, we express the matrix objective in terms of the relative
entropy of coherence, providing another way to study this optimization.

\subsubsection{The conditional entropy calculation}
We now evaluate the information quantities of the canonical instrument
in~\eqref{eq:canonical}. The key step is to represent its discarded-system
overlaps by a density matrix, whose spectrum also determines the required
entropies after the Schur channel.

Recall the definition of $p_i$ in \eqref{eq:p_i-def-max-corr}. For a canonical instrument, set
$w_x\coloneqq\sum_i p_i a(x|i)$ and
$q_i^x\coloneqq p_i a(x|i)/w_x$, omitting zero-probability outcomes. 
Define
\begin{align}
\ket{\Psi_x}_{SFBE}
&\coloneqq\sum_i\sqrt{q_i^x}\ket i_S\ket{f_{xi}}_F\ket i_B\ket{e_i}_E,
\label{eq:canonical-state}\\
\sigma_{XSFBE}&\coloneqq\sum_xw_x\proj x_X\otimes\proj{\Psi_x}_{SFBE}.
\end{align}
Here and below, $\sigma$ denotes the state for the canonical instrument under
consideration. We write $\Psi_x\coloneqq\proj{\Psi_x}$ and
$f_{xi,F}\coloneqq\proj{f_{xi}}_F$ for the corresponding density operators. Its associated classical--quantum extension is
\begin{equation}
\omega_{Y X F E}\coloneqq\sum_{i,x}p_i a(x|i)
\proj i_Y\otimes\proj x_X\otimes\proj{f_{xi}}_F\otimes e_{i,E}.
\end{equation}
It extends the fixed source $\omega_{YE}$ in~\eqref{eq:cqsource} and has
$\omega_{XFE}=\sigma_{XFE}$. Define a density matrix on $\widetilde A$ by
\begin{equation}
(\theta_x)_{ij}\coloneqq
\sqrt{q_i^xq_j^x}\langle f_{xj}|f_{xi}\rangle.
\label{eq:theta}
\end{equation}
It is the reduced state on $\widetilde A$ of
$\sum_i\sqrt{q_i^x}\ket i_{\widetilde A}\ket{f_{xi}}_F$. Consequently,
\begin{align}
\theta_x&\geq0,\qquad\Tr\theta_x=1,\qquad
\Delta(\theta_x)=\diag(q^x),\\
H(F)_{\Psi_x}&=H(\widetilde A)_{\theta_x},
\label{eq:theta-spectral}
\end{align}
where $\Delta$ completely dephases $\widetilde A$ in its distinguished basis.
Using the definitions of $\theta_x$ in~\eqref{eq:theta}
and $G$ in~\eqref{eq:G}, we find that the $SB$ marginal
of the state in~\eqref{eq:canonical-state} is
\begin{equation}
(\Psi_x)_{SB}
=\sum_{i,j}(\theta_x)_{ij}G_{ij}\ket{ii}\bra{jj}_{SB}.
\end{equation}
Consequently,
\begin{equation}
H(FE)_{\Psi_x}=H(SB)_{\Psi_x}=H(\widetilde A)_{\cN_G(\theta_x)}.
\label{eq:FE-spectral}
\end{equation}
The first equality follows from purity of $\Psi_x$.
The second follows because $(\Psi_x)_{SB}$ is the isometric embedding
of $\cN_G(\theta_x)$ under
$\ket{i}_{\widetilde A}\mapsto\ket{ii}_{SB}$, which preserves the nonzero
eigenvalues and hence the entropy.
It follows that
\begin{align}
I(S\rangle BX)_\sigma
&=\sum_xw_x\big[H(\widetilde A)_{\Delta(\theta_x)}
-H(\widetilde A)_{\cN_G(\theta_x)}\big],\label{eq:matrix-c}\\
I(SX;E|B)_\sigma
&=\sum_xw_x\big[H(\widetilde A)_{\cN_G(\theta_x)}
-H(\widetilde A)_{\theta_x}\big],\label{eq:matrix-ell}\\
I(X;E|B)_\sigma&=0.
\end{align}
The only ensemble constraint is
\begin{equation}
\sum_xw_x\Delta(\theta_x)=\diag(p).
\label{eq:diagonal-constraint}
\end{equation}
No constraint $\sum_xw_x\theta_x=r$ is imposed. The matrices $\theta_x$ encode
the discarded vectors through their Gram matrices; they are distinct from the
physical average channel input in other capacity optimizations.

\begin{theorem}[Exact one-copy support formula]
\label{thm:matrix-formula}
Let $\rho_{\mc,AB}$ and its coefficient density matrix $r$ be as
in~\eqref{eq:mcstate}, with $p_i\coloneqq r_{ii}>0$ after deletion of
zero-probability coordinates. Let $p\coloneqq(p_1,\ldots,p_d)$, let $G$
be the correlation matrix in~\eqref{eq:G}, and let $\cN_G$ be the Schur
channel in~\eqref{eq:Schur}. On the $d$-dimensional matrix system $\widetilde A$,
let $\Delta$ denote dephasing in the distinguished basis. For $\mu\geq0$
and a density operator $\theta$ on $\widetilde A$, define
\begin{equation}
L_\mu^G(\theta)\coloneqq H(\widetilde A)_{\Delta(\theta)}
+\mu H(\widetilde A)_\theta-(1+\mu)H(\widetilde A)_{\cN_G(\theta)}
\label{eq:L}
\end{equation}
and
\begin{equation}
\mathfrak S_\mu(G,p)
\coloneqq\max_{\substack{\left(w_x,\theta_x\right)_x\\
\sum_xw_x\Delta(\theta_x)=\diag(p)}}
\sum_xw_xL_\mu^G(\theta_x).
\label{eq:matrix-functional}
\end{equation}
The maximum is over finite ensembles of density operators $\theta_x$ on
$\widetilde A$ with probability weights $w_x$, subject to the displayed mean-diagonal
constraint; $X$ denotes their classical label. For every
$\mu\leq\lambda\leq\mu+1$, the static capacity formula in~\eqref{eq:support} obeys
\begin{equation}
\sS_{\lambda,\mu}(\rho_{\mc,AB})=\mathfrak S_\mu(G,p).
\label{eq:exact-support}
\end{equation}
All matrices in~\eqref{eq:matrix-functional} act on the $d$-dimensional system
$\widetilde A$, and $|X|\leq d$ suffices. The optimized support is independent of $\lambda$.
\end{theorem}
\begin{proof}
See Appendix~\ref{app:mc-matrix-formula}.
\end{proof}

\subsubsection{The complete one-copy information region}
The scalar optimization in~\eqref{eq:matrix-functional} determines one
supporting value at a time. The same entropy calculation also gives the
generating rate points themselves, and hence the entire one-copy information
region after adding the unit-resource cone in~\eqref{eq:unit}.

For a density operator $\theta$ on $\widetilde A$, define
\begin{align}
Q(\theta)&\coloneqq-\tfrac12\big[
H(\widetilde A)_{\cN_G(\theta)}-H(\widetilde A)_\theta\big],\label{eq:primitive-matrix-Q}\\
E(\theta)&\coloneqq H(\widetilde A)_{\Delta(\theta)}
-\tfrac12\big[H(\widetilde A)_\theta+H(\widetilde A)_{\cN_G(\theta)}\big].
\label{eq:primitive-matrix}
\end{align}
The exact one-copy region is
\begin{equation}
\cR^{(1)}(\rho_{\mc,AB})
=\bigcup_{\substack{\left(w_x,\theta_x\right)_x\\
\sum_xw_x\Delta(\theta_x)=\diag(p)}}
\left[\left(0,\sum_xw_xQ(\theta_x),\sum_xw_xE(\theta_x)\right)+\cU\right].
\label{eq:exact-onecopy-region}
\end{equation}
The equality and the sufficient cardinality bound $|X|\leq d+1$ are proved in
Appendix~\ref{app:mc-region}. The primitive points satisfy $Q(\theta)\leq0$
and $E(\theta)\geq0$. The cardinality bound and continuity of entropy make
the constrained set of average primitive points compact. Its sum with the
closed cone $\cU$ is therefore closed, so no additional closure is needed
in~\eqref{eq:exact-onecopy-region}.

\subsubsection{What would remove regularization}
The preceding formula evaluates a single copy. To identify it with the
operational capacity region, we must control collective instruments as
required by Section~\ref{sec:single-letterization}. The matrix reduction
turns this issue into the following additivity question.

At blocklength $n$, the exact formula is
\begin{equation}
\sS_{\lambda,\mu}(\rho_{\mc,AB}^{\otimes n})
=\mathfrak S_\mu(G^{\otimes n},p^{\otimes n}).
\label{eq:block-formula}
\end{equation}
The matrices on the right act on $\widetilde A^n$ and need not be products.
For all $n\in\mathbb N$, the corresponding sufficient outcome bounds are
$d^n$ for the support optimization and $d^n+1$ for an arbitrary primitive
rate point, as proved in Appendices~\ref{app:mc-matrix-formula}
and~\ref{app:mc-region}.
A sufficient additivity theorem would be
\begin{equation}
\mathfrak S_\mu(G_1\otimes G_2,p_1\otimes p_2)
\leq\mathfrak S_\mu(G_1,p_1)+\mathfrak S_\mu(G_2,p_2).
\label{eq:target-tensorization}
\end{equation}
Product ensembles give the reverse inequality. Self-additivity for every
$\mu\geq0$ would identify~\eqref{eq:exact-onecopy-region} with the full operational region.

\subsubsection{A coherence interpretation}
To study the additivity question in~\eqref{eq:target-tensorization},
it is useful to rewrite the matrix objective as a difference of coherence
quantities. This representation will also provide the connection to the
dephasing-contraction argument in Section~\ref{sec:dephasing-contraction}.

The relative entropy of coherence is defined as
$\Cr(\theta)\coloneqq H(\widetilde A)_{\Delta(\theta)}-H(\widetilde A)_\theta$;
see~\cite{WY16} for its operational role. Since $\Delta\circ \cN_G=\Delta$,
\begin{align}
&(1+\mu)\Cr(\cN_G(\theta))-\mu\Cr(\theta)\nonumber \\
&=(1+\mu)\big[H(\widetilde A)_{\Delta(\theta)}-H(\widetilde A)_{\cN_G(\theta)}\big]-\mu\big[H(\widetilde A)_{\Delta(\theta)}-H(\widetilde A)_\theta\big]\\
&=L_\mu^G(\theta).
\label{eq:coherence-form}
\end{align}
This identity does not invoke an existing additivity theorem for this weighted
coherence difference.

\subsection{Efficient instruments, classical memories, and the coherent-information face}
We next evaluate two restrictions of the instrument optimization and
identify an exact supporting face. Efficient instruments have no discarded
quantum system, and commuting discarded states can be reduced to this
case without decreasing the score. The resulting restricted formula is
additive. At $\mu=0$, it also attains the unrestricted support.
\subsubsection{The efficient restriction is additive}
We first restrict the canonical instruments in~\eqref{eq:canonical}
to those with no discarded quantum system. Their matrix variables become
pure states, and the remaining optimization is over classical probability
distributions. This restricted optimization admits an additivity proof.

An efficient instrument has a one-dimensional discarded system $F$. Its canonical
form has vectors $f_{xi}$ equal up to phases, so each $\theta_x$ has rank one.
For a distribution $q$, define
\begin{equation}
q_Y\coloneqq\sum_iq_i\proj i_Y,
\qquad\eta^q_{YE}\coloneqq\sum_iq_i\proj i_Y\otimes e_{i,E}.
\end{equation}
The efficient objective depends only on $q$:
\begin{equation}
g_\mu(q)\coloneqq H(Y)_q-(1+\mu)H(E)_{\eta^q}.
\end{equation}
It follows that
\begin{equation}
\sS_\mu^{\eff}(p)
\coloneqq\max_{\sum_xw_xq^x=p}\sum_xw_xg_\mu(q^x).
\label{eq:efficient}
\end{equation}
The notation suppresses the dependence on the fixed correlation matrix $G$,
or equivalently on the conditional states $\left(e_{i,E}\right)_i$.
This expression also optimizes over arbitrary efficient instruments before
the basis-copying construction: that construction remains efficient and dominates
the original information region.

\begin{proposition}[Additivity of the efficient restriction]\label{prop:efficient}
For every $\mu\geq0$ and two maximally correlated states as
in~\eqref{eq:mcstate}, the efficient optimum defined in~\eqref{eq:efficient}
for their product equals the sum of the separate optima at the same $\mu$.
The product uses the product input distribution and the tensor-product
correlation matrix, with each separate optimum evaluated for its own
correlation matrix and distribution.
\end{proposition}
\begin{proof}
See Appendix~\ref{app:mc-efficient}.
\end{proof}

\subsubsection{Commuting discarded states}
We next allow a discarded system but require its conditional states to
commute. A refinement in their common eigenbasis reduces this class to the
efficient optimization in~\eqref{eq:efficient}, so such memories cannot
improve its value.

For the instrument operators in~\eqref{eq:instrument}, define
\begin{equation}
a(x|i)\coloneqq\|V_x\ket i_A\|^2,\qquad
\tau_{xi,F}\coloneqq\frac{\Tr_S[V_x\proj i_A V_x^\dagger]}{a(x|i)}
\quad\text{when }a(x|i)>0.
\label{eq:commuting-conditional-states}
\end{equation}
The condition is that, for each fixed $x$,
$[\tau_{xi,F},\tau_{xj,F}]=0$ for all defined pairs $i,j$.
Instruments satisfying this condition cannot improve the efficient optimum. Appendix~\ref{app:mc-commuting} proves this by refining in
a common eigenbasis and incorporating its label into the classical outcome.

\subsubsection{The exact \texorpdfstring{$\mu=0$}{mu=0} face}
At $\mu=0$, the coherent-information result for Hadamard states already
determines the unrestricted optimization. We record its expression in the
matrix notation of~\eqref{eq:cohinfo}, including its value on tensor powers.

\begin{proposition}\label{prop:mc-mu-zero}
Let $\rho_{\mc,AB}$ and its coefficient state $r_{\widetilde A}$ be
as in~\eqref{eq:mcstate}, and let $\omega_{YE}$ be their associated
source in~\eqref{eq:cqsource}. For all $n\in\mathbb{N}$ and
$0\leq\lambda\leq1$, the static capacity formula in~\eqref{eq:support} satisfies
\begin{equation}
\sS_{\lambda,0}(\rho_{\mc,AB}^{\otimes n})
=n\big[H(Y)_\omega-H(\widetilde A)_r\big].
\label{eq:mu0}
\end{equation}
\end{proposition}

\begin{proof}
As established after~\eqref{eq:Schur}, the maximally correlated
state $\rho_{\mc,AB}$ is Hadamard generated.
Theorem~\ref{thm:had-face} therefore gives, for all
$n\in\mathbb N$ and $0\leq\lambda\leq1$,
\begin{equation}
\sS_{\lambda,0}(\rho_{\mc,AB}^{\otimes n})
=nI(A\rangle B)_{\rho_{\mc}}
=n\big[H(Y)_\omega-H(\widetilde A)_r\big],
\end{equation}
where the second equality follows from~\eqref{eq:cohinfo}.
\end{proof}

\subsection{A strict quantum-memory advantage}\label{sec:qubit}

We now exhibit a point in the static capacity region that efficient
instruments cannot attain, even with collective processing and the unit
protocols. The example is a two-qubit maximally correlated state.
Theorem~\ref{thm:quantum-memory} gives an explicit instrument
with nonorthogonal conditional states in the discarded system $F$
whose static capacity formula value strictly exceeds the efficient optimum
in~\eqref{eq:efficient}. By Proposition~\ref{prop:efficient}, this
advantage persists even when efficient instruments act collectively
on arbitrarily many copies. The comparison also includes instruments
with conditionally commuting discarded states, demonstrating that
allowing noncommuting states in $F$ can strictly improve the trade-off.

\begin{theorem}[Advantage of a discarded quantum system]
\label{thm:quantum-memory}
Consider the two-qubit maximally correlated state of~\eqref{eq:mcstate},
with basis labels $i\in\{0,1\}$ and coefficient matrix
\begin{equation}
r\coloneqq\begin{pmatrix}1/2&2/5\\2/5&1/2\end{pmatrix}.
\end{equation}
Thus $p=(1/2,1/2)$ and $G_{01}=G_{10}=4/5$ for $G$ in~\eqref{eq:G}.
For $\mu=1$ and every $1\leq\lambda\leq2$, its efficient support
from~\eqref{eq:efficient} is
\begin{equation}
\sS_1^{\eff}(p)=1-2h_2(9/10)\approx0.0620088128.
\end{equation}
Choose unit vectors $f_0,f_1$ with $\langle f_1|f_0\rangle=9/10$ and the
single-outcome instrument $V\ket i_A\coloneqq\ket i_S\ket{f_i}_F$. Its matrix
$\theta_t\coloneqq\frac12\bigl(\begin{smallmatrix}1&t\\t&1\end{smallmatrix}\bigr)$,
where $t\coloneqq9/10$, has the following value of the objective in~\eqref{eq:L}:
\begin{equation}
L_1^G(\theta_t)=1+h_2(19/20)-2h_2(43/50)
\approx0.1179193338>\sS_1^{\eff}(p).
\end{equation}
Its achievable primitive rates from~\eqref{eq:primitive-matrix-Q}--\eqref{eq:primitive-matrix} are
\begin{equation}
(C_0,Q_0,E_0)=\left(0,-\frac{h_2(43/50)-h_2(19/20)}2,
1-\frac{h_2(19/20)+h_2(43/50)}2\right).
\end{equation}
This improvement holds over all efficient instruments and all conditionally
commuting discarded-memory instruments as defined in Appendix~\ref{app:mc-commuting},
including their collective versions per copy and addition of the
unit-resource cone~\eqref{eq:unit}.
\end{theorem}
\begin{proof}
Appendix~\ref{app:mc-quantum-memory} proves the efficient optimum by a
concavity calculation and verifies the strict comparison with rational
entropy bounds. Proposition~\ref{prop:efficient} extends the efficient bound
to collective instruments.
\end{proof}
The system $F$ is Alice's local discarded output, not a separately supplied
resource. The witness uses one copy and therefore does not resolve the
unrestricted single-letterization question.

\subsection{An information-bottleneck formulation}

We now reformulate the matrix optimization as an information-bottleneck
problem for the classical--quantum source in~\eqref{eq:cqsource}.
In this formulation, a quantum memory is prepared from the classical
label, and the objective balances its information about the environment
against its information about the label itself. We identify the class
of memory states corresponding to canonical instruments and obtain
the exact identity in~\eqref{eq:restricted-exact}. Allowing arbitrary
conditional memory states then gives the upper bound
in~\eqref{eq:outer-onecopy}, whose strong additivity is established
in the following subsection.

\subsubsection{The associated classical--quantum source}
We use the source in~\eqref{eq:cqsource} to express the information kept
in a discarded memory. Distinguishing memories realizable by canonical
instruments from arbitrary quantum memories will give, respectively, an
exact identity and an outer bound.

Recall the source $\omega_{YE}=\sum_i p_i\proj i_Y\otimes e_{i,E}$
from~\eqref{eq:cqsource}, with $Y$ classical and every conditional state
$e_{i,E}$ pure. Write $\alpha\coloneqq1+\mu$. For a quantum memory prepared from $Y$,
define the extension
\begin{equation}
\omega_{YTE}\coloneqq\sum_i p_i\proj i_Y\otimes\tau_{i,T}\otimes e_{i,E}.
\label{eq:memory-source}
\end{equation}
The objective $\alpha I(T;E)_\omega-I(T;Y)_\omega$ rewards information about
$E$ while penalizing information retained about the input label. Quantum
information-bottleneck formulations and optimization methods appear
in~\cite{DHW19,HY23}. The restricted identity and the additivity argument below
are proved explicitly.

\begin{definition}[Classically flagged pure memories]
Let the memory system in~\eqref{eq:memory-source} be
$T\coloneqq XF$, where $X$ is a classical register and $F$ is a
quantum register. A family of conditional memory states
$\left(\tau_{i,T}\right)_i$ is classically flagged pure if
\begin{equation}
\tau_{i,T}=\tau_{i,XF}
=\sum_x a(x|i)\proj{x}_X\otimes\proj{f_{xi}}_F,
\qquad \sum_x a(x|i)=1.
\label{eq:flagged-pure}
\end{equation}
Here $a(x|i)\geq0$, every $\ket{f_{xi}}_F$ is a unit vector,
and the same orthonormal basis
$\left\{\ket{x}_X\right\}_x$ is used for all $i$.
Thus every nonzero classical block is rank one.
\end{definition}
This restriction does not include arbitrary mixed, noncommuting memory families.
It also does not require the full state $\tau_{i,XF}$ to be pure.

Define
\begin{align}
\sJ_\alpha^{\fp}(\omega)
&\coloneqq\sup_{\left(\tau_{i,XF}\right)_i\text{ as in }\eqref{eq:flagged-pure}}
\big[\alpha I(T;E)_\omega-I(T;Y)_\omega\big],\label{eq:J-restricted}\\
\sJ_\alpha(\omega)
&\coloneqq\sup_{\left(\tau_{i,T}\right)_i}
\big[\alpha I(T;E)_\omega-I(T;Y)_\omega\big],
\label{eq:J-unrestricted}
\end{align}
where the state in each objective is the extension~\eqref{eq:memory-source}
for the family being optimized. In the first line $T=XF$; the second line allows
arbitrary density operators on a memory $T$ of arbitrary finite dimension.
No bound on the unrestricted memory dimension is asserted.

\subsubsection{Exact identity for the restricted memory class}
The restricted memory family in~\eqref{eq:flagged-pure} is precisely
the family arising from canonical instruments. We now translate their
matrix score into a bottleneck objective and then obtain an upper bound
by enlarging the allowed memory family.

For a canonical instrument, the associated memory is $T=XF$.
The entropy expansion in Appendix~\ref{app:mc-bottleneck-identity} gives
\begin{equation}
H(Y)_\omega-(1+\mu)H(E)_\omega
+(1+\mu)I(T;E)_\omega-I(T;Y)_\omega
=I(S\rangle BX)_\sigma-\mu I(SX;E|B)_\sigma.
\end{equation}
Every flagged-pure family yields a canonical instrument and conversely. Therefore
\begin{equation}
\mathfrak S_\mu(G,p)
=H(Y)_\omega-(1+\mu)H(E)_\omega+\sJ_{1+\mu}^{\fp}(\omega).
\label{eq:restricted-exact}
\end{equation}
Relaxing the memory restriction gives
\begin{equation}
\mathfrak S_\mu(G,p)\leq B_\mu(\omega)
\coloneqq H(Y)_\omega-(1+\mu)H(E)_\omega+\sJ_{1+\mu}(\omega).
\label{eq:outer-onecopy}
\end{equation}

\subsection{Strong additivity of the unrestricted outer bound}
We now prove that the unrestricted memory functional in
\eqref{eq:J-unrestricted} is additive for products of classical--quantum
sources. Together with~\eqref{eq:outer-onecopy}, this gives an outer bound
on the operational region that requires no blocklength regularization.
\begin{theorem}[Unrestricted bottleneck additivity]
\label{thm:Jadd}
Let $(\omega_1)_{Y_1E_1}$ and $(\omega_2)_{Y_2E_2}$ be arbitrary
classical--quantum sources, with $Y_1$ and $Y_2$ classical. Their conditional
environmental states need not be pure. For every $\alpha\geq0$, the
unrestricted memory functional $\sJ_\alpha$ defined in~\eqref{eq:J-unrestricted},
using extensions of the form~\eqref{eq:memory-source}, satisfies
\begin{equation}
\sJ_\alpha(\omega_1\otimes\omega_2)
=\sJ_\alpha(\omega_1)+\sJ_\alpha(\omega_2).
\label{eq:J-additivity}
\end{equation}
\end{theorem}
\begin{proof}
See Appendix~\ref{app:mc-bottleneck-additivity}.
\end{proof}

\begin{corollary}[A single-letter outer bound for the operational region]
\label{cor:mc-outer}
Let $\rho_{\mc,AB}$ be as in~\eqref{eq:mcstate}, with associated source
$\omega_{YE}$ from~\eqref{eq:cqsource} and outer functional $B_\mu$
from~\eqref{eq:outer-onecopy}. For all $n\in\mathbb{N}$ and every
$\mu\geq0$ with $\mu\leq\lambda\leq\mu+1$,
\begin{equation}
\sS_{\lambda,\mu}(\rho_{\mc,AB}^{\otimes n})\leq nB_\mu(\omega).
\label{eq:block-outer}
\end{equation}
Here $\sS_{\lambda,\mu}$ is defined in~\eqref{eq:support}. Consequently,
every achievable rate triple in the sense of Section~\ref{sec:notation} satisfies
\begin{align}
C+2Q&\leq0,\\
\lambda C+(1+2\mu)Q+E&\leq B_\mu(\omega)
\quad\text{for all }\mu\geq0,\ \mu\leq\lambda\leq\mu+1.
\label{eq:capacity-outer}
\end{align}
Here single-letter means that the outer bound is expressed in terms of one copy
of the source and has no blocklength regularization. The unrestricted memory
supremum still has neither a proved dimension bound nor a closed-form evaluation.
\end{corollary}
\begin{proof}
Apply~\eqref{eq:restricted-exact} at blocklength $n$, relax to the
unrestricted functional, and apply Theorem~\ref{thm:Jadd}.
Entropy is additive on the product marginals
$\omega_Y^{\otimes n}$ and $\omega_E^{\otimes n}$, giving
\eqref{eq:block-outer}. Dividing by $n$ and taking the supremum
over $n\in\mathbb{N}$ as in~\eqref{eq:reg-support} gives
$\sS_{\lambda,\mu}^{\reg}(\rho_{\mc,AB})\leq B_\mu(\omega)$.
The supporting inequalities in~\eqref{eq:halfspace} of
Theorem~\ref{thm:formula} then imply~\eqref{eq:capacity-outer}
for the operational region characterized by
Theorem~\ref{thm:static}. These inequalities are preserved under
closure because each defines a closed half-space.
\end{proof}

The outer bound is nonnegative. To see this, choose the classical memory
$\tau_{i,T}=\proj i_T$, so $T$ is a separate copy of $Y$ in
\eqref{eq:memory-source}. Since the states $e_{i,E}$ are pure,
\begin{align}
I(T;E)_\omega&=H(E)_\omega-\sum_i p_iH(E)_{e_i}=H(E)_\omega,\\
I(T;Y)_\omega&=H(Y)_\omega.
\end{align}
Thus $\sJ_\alpha(\omega)\geq\alpha H(E)_\omega-H(Y)_\omega$, which implies
$B_\mu(\omega)\geq0$. The functional is finite: for every memory extension,
data processing gives
$I(T;E)_\omega\leq I(Y;E)_\omega\leq H(Y)_\omega$, while
$I(T;Y)_\omega\geq0$.

\subsection{The remaining memory-reduction question}
The additive outer bound need not equal the optimized instrument value.
We now state the remaining equality question explicitly and explain how
a positive answer would remove regularization from the maximally correlated
capacity formula.

The preceding results give
\begin{equation}
\mathfrak S_\mu(G,p)
\leq\sS_{\lambda,\mu}^{\reg}(\rho_{\mc,AB})
\leq B_\mu(\omega).
\label{eq:sandwich}
\end{equation}
The regularized static capacity formula $\sS_{\lambda,\mu}^{\reg}(\rho_{\mc,AB})$ is defined in~\eqref{eq:reg-support}.
The first inequality follows by using the one-copy ensemble independently;
the second follows from~\eqref{eq:block-outer}.

\begin{conjecture}[A sufficient memory-reduction statement]
\label{conj:memory}
For a classical--quantum source
$\omega_{YE}\coloneqq\sum_i p_i\proj i_Y\otimes e_{i,E}$, where
$\left(p_i\right)_i$ is a probability distribution and every $e_{i,E}$ is a pure
density operator, the functionals in~\eqref{eq:J-restricted}--\eqref{eq:J-unrestricted} satisfy,
for every $\alpha\geq1$,
\begin{equation}
\sJ_\alpha(\omega)=\sJ_\alpha^{\fp}(\omega).
\label{eq:memory-conjecture}
\end{equation}
The restricted functional uses the memory class in~\eqref{eq:flagged-pure}.
\end{conjecture}
If~\eqref{eq:memory-conjecture} holds for a source and every relevant $\alpha$,
both inequalities in~\eqref{eq:sandwich} become equalities. This would establish
single-letterization for the associated maximally correlated state. The conjecture
is sufficient; it is not claimed to be equivalent to single-letterization.
The outer relaxation could be strict even if a different argument establishes
additivity of the restricted functional.

The unrestricted additivity proof does not establish this conjecture:
averaging a joint memory can leave the flagged-pure class. For $\mu>0$,
appending a flag changes the objective by a weighted difference of
conditional mutual informations whose sign is not fixed in general.
Appendix~\ref{app:mc-obstruction} gives the full calculation
and compares the static expression with the dynamic Hadamard formula.

\section{Static trade-offs for states obtained by qubit dephasing}
\label{sec:dephasing}
This section applies the exact one-copy matrix optimization for
maximally correlated states in Theorem~\ref{thm:matrix-formula}
to states obtained by sending one share of a Schmidt-aligned pure
state of two qubits through a qubit dephasing channel. The optimization
in~\eqref{eq:matrix-functional} is over ensembles of density
operators whose average diagonal is fixed by the resource state,
as specified in~\eqref{eq:diagonal-constraint}. Here we give an exact one-copy optimization,
a threshold at which the regularized static capacity formula vanishes,
and the entanglement gain at vanishing quantum-communication cost.
Capacity slices illustrate these results. The final subsection explains
why arbitrary pure inputs require a separate analysis.

Let $\widetilde A$ be a qubit with computational basis $\{\ket0,\ket1\}$, and let
$Z\coloneqq\proj0-\proj1$. The dephasing channel is defined as
\begin{equation}
\mathcal Z_\gamma(\theta_{\widetilde A})\coloneqq
\frac{1+\gamma}{2}\theta_{\widetilde A}
+\frac{1-\gamma}{2}Z\theta_{\widetilde A}Z,
\qquad0\leq\gamma\leq1.
\label{eq:dephasing-channel}
\end{equation}
For matrix entropy calculations we identify the output qubit with $\widetilde A$.
In a physical realization, its output is Bob's system $B$. In the convention
$(1-\delta)\theta_{\widetilde A}+\delta Z\theta_{\widetilde A}Z$, the parameter is
$\gamma=|1-2\delta|$: a negative off-diagonal factor can be removed by a local
$Z$ on $B$. Fix unit environment vectors $\ket{e_0}_E$ and $\ket{e_1}_E$ such that
$\langle e_1|e_0\rangle=\gamma$. A Stinespring isometry of $\mathcal Z_\gamma$ is
\begin{equation}
U_{\widetilde A\to BE}\ket i_{\widetilde A}
\coloneqq\ket i_B\ket{e_i}_E,\qquad i\in\{0,1\}.
\label{eq:dephasing-isometry}
\end{equation}

\subsection{Schmidt-aligned pure inputs}
We first take the Schmidt basis of the input on the transmitted qubit
to agree with the dephasing basis. The resulting state is maximally
correlated, so Theorem~\ref{thm:matrix-formula} applies. We evaluate that
one-copy optimization explicitly in qubit coordinates and identify the
corresponding achievable rate points.

For $0\leq p\leq1$, define the pure input and its density operator by
\begin{equation}
\ket{\phi_p}_{A\widetilde A}\coloneqq
\sqrt p\ket0_A\ket0_{\widetilde A}+\sqrt{1-p}\ket1_A\ket1_{\widetilde A},
\qquad \phi^p_{A\widetilde A}\coloneqq\proj{\phi_p}_{A\widetilde A}.
\end{equation}
The resource state and a purification of it are
\begin{align}
\rho_{AB}^{p,\gamma}
&\coloneqq(\id_A\otimes\mathcal Z_\gamma)(\phi^p_{A\widetilde A})\\
&=p\proj{00}_{AB}+(1-p)\proj{11}_{AB}
+\gamma\sqrt{p(1-p)}
 (\ket{00}\bra{11}+\ket{11}\bra{00})_{AB},
\label{eq:dephasing-resource}\\
\ket{\psi_{p,\gamma}}_{ABE}
&\coloneqq(I_A\otimes U_{\widetilde A\to BE})\ket{\phi_p}_{A\widetilde A},
\qquad\psi_{ABE}^{p,\gamma}\coloneqq\proj{\psi_{p,\gamma}}_{ABE}.
\label{eq:dephasing-purification}
\end{align}
The state in~\eqref{eq:dephasing-resource} is maximally correlated, as defined
in~\eqref{eq:mcstate}. As such, Theorem~\ref{thm:matrix-formula} is applicable.

To specialize the matrix optimization in~\eqref{eq:matrix-functional}
of Theorem~\ref{thm:matrix-formula}, we first parametrize the density
operators on the qubit system $\widetilde A$. An arbitrary such operator has
the form
\begin{equation}
\theta_{\widetilde A}
=\begin{pmatrix}
q & z\\
\overline z & 1-q
\end{pmatrix},
\qquad
0\leq q\leq1,\quad z\in\mathbb C,\quad |z|^2\leq q(1-q).
\end{equation}
The last inequality follows from positive semidefiniteness, which
requires the determinant to be nonnegative.

Conjugation by a diagonal unitary can make the off-diagonal entry
real and nonnegative. This conjugation preserves the diagonal of
$\theta_{\widetilde A}$ and hence the ensemble constraint
in~\eqref{eq:diagonal-constraint}. It also preserves all three
entropies in the objective~\eqref{eq:L}: the entropy of $\theta_{\widetilde A}$
is unitarily invariant, its completely dephased state is unchanged,
and the qubit dephasing channel satisfies
\begin{equation}
\mathcal Z_\gamma(U\theta_{\widetilde A}U^\dagger)
=U\mathcal Z_\gamma(\theta_{\widetilde A})U^\dagger
\end{equation}
for every diagonal unitary $U$. Here the Schur channel $\cN_G$
in~\eqref{eq:L} is $\mathcal Z_\gamma$.
The phase can therefore be removed separately for each ensemble
element without changing the optimization.

For $0<q<1$, set $t\coloneqq |z|/\sqrt{q(1-q)}$.
Positive semidefiniteness gives $0\leq t\leq1$. At $q=0$ or $q=1$,
the off-diagonal entry vanishes, and we may set $t=0$.
Consequently, it suffices to consider
\begin{equation}
\theta_{\widetilde A}(q,t)\coloneqq
\begin{pmatrix}
q & t\sqrt{q(1-q)}\\
t\sqrt{q(1-q)} & 1-q
\end{pmatrix},
\qquad 0\leq q,t\leq1.
\label{eq:qubit-test-state}
\end{equation}
Define
\begin{align}
u(q,t)&\coloneqq\sqrt{(2q-1)^2+4q(1-q)t^2},\\
b(q,t)&\coloneqq h_2\!\left(\frac{1+u(q,t)}2\right),\label{eq:qubit-entropy}\\
F_\mu(q,t)&\coloneqq h_2(q)+\mu b(q,t)-(1+\mu)b(q,\gamma t).
\label{eq:qubit-score}
\end{align}
The eigenvalues of $\theta_{\widetilde A}(q,t)$ are $(1\pm u(q,t))/2$, obtained from
its trace and determinant. Thus $H(\widetilde A)_{\theta(q,t)}=b(q,t)$, and dephasing
replaces $t$ by $\gamma t$.

\begin{theorem}[One-copy support for Schmidt-aligned qubit dephasing]
\label{thm:qubit-onecopy}
Let $\rho_{AB}^{p,\gamma}$ be the state in~\eqref{eq:dephasing-resource},
with $0\leq p,\gamma\leq1$, and let $F_\mu$ be the function
in~\eqref{eq:qubit-score}. For the parameters $\mu\geq0$ and
$\mu\leq\lambda\leq\mu+1$, the static capacity formula in~\eqref{eq:support} satisfies
\begin{equation}
\sS_{\lambda,\mu}(\rho_{AB}^{p,\gamma})
=\max_{\substack{\left(w_x,q_x,t_x\right)_x\\\sum_xw_xq_x=p}}
\sum_xw_x F_\mu(q_x,t_x).
\label{eq:qubit-onecopy}
\end{equation}
The maximization is over probability weights $w_x$ and parameters
$0\leq q_x,t_x\leq1$. Two outcomes suffice.
\end{theorem}
\begin{proof}
If $p\in\{0,1\}$, the mean constraint forces $q_x=p$ for every
positive-probability outcome, and both sides vanish: the resource is a product
state, so Theorem~\ref{thm:extendible} applies. Assume henceforth $0<p<1$.
Apply Theorem~\ref{thm:matrix-formula} with $d=2$ and the Schur channel
$\cN_G=\mathcal Z_\gamma$. The diagonal entropy is
$H(\widetilde A)_{\Delta(\theta(q,t))}=h_2(q)$, where $\Delta$ is complete dephasing
in the computational basis. The two remaining entropies in~\eqref{eq:L}
are $H(\widetilde A)_{\theta(q,t)}=b(q,t)$ and
$H(\widetilde A)_{\mathcal Z_\gamma(\theta(q,t))}=b(q,\gamma t)$.
Consequently the objective is exactly~\eqref{eq:qubit-score}, and the
constraint in~\eqref{eq:diagonal-constraint} becomes $\sum_xw_xq_x=p$.
A diagonal unitary changes the phase of the off-diagonal entry but none of
these entropies or the constraint. Thus $t\geq0$ entails no loss.
The cardinality bound in Theorem~\ref{thm:matrix-formula}, with $d=2$,
gives $|X|\leq2$.
\end{proof}

Equivalently,~\eqref{eq:qubit-onecopy} is the upper concave envelope,
evaluated at $p$, of $q\mapsto\max_{0\leq t\leq1}F_\mu(q,t)$.
The two-outcome bound makes the mean constraint explicit:
\begin{equation}
\sS_{\lambda,\mu}(\rho^{p,\gamma}_{AB})
=\max_{\substack{0\leq q_0\leq p\leq q_1\leq1,\ q_0<q_1\\
0\leq t_0,t_1\leq1}}
\left[\frac{q_1-p}{q_1-q_0}F_\mu(q_0,t_0)
+\frac{p-q_0}{q_1-q_0}F_\mu(q_1,t_1)\right].
\label{eq:qubit-two-outcome}
\end{equation}
The weights are determined by their sum and the prescribed mean $p$.
An outcome of zero probability is allowed, so the formula includes
every single-outcome instrument. If two outcomes have the same
diagonal $q_0=q_1=p$, retaining the one with the larger score gives
at least their average score; hence excluding this redundant case
does not change the maximum.

For the symmetric input $p=1/2$, the identity
$F_\mu(q,t)=F_\mu(1-q,t)$ gives the further simplification
\begin{equation}
\sS_{\lambda,\mu}(\rho^{1/2,\gamma}_{AB})
=\max_{0\leq q\leq1/2,\ 0\leq t\leq1}F_\mu(q,t).
\label{eq:qubit-symmetric-support}
\end{equation}
Every ensemble average is at most the maximum on the right.
Conversely, the equiprobable pair $(q,t)$ and $(1-q,t)$ has the
required mean diagonal and attains $F_\mu(q,t)$. This identity
justifies using that symmetric family in the numerical exploration
described in Appendix~\ref{app:figures}.

Three outcomes suffice to preserve an arbitrary primitive rate point, by the
$d+1$ bound following~\eqref{eq:exact-onecopy-region}. The conditional state
with parameters $(q,t)$ gives the primitive point, in $(C,Q,E)$ coordinates,
\begin{equation}
\left(0,-\frac{b(q,\gamma t)-b(q,t)}2,
 h_2(q)-\frac{b(q,t)+b(q,\gamma t)}2\right).
\label{eq:qubit-point}
\end{equation}
This follows by substituting the same entropies into~\eqref{eq:primitive-matrix-Q}--\eqref{eq:primitive-matrix}.
Take convex combinations obeying the mean constraint and add the
unit-resource cone $\cU$ from~\eqref{eq:unit} to obtain the exact one-copy
information region. At blocklength $n$, the exact formula is
instead~\eqref{eq:block-formula} for density operators on $\widetilde A^n$.
Equation~\eqref{eq:qubit-onecopy} alone does not remove that optimization.

The reduced optimization also permits a modest improvement of the explicit
example in Theorem~\ref{thm:quantum-memory}. For $p=1/2$, $\gamma=4/5$,
and $\mu=1$, the single-outcome instrument with $t=223/250$ gives
\begin{equation}
F_1(1/2,223/250)
=1+h_2(473/500)-2h_2(1071/1250)
\approx0.1180348716.
\label{eq:refined-witness}
\end{equation}
This exceeds the value $0.1179193338\ldots$ of the simpler $t=9/10$
example. Appendix~\ref{app:mc-quantum-memory} certifies the strict
comparison by rational bounds. The ancillary code also explores
the two-outcome formula~\eqref{eq:qubit-two-outcome} and refines
the feasible instruments used for Figure~\ref{fig:dephasing-slices}.
These computations give achievable lower bounds; they do not certify
global optimality or remove regularization.

\Needspace{8\baselineskip}
\subsection{A complete entropy-contraction inequality}\label{sec:dephasing-contraction}
The one-copy formula in~\eqref{eq:qubit-onecopy} gives an achievable
lower bound on the regularized static capacity formula
in~\eqref{eq:reg-support}, since a one-copy instrument can be applied
independently to each copy. To obtain an upper bound, however, we
must also control the values attainable by arbitrary collective
instruments. For this purpose, we use the coherence representation
in~\eqref{eq:coherence-form}: the matrix objective compares the
coherence before and after dephasing. This representation allows us
to bound the collective objective by controlling how much coherence
remains after product dephasing. The required inequality must hold
at every blocklength and allow arbitrary correlations among the
input qubits.

Theorem~\ref{thm:contraction} establishes such an inequality in a
complete form: the same bound holds when the input qubits are
correlated with an arbitrary untouched auxiliary system, with a
constant independent of that system. The proof, given in
Appendix~\ref{app:contraction}, uses an entropy-production argument
of the type underlying logarithmic Sobolev inequalities~\cite{KT13}.
Complete entropic inequalities and generalized dephasing semigroups
are studied in Ref.~\cite[Section~6.1]{GR22}; related complete
pinching and coherence estimates appear in
Ref.~\cite[Section~4.2]{GZ26}. We give a self-contained proof of the
binary estimate needed here. The following subsection applies
this inequality to determine exactly when the static capacity
formula vanishes at every blocklength.

For $n\in\mathbb N$, write
$\widetilde A^n\coloneqq\widetilde A_1\cdots\widetilde A_n$ for the
$n$ input qubits. Let $R$ be an arbitrary auxiliary system, and let
$\Delta_i$ dephase $\widetilde A_i$, acting as the identity on the other
qubits and on $R$.
Set $\Delta_{[n]}\coloneqq\Delta_1\circ\cdots\circ\Delta_n$.
These channels commute because they act on distinct qubits, so the order
of composition is immaterial. For a density operator
$\theta_{\widetilde A^nR}$ define the density operator
$\widehat\theta_{\widetilde A^nR}\coloneqq\Delta_{[n]}(\theta_{\widetilde A^nR})$ and the functional
\begin{equation}
\mathcal C(\theta_{\widetilde A^nR})\coloneqq
H(\widetilde A^nR)_{\widehat\theta}-H(\widetilde A^nR)_\theta
=D(\theta_{\widetilde A^nR}\|\widehat\theta_{\widetilde A^nR}).
\label{eq:partial-coherence}
\end{equation}
For trivial $R$ this is the relative entropy of coherence $\Cr$ used
in~\eqref{eq:coherence-form}. The relative-entropy identity follows because
$\log_2\widehat\theta_{\widetilde A^nR}$ is block diagonal and
\begin{equation}
\Tr[\theta_{\widetilde A^nR}\log_2\widehat\theta_{\widetilde A^nR}]
=\Tr[\widehat\theta_{\widetilde A^nR}\log_2\widehat\theta_{\widetilde A^nR}].
\end{equation}
For singular states the logarithms are evaluated on the relevant support,
or equivalently the identity follows by a full-rank approximation.
In particular, $\widehat\theta$ is the average of the $2^n$ conjugations
by products of the qubit phase-flip operators, including the identity.
Thus $\theta\leq2^n\widehat\theta$, which implies
$\supp\theta\subseteq\supp\widehat\theta$ and makes the relative entropy finite.

\begin{theorem}[Entropy contraction under product binary dephasing]
\label{thm:contraction}
Let $n\in\mathbb{N}$, let $\widetilde A$ be a qubit, and let $R$ be an
arbitrary auxiliary system. For a density operator $\theta_{\widetilde A^nR}$
and the channel $\mathcal Z_\gamma$ in~\eqref{eq:dephasing-channel},
with $0\leq\gamma\leq1$, define
\begin{equation}
\xi_{\widetilde A^nR}\coloneqq
(\mathcal Z_\gamma^{\otimes n}\otimes\id_R)(\theta_{\widetilde A^nR}).
\end{equation}
Then the functional in~\eqref{eq:partial-coherence} satisfies
\begin{equation}
\mathcal C(\xi_{\widetilde A^nR})\leq\gamma^2\mathcal C(\theta_{\widetilde A^nR}).
\label{eq:coherence-contraction}
\end{equation}
\end{theorem}
\begin{proof}
See Appendix~\ref{app:contraction} for the complete proof.
\end{proof}

\subsection{When the static capacity formula vanishes}
The contraction inequality in Theorem~\ref{thm:contraction} determines exactly which admissible
parameter pairs $(\lambda,\mu)$ give a vanishing static capacity
formula at every blocklength, including when collective instruments
are allowed. The following theorem identifies this threshold.

\begin{theorem}[Exact threshold for vanishing of the static capacity formula]
\label{thm:threshold}
Let $\rho_{AB}^{p,\gamma}$ be the state in~\eqref{eq:dephasing-resource},
with $0<p<1$ and $0<\gamma<1$, and define
\begin{equation}
\mu_\gamma\coloneqq\frac{\gamma^2}{1-\gamma^2}.
\label{eq:dephasing-threshold-parameter}
\end{equation}
For all $n\in\mathbb{N}$ and every admissible $\lambda,\mu$ in
\eqref{eq:parameters}, the static capacity formula in~\eqref{eq:support} satisfies
\begin{equation}
\sS_{\lambda,\mu}((\rho_{AB}^{p,\gamma})^{\otimes n})=0
\quad\text{if and only if}\quad\mu\geq\mu_\gamma.
\label{eq:threshold}
\end{equation}
The same threshold holds for the regularized static capacity formula in~\eqref{eq:reg-support}.
\end{theorem}
\begin{proof}
Let $G_\gamma\coloneqq\left(\begin{smallmatrix}1&\gamma\\\gamma&1\end{smallmatrix}\right)$
be the correlation matrix of~\eqref{eq:dephasing-resource}.
For a density operator $\theta_{\widetilde A^n}$ in the blocklength-$n$ matrix
optimization~\eqref{eq:block-formula}, let
$\xi_{\widetilde A^n}\coloneqq\mathcal Z_\gamma^{\otimes n}(\theta_{\widetilde A^n})$.
The identity in~\eqref{eq:coherence-form}, with a trivial auxiliary in
\eqref{eq:partial-coherence}, gives
\begin{equation}
L_\mu^{G_\gamma^{\otimes n}}(\theta_{\widetilde A^n})
=(1+\mu)\mathcal C(\xi_{\widetilde A^n})-\mu\mathcal C(\theta_{\widetilde A^n}).
\end{equation}
Theorem~\ref{thm:contraction} implies
\begin{equation}
L_\mu^{G_\gamma^{\otimes n}}(\theta_{\widetilde A^n})
\leq[(1+\mu)\gamma^2-\mu]\mathcal C(\theta_{\widetilde A^n}).
\end{equation}
This is nonpositive when $\mu\geq\mu_\gamma$. The bound holds for every
matrix in every allowed ensemble. A diagonal state with the required
diagonal probabilities attains zero. This proves the zero statement for
arbitrary collective instruments.

For strict positivity below the threshold, use the single-outcome state
\begin{equation}
\theta_{z,\widetilde A}\coloneqq
\begin{pmatrix}p&z\\\bar z&1-p\end{pmatrix},
\qquad 0<|z|<\sqrt{p(1-p)},
\label{eq:small-coherence-state}
\end{equation}
whose diagonal is the required $(p,1-p)$.
Since $0<p<1$, entropy is analytic near the full-rank diagonal state
$\theta_{0,\widetilde A}=\diag(p,1-p)$, and
\begin{equation}
\Cr(\theta_{z,\widetilde A})
=H(\widetilde A)_{\theta_0}-H(\widetilde A)_{\theta_z}
=a_p|z|^2+O(|z|^4),
\label{eq:small-coherence-expansion}
\end{equation}
where
\begin{equation}
a_p\coloneqq
\begin{cases}
\dfrac{\ln p-\ln(1-p)}{(2p-1)\ln2},&p\ne1/2,\\[4pt]
\dfrac2{\ln2},&p=1/2.
\end{cases}
\end{equation}
The coefficient is positive for $0<p<1$. The expansion follows by
inserting the eigenvalues
$(1\pm\sqrt{(2p-1)^2+4|z|^2})/2$ into the binary entropy;
the value at $p=1/2$ is the continuous extension.
Dephasing replaces $z$ by $\gamma z$, so
\begin{equation}
L_\mu^{G_\gamma}(\theta_{z,\widetilde A})
=[(1+\mu)\gamma^2-\mu]a_p|z|^2+O(|z|^4)>0
\end{equation}
for sufficiently small $|z|$ when $\mu<\mu_\gamma$.
The product state $\theta_{z,\widetilde A}^{\otimes n}$ has the required product
diagonal, and its score is $nL_\mu^{G_\gamma}(\theta_{z,\widetilde A})$ by entropy
additivity. Thus the static capacity formula is positive at every blocklength and the
regularized formula is bounded below by the same positive one-copy score.
This proves both directions.
\end{proof}

The expansion~\eqref{eq:small-coherence-expansion} also proves that the
contraction factor in Theorem~\ref{thm:contraction} is optimal:
\begin{equation}
\lim_{z\to0,\,z\ne0}
\frac{\Cr(\mathcal Z_\gamma(\theta_{z,\widetilde A}))}
{\Cr(\theta_{z,\widetilde A})}=\gamma^2.
\end{equation}
Thus the constant cannot be decreased even for a single qubit without an
auxiliary system.

At the threshold, the two endpoint choices $\lambda=\mu_\gamma$ and
$\lambda=\mu_\gamma+1$ give the exact supporting inequalities
\begin{align}
\gamma^2 C+(1+\gamma^2)Q+(1-\gamma^2)E&\leq0,\label{eq:dephase-origin-first}\\
C+(1+\gamma^2)Q+(1-\gamma^2)E&\leq0.
\label{eq:dephase-origin-facets}
\end{align}
All intervening $\lambda$ give their convex combinations. Together with
$C+2Q\leq0$, these inequalities also imply the supporting inequalities with zero right-hand side
for $\mu>\mu_\gamma$: the normal at fixed $\lambda-\mu$ changes by a
nonnegative multiple of $(1,2,0)$.

\begin{corollary}[Entanglement gain at vanishing quantum-communication cost]
\label{cor:origin-slope}
Let $\rho_{AB}^{p,\gamma}$ be as in~\eqref{eq:dephasing-resource},
with $0<p<1$ and $0<\gamma<1$, and let $\cC_{\rm stat}$ be the region
of Section~\ref{sec:notation}. For a supplied quantum-communication
rate $Q_c\geq0$, define
\begin{equation}
E_{\max}(Q_c)\coloneqq
\sup\{E:(0,-Q_c,E)\in\cC_{\rm stat}(\rho_{AB}^{p,\gamma})\}.
\label{eq:dephasing-emax}
\end{equation}
Then $E_{\max}(0)=0$ and
\begin{equation}
\lim_{Q_c\downarrow0}\frac{E_{\max}(Q_c)}{Q_c}
=\frac{1+\gamma^2}{1-\gamma^2}.
\label{eq:origin-slope}
\end{equation}
\end{corollary}
\begin{proof}
Equation~\eqref{eq:dephase-origin-facets} gives
$E_{\max}(Q_c)\leq(1+\gamma^2)Q_c/(1-\gamma^2)$.
The zero rate point is achievable, so this also proves $E_{\max}(0)=0$.
For the state $\theta_{z,\widetilde A}$ in~\eqref{eq:small-coherence-state},
the primitive rates in~\eqref{eq:qubit-point} satisfy
\begin{align}
Q_c&=\tfrac12(1-\gamma^2)a_p|z|^2+O(|z|^4),\\
E&=\tfrac12(1+\gamma^2)a_p|z|^2+O(|z|^4),
\end{align}
by~\eqref{eq:small-coherence-expansion}.
As a function of $|z|$, the first rate is analytic and strictly increasing
for sufficiently small positive $|z|$, since its derivative is
$(1-\gamma^2)a_p|z|+O(|z|^3)>0$.
These achievable points therefore supply all sufficiently small positive
$Q_c$ and approach the displayed ratio. Combining this lower bound with
the upper bound proves the limit. No full additivity assumption is used.
\end{proof}

The equality $E_{\max}(0)=0$ does not contradict positive one-way distillable
entanglement. The usual distillation convention permits free forward classical
communication; the slice~\eqref{eq:dephasing-emax} accounts for its net rate.

At $\gamma=1$, the state is pure and the full region is
$(0,0,h_2(p))+\cU$ by Theorem~\ref{thm:flagged}. At $\gamma=0$, or at
$p\in\{0,1\}$, the state is separable and the region is $\cU$ by
Theorem~\ref{thm:extendible}. For $0<\gamma<1$ and $0<p<1$, define
\begin{align}
D_{p,\gamma}
&\coloneqq I(A\rangle B)_{\rho^{p,\gamma}}
=H(B)_{\psi^{p,\gamma}}-H(E)_{\psi^{p,\gamma}}\\
&=h_2(p)
-h_2\!\left(\frac{1+\sqrt{\gamma^2+(1-\gamma^2)(2p-1)^2}}2\right).
\end{align}
The support at $\mu=0$ is exactly $D_{p,\gamma}$ by~\eqref{eq:had-mu0}.
Consequently $Q+E\leq D_{p,\gamma}$ and $C+Q+E\leq D_{p,\gamma}$ are exact
supporting bounds. Full single-letterization of the interval
$0<\mu<\mu_\gamma$ remains open.

\subsection{Capacity slices for a dephased maximally entangled state}
We now combine the exact supporting inequalities just obtained with
achievable points from~\eqref{eq:qubit-point}. The resulting two-dimensional
slices illustrate both the determined boundary portions and the gap that
remains between the sampled achievable rates and the collective outer bound.

Figure~\ref{fig:dephasing-slices} considers $p=1/2$ and $\gamma=0.8$.
Set $D\coloneqq D_{1/2,\gamma}=1-h_2((1+\gamma)/2)$.
The coherent-information bound and~\eqref{eq:dephase-origin-first}--\eqref{eq:dephase-origin-facets} give
\begin{align}
E&\leq\min\!\left\{D,\frac{\gamma^2(-C)}{1-\gamma^2}\right\}
&&\text{when }Q=0,\ C\leq0,\label{eq:dephasing-C-outer}\\
E&\leq\min\!\left\{D-Q,\frac{(1+\gamma^2)(-Q)}{1-\gamma^2}\right\}
&&\text{when }C=0,\ Q\leq0.\label{eq:dephasing-Q-outer}
\end{align}
These are outer bounds on the full operational region, including collective
instruments. The achievable curves use a finite family of the canonical
instruments in~\eqref{eq:qubit-point}, followed by time sharing and the unit
protocols. Appendix~\ref{app:figures} gives the construction and numerical
parameters. The gap in the figure remains unresolved; the sampled curve is
not asserted to exhaust even the complete one-copy information region.
For sufficiently large communication consumption, the identity instrument
and the unit protocols attain the coherent-information bounds, so the two
curves coincide.

\begin{figure}[t]
\centering
\includegraphics[width=\linewidth]{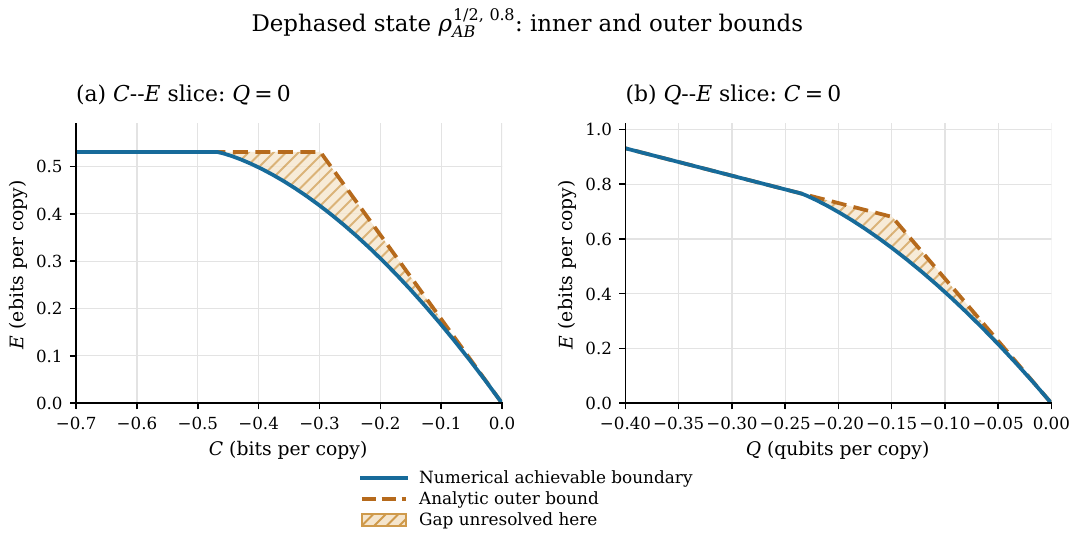}
\caption{Static trade-off bounds for the state $\rho_{AB}^{1/2,0.8}$
defined in~\eqref{eq:dephasing-resource}. The solid curves are achievable
by the finite construction in Appendix~\ref{app:figures}.
The dashed curves are the collective outer bounds in~\eqref{eq:dephasing-C-outer} and~\eqref{eq:dephasing-Q-outer}, with
$D=1-h_2(0.9)$. The shaded gap is unresolved. Only the portions with
nonnegative entanglement generation and nonpositive communication rates
are displayed. Points below a solid curve are achievable; points above a
dashed curve are excluded.}
\label{fig:dephasing-slices}
\end{figure}

\subsection{Arbitrary pure inputs need not be Schmidt aligned}
The preceding analysis used alignment of the Schmidt and dephasing bases.
We now explain why a general pure input need not yield a maximally
correlated state, and identify the conclusions that still follow from
the Hadamard-state results in Section~\ref{sec:hadamard}.

Let $\phi_{A\widetilde A}\coloneqq\proj\phi_{A\widetilde A}$ be an arbitrary pure state, with qubit input
marginal
\begin{equation}
\tau_{\widetilde A}\coloneqq\Tr_A\phi_{A\widetilde A}
=\begin{pmatrix}p&z\\\bar z&1-p\end{pmatrix}.
\end{equation}
Define its output state and a purification using~\eqref{eq:dephasing-isometry}:
\begin{equation}
\rho_{AB}\coloneqq(\id_A\otimes\mathcal Z_\gamma)(\phi_{A\widetilde A}),
\qquad
\ket\psi_{ABE}\coloneqq(I_A\otimes U_{\widetilde A\to BE})\ket\phi_{A\widetilde A},
\quad\psi_{ABE}\coloneqq\proj\psi_{ABE}.
\end{equation}
Expand
\begin{equation}
\ket\phi_{A\widetilde A}=\sqrt p\ket{a_0}_A\ket0_{\widetilde A}
+\sqrt{1-p}\ket{a_1}_A\ket1_{\widetilde A},
\end{equation}
with normalized $a_i$ when their probabilities are nonzero. Their overlap
need not vanish. Thus $\rho_{AB}$ need not be locally isometric to the
maximally correlated state~\eqref{eq:dephasing-resource}.
For an explicit example, take
$\ket\phi_{A\widetilde A}=\ket0_A\ket+_{\widetilde A}$ and $0<\gamma<1$.
The output $\rho_{AB}$ then has a rank-one marginal on $A$ and a rank-two
marginal on $B$, whose nonzero eigenvalues are $(1\pm\gamma)/2$.
Every maximally correlated state has equal marginal ranks, as is immediate
from~\eqref{eq:mcstate}; local isometries preserve these ranks. Hence this
output admits no maximally correlated realization under local isometries.
When $z=0$ and $0<p<1$, the two vectors are orthogonal and the Schmidt-aligned
analysis applies.

For every such input, $\rho_{AB}$ is still Hadamard generated. The marginal
states of its purification satisfy
$\psi_B=\mathcal Z_\gamma(\tau_{\widetilde A})$ and
$\psi_E=p\proj{e_0}_E+(1-p)\proj{e_1}_E$.
Their entropies are
\begin{align}
H(A)_\psi
&=h_2\!\left(\frac{1+\sqrt{(2p-1)^2+4|z|^2}}2\right),\\
H(B)_\psi
&=h_2\!\left(\frac{1+\sqrt{(2p-1)^2+4\gamma^2|z|^2}}2\right),\\
H(E)_\psi
&=h_2\!\left(\frac{1+\sqrt{\gamma^2+(1-\gamma^2)(2p-1)^2}}2\right).
\end{align}
The first follows from $H(A)_\psi=H(\widetilde A)_\phi=H(\widetilde A)_\tau$ and the
eigenvalues of $\tau_{\widetilde A}$; the second follows from the eigenvalues of
$\psi_B$. Represent the complementary state $\psi_E$ on a two-dimensional space
containing the span of $e_0,e_1$, adjoining a zero eigenvalue if necessary.
It has trace one and determinant $p(1-p)(1-\gamma^2)$, proving the third
formula, including the degenerate cases.
Hence the additive $\mu=0$ face is exactly $H(B)_\psi-H(E)_\psi$
by~\eqref{eq:had-mu0}. The identity instrument achieves the primitive point
\begin{equation}
\left(0,-\frac{H(A)_\psi+H(E)_\psi-H(B)_\psi}{2},
\frac{H(A)_\psi+H(B)_\psi-H(E)_\psi}{2}\right)
\label{eq:nonaligned-primitive}
\end{equation}
by the pure-state entropy identities in~\eqref{eq:CASR-C}--\eqref{eq:CASR-E}.
We do not extend the matrix formula or Theorem~\ref{thm:threshold} to
nonaligned inputs without an additional argument.

\section{Conclusions and future directions}

\subsection{Summary}
We have presented the direct static capacity theorem of Ref.~\cite{HW10}
through a common entropy
description of classical communication, quantum communication, and entanglement.
The direct construction combines instrument compression with quantum
side information~\cite{WHBH12} and state redistribution~\cite{DY08,YD09}
with the unit-resource protocols. The finite-block converse in
Theorem~\ref{thm:finite-converse} derives all three inequalities from one
instrument, for every choice of signs of the net rates. No-signalling accounts
for the supplied entanglement and message references, while continuity bounds
control the decoding error without introducing the dimension of a discarded
system. Theorem~\ref{thm:formula} then gives the exact supporting-hyperplane
description and identifies the additivity question needed for a single-letter
capacity formula.

The complete erased-state region in Theorem~\ref{thm:erased-region} follows
from subset-entropy bounds valid for arbitrary collective instruments,
using methods developed in Refs.~\cite{MW23,GHW22}.
Symmetrically extendible states and locally flagged pure-state mixtures give
two additional complete classes. For Hadamard states, the degradable-state
distillation theorem of Ref.~\cite[Proposition~4]{LDS18} determines
the coherent-information face, as stated in Theorem~\ref{thm:had-face}. The penalty
identity in~\eqref{eq:had-penalty} isolates the remaining optimization.
For maximally correlated states, the matrix reduction in
Theorem~\ref{thm:matrix-formula} preserves the discarded quantum memory, and
Corollary~\ref{cor:mc-outer} provides an additive outer bound. The strict example in
Theorem~\ref{thm:quantum-memory} shows why this memory cannot in general be
removed. For Schmidt-aligned qubit dephasing, the one-copy formula, the exact
threshold for its vanishing at all blocklengths, and the origin slope give
exact information about the region, while leaving the intermediate
trade-offs open.

\subsection{Directions for future research}\label{sec:future}
There are several interesting questions to consider going forward.
The first question is whether the matrix functional for maximally correlated
states tensorizes as in~\eqref{eq:target-tensorization}. Binary correlation
matrices provide a concrete starting point. The reduction proposed in
Conjecture~\ref{conj:memory}, together with the additivity theorem already
proved in Theorem~\ref{thm:Jadd}, would make both inequalities
in~\eqref{eq:sandwich} equalities and establish single-letterization for
maximally correlated states. For dephased states, the unresolved optimization
concerns $0<\mu<\mu_\gamma$, where $\mu_\gamma$ is defined
in~\eqref{eq:dephasing-threshold-parameter}.
The contraction inequality in~\eqref{eq:coherence-contraction} determines the
threshold for vanishing of the static capacity formula; a bound that also
retains the dependence on the fixed input diagonal could help determine
its positive values.

A second question concerns the general Hadamard penalty in~\eqref{eq:had-gamma}.
Each of its terms is nonnegative, but this alone does not imply additivity of
their optimized sum. A proof would need to relate a collective discarded
quantum system to admissible one-copy instruments. Conversely, a strict
collective advantage would identify a limitation of the single-copy
description even within this degradable family. Nonaligned inputs to qubit
dephasing are another useful test, since they retain the Hadamard structure
without necessarily being maximally correlated.

The possibility of nonadditive static trade-off curves deserves a separate
investigation. In the dynamic setting, Ref.~\cite{ZZS17}
constructed a channel with additive classical capacity but superadditive
classical capacity under limited entanglement assistance. Ref.~\cite{ZZHS19} established further separations between
additivity of individual capacities and additivity of trade-off regions.
A corresponding static question is whether a bipartite state $\rho_{AB}$
can have an additive coherent-information optimization while satisfying
\begin{equation}
\sS_{\lambda,\mu}(\rho^{\otimes n})
>n\sS_{\lambda,\mu}(\rho)
\label{eq:static-nonadditivity-question}
\end{equation}
for some $n\in\mathbb N$ with $n\geq2$ and parameters
in~\eqref{eq:parameters} with $\mu>0$.
Hadamard states provide a concrete setting because their coherent-information
face is already additive by Theorem~\ref{thm:had-face}. For the dephased
states covered by Theorem~\ref{thm:threshold}, a possible separation is
further restricted to $0<\mu<\mu_\gamma$.
Such a strict inequality would give a collective improvement of a supporting
value of the static region. The one-copy quantum-memory separation in
Theorem~\ref{thm:quantum-memory} does not establish this collective advantage.

Another direction is to determine the static capacity region of finite-energy
two-mode bosonic Gaussian states. After removing first moments by local
displacements, their covariance matrices can be brought to standard form by
local symplectic transformations~\cite{Simon00}. The bosonic
loss-channel and amplifier-channel results~\cite{WHG12a,WHG12b,QW17} suggest
starting with states obtained by sending one share of a two-mode squeezed
vacuum state through either channel. For a fixed infinite-dimensional resource state, the communication
and entanglement registers remain finite-dimensional at each
blocklength. The conditional-entropy continuity bounds used in
the converse therefore remain applicable, because they depend
only on the dimensions of these registers~\cite{Winter16}.
Extending the full static capacity theorem to this setting
calls for an appropriate treatment of the information quantities
and a justification of the direct coding argument. An extension
of the direct construction restricted to Gaussian instruments would give an
achievable region; optimality would require a converse allowing arbitrary
collective instruments. Establishing when Gaussian instruments suffice, and
when collective non-Gaussian operations improve the trade-off, are natural
parts of this problem.

The exactly solved families also suggest extensions. Products of erased states
with a common erasure probability obey the same entropy argument, whereas
different probabilities require a more general interpolation inequality.
Locally accessible flags may allow further mixtures to be evaluated explicitly,
provided that the cost of making a flag available to both parties is accounted
for. Finally, quantitative versions of the finite-block converse may help
determine convergence rates and stronger error bounds. Such refinements would
require additional estimates beyond the continuity argument used here.

\section*{Acknowledgements}

I acknowledge support from the Cornell University School of Electrical and Computer Engineering.

\section*{Statement on AI-assisted preparation}

ChatGPT Pro 6 (Astra) and Codex (OpenAI) were used extensively throughout
the development of this paper, including for mathematical exploration,
development and checking of proof arguments, drafting and revision of
the text, verification of bibliographic information, and preparation
of numerical and figure-generation code. The author directed this
process, read and revised successive drafts, requested changes to
the mathematical arguments and presentation, and determined which
suggestions to incorporate.

The retained records document an initial drafting cycle and at least
six subsequent substantial manuscript-revision cycles. Across
these cycles, twenty-two rounds of internal AI-assisted mathematical
and editorial review produced at least seventy-nine written reports.
The manuscript was revised in response to these reports and to the
author's comments. Six revision checklists contain seventy-six
grouped author-requested items, covering both implementation of
changes and verification of revisions already incorporated by the
author. This checklist count includes revisited items and does not
represent seventy-six distinct textual edits; numerous individual
comments and reviewer-requested corrections are additional.

Review tasks were distributed among AI agents with specified scopes.
They examined the coding theorem and finite-block entropic converse;
the linear-programming derivation of the static capacity formula;
the erased-state subset-entropy bounds and other exact state families;
the Hadamard and maximally correlated reductions, cardinality bounds,
and additivity claims; and the quantum-memory example and complete
dephasing-contraction argument. Other assignments checked numerical
certificates, figure data and reproducibility, bibliographic information,
notation, and cross-references. Some reviewers were assigned different
topics in later rounds, while other rounds used fresh review contexts.
The report count records completed review assignments, rather than
distinct agents; the retained records do not establish a cumulative
count of distinct reviewer agents.

In the most recent revision cycle, fourteen freshly initialized lead
reviewer contexts produced fourteen scoped reports across four rounds.
These reviewers received the manuscript version assigned to their
round without earlier referee reports. Their assignments covered the
core capacity theorem and exact state families, the maximally correlated
and dephasing results, and scholarly presentation. Subsidiary checks
were incorporated into the lead reports and are not counted as
additional reports or lead reviewer contexts.

The reviews formed part of the internal preparation of the
manuscript and do not constitute external peer review. The author
takes responsibility for the mathematical claims, proofs,
citations, and presentation in the final manuscript.

\bibliographystyle{alphaurl}
\bibliography{static_capacity}

\appendix

\section{Detailed achievability proof}\label{app:direct}
This appendix proves achievability of the rate point in
\eqref{eq:CASR-C}--\eqref{eq:CASR-E} by composing instrument compression
with quantum side information and state redistribution. It accounts for
fixed resource amounts, removes the shared randomness used in the
intermediate simulation, and then treats regularization. We begin with
the two coding theorems used in this construction.
For the instrument state~\eqref{eq:instrument-state}, instrument compression
with quantum side information simulates the instrument on many independent
copies, with its quantum outputs $SF$ at Alice and copies of $X$ at both
parties, using classical communication rate
\begin{equation}
I(X;E|B)_{\sigma}
\end{equation}
and shared randomness rate $H(X|BE)_{\sigma}$~\cite[Theorem~12, end of
Section~5.3, and Section~5.6.1]{WHBH12}. More precisely, define the state
\begin{equation}
\sigma_{X_AX_BSFBE}\coloneqq
\sum_xp_x\proj x_{X_A}\otimes\proj x_{X_B}\otimes\psi^x_{SFBE},
\label{eq:compression-target}
\end{equation}
where the probabilities and conditional states are those of
\eqref{eq:instrument-state}. Alice holds $X_ASF$ and Bob holds $X_BB$.
The simulation approximates tensor powers of this entire state, including
the reference $E$, in normalized trace distance. Its error tends to zero.
Arbitrarily small positive rate slack is permitted. This is the
feedback version of instrument compression: Alice also learns the outcome.
The nonnegativity of $H(X|BE)_{\sigma}$ follows because $X$ is classical.

For a pure state $\chi_{SFBE}$, state redistribution transfers $S$ from Alice,
who also holds $F$, to Bob, who holds $B$, at quantum communication rate
\begin{equation}
Q_{\rm red}\coloneqq\tfrac12 I(S;E|B)_{\chi}
\end{equation}
and net entanglement consumption rate
\begin{equation}
E_{\rm red}\coloneqq\tfrac12[I(S;F)_{\chi}-I(S;B)_{\chi}].
\end{equation}
A negative consumption rate means entanglement generation. The transferred
state and the generated or returned ideal entanglement are preserved jointly
\cite{DY08,YD09}. These formulas are equivalently
$Q_{\rm red}\geq\tfrac12 I(S;E|B)_{\chi}$ and
$Q_{\rm red}+E_{\rm red}\geq H(S|B)_{\chi}$ at their common corner.

Apply instrument compression and then redistribution conditional on $X$.
The latter conditioning has a concrete implementation. On a typical outcome
sequence, locally permute equal labels into blocks, use the redistribution
code for each normalized state $\psi^x_{SFBE}$ of~\eqref{eq:instrument-state}, and undo the permutations. Outcome frequencies
converge to $p_x$; the sum of the costs divided by blocklength therefore converges
to the average of the conditional costs. Reserve a small rate slack for deviations
and use an arbitrary output on atypical sequences. Their probability tends to
zero. On typical sequences, each positive-probability outcome occurs a
number of times tending to infinity. Choose each redistribution code with
error uniformly small for all sufficiently large blocklengths. There are
only finitely many such outcomes, so the sum of their coding errors tends
to zero. Contractivity of normalized trace distance under the subsequent
conditional redistribution map shows that the instrument-simulation error adds
to these coding errors.

These conditional protocols can be implemented with fixed resource amounts.
Choose a fixed order for the blocks with the same outcome. Supply enough initial ebits for
the largest cumulative deficit among all typical sequences in that order,
including temporary entanglement consumption inside each conditional redistribution
code. The standard redistribution constructions have finite consumed and
returned entanglement rates, so this reserve is $O(n)$. Pad each communication
register to its largest typical dimension. At the
end retain a fixed number of ebits no larger than the smallest typical final
yield, discarding any surplus. Frequencies within a sufficiently small typical
window change all total net costs by at most an arbitrarily small multiple
of $n$. On an atypical sequence output an arbitrary state on these same fixed
registers. Thus every random-seed realization has the same communication
dimensions, supplied entanglement amount, and tested output dimension. The initial
entanglement reserve is counted both in the supply and in the returned yield;
their difference has the averaged conditional net rate below.

The resulting net point, already stated in~\eqref{eq:CASR-C}--\eqref{eq:CASR-E}, is
\begin{align}
C_0&\coloneqq-I(X;E|B)_{\sigma},\\
Q_0&\coloneqq-\tfrac12I(S;E|BX)_{\sigma},\\
E_0&\coloneqq\tfrac12[I(S;B|X)_{\sigma}-I(S;F|X)_{\sigma}].
\label{eq:CASR-appendix}
\end{align}
The two conditional mutual information expressions
$I(S;E|BX)_{\sigma}$ and $I(S;E|FX)_{\sigma}$ agree by conditional purity.

There is no free shared randomness in the operational definition. It can be
removed \emph{after} this composition. Include any returned entanglement in
the final pure target $\Phi$. If $K$ is the initially shared random seed and
$\omega_k$ is the final marginal state on the retained entanglement
registers for seed $k$, then
\begin{equation}
\sum_kp_k\Tr[\Phi\omega_k]\geq1-\eta
\end{equation}
implies that some $k$ satisfies $\Tr[\Phi\omega_k]\geq1-\eta$. Fix that seed in
both code descriptions. The resulting deterministic protocol has the same
communication and entanglement amounts, and normalized trace-distance error at most
$\sqrt\eta$. Intermediate instrument simulation need no longer be accurate for
this fixed seed; only the final entanglement target is required. If the primitive
point has $E_0\leq0$, then all three net rates are nonpositive, because
$C_0$ and $Q_0$ in~\eqref{eq:CASR-C}--\eqref{eq:CASR-Q} are negatives
of nonnegative conditional mutual informations. This point is already
achievable by resource disposal. This is the elementary derandomization step
underlying the resource formulation in Refs.~\cite{HW10,DHW08}; see also
Ref.~\cite[Section~5.6.1, following Eq.~(81)]{WHBH12}.

We next verify the three facet values of~\eqref{eq:CASR-C}--\eqref{eq:CASR-E}. Purity gives
\begin{align}
I(S;E|FX)_{\sigma}
&=H(SF|X)_{\sigma}+H(EF|X)_{\sigma}-H(F|X)_{\sigma}-H(SEF|X)_{\sigma}\\
&=H(BE|X)_{\sigma}+H(SB|X)_{\sigma}-H(SBE|X)_{\sigma}-H(B|X)_{\sigma}\\
&=I(S;E|BX)_{\sigma}.
\end{align}
Expanding and cancelling terms yields
\begin{align}
2(Q_0+E_0)
&=-I(S;E|FX)_{\sigma}+I(S;B|X)_{\sigma}-I(S;F|X)_{\sigma}\\
&=-H(EF|X)_{\sigma}+H(SEF|X)_{\sigma}+H(B|X)_{\sigma}-H(SB|X)_{\sigma}\\
&=2H(B|X)_{\sigma}-2H(SB|X)_{\sigma}=2I(S\rangle BX)_{\sigma}.
\end{align}
The chain rule then gives
\begin{align}
C_0+2Q_0&=-I(SX;E|B)_{\sigma},\\
Q_0+E_0&=I(S\rangle BX)_{\sigma},\\
C_0+Q_0+E_0&=I(S\rangle BX)_{\sigma}-I(X;E|B)_{\sigma}.
\end{align}
Consequently the polyhedron~\eqref{eq:facet1}--\eqref{eq:facet3} is exactly
$(C_0,Q_0,E_0)+\cU$. All its points are achievable by the three unit protocols.
If a composition temporarily needs ebits that it later returns, they are
included among the consumed and generated ideal resources in the catalytic
code definition.
More explicitly, first run the primitive protocol and then apply the unit
protocols to its entanglement output and independent message inputs.
Contractivity of normalized trace distance under these subsequent channels
preserves the approximation to the joint communication and entanglement
target. For each fixed point of the translated cone, all resource amounts,
including a reserve for temporary entanglement consumption, have finite rates.

Finally fix an instrument on $k$ copies, regard $\rho^{\otimes k}$ as one
source symbol, and run the preceding direct proof on many independent symbols.
Its rates per original copy are divided by $k$. Taking all $k$ and then the
closure proves the direct inclusion in~\eqref{eq:regularized-region}.

\section{Proof of the erased-state capacity theorem}\label{app:erased}\label{sec:subset}
We first establish the entropy inequalities needed to analyze arbitrary
collective instruments on an erased state. All logarithms have base two.
For a classical--quantum state
\begin{equation}
\sigma_{XU}\coloneqq\sum_x p_x\proj{x}_X\otimes\sigma_U^x,
\end{equation}
where $\left(p_x\right)_x$ is a probability distribution and each $\sigma_U^x$ is a
normalized state, conditional entropy satisfies
\begin{equation}
H(U|X)_\sigma=\sum_x p_x H(U)_{\sigma^x}.
\end{equation}
In particular, $H(U|X)_\sigma\geq0$. Conditional entropy given a quantum
system can be negative.

\subsection{Strong subadditivity and weak monotonicity}
Let $\theta_{UVW}$ be an arbitrary tripartite state. Strong subadditivity
states that
\begin{equation}
I(U;V|W)_\theta
=H(U|W)_\theta-H(U|VW)_\theta\geq0.
\label{eq:SSA}
\end{equation}
To obtain weak monotonicity, let $\chi_{FUVW}$ be a purification of an
arbitrary state $\chi_{FUV}$ on the disjoint systems $F,U,V$. We claim that
\begin{equation}
H(F|U)_\chi+H(F|V)_\chi\geq0.
\label{eq:WM}
\end{equation}
Purity gives $H(F|U)_\chi=-H(F|VW)_\chi$, and consequently
\begin{align}
H(F|U)_\chi+H(F|V)_\chi
&=H(F|V)_\chi-H(F|VW)_\chi\\
&=I(F;W|V)_\chi\\
&\geq0.
\end{align}
The last inequality follows from~\eqref{eq:SSA} applied to the appropriate marginal of
$\chi_{FUVW}$.

Now let
\begin{equation}
\kappa_{XFUV}\coloneqq\sum_x p_x\proj{x}_X\otimes\kappa_{FUV}^x
\end{equation}
be a state that is classical on $X$. Applying~\eqref{eq:WM} to each
conditional state $\kappa_{FUV}^x$ and averaging gives
\begin{equation}
H(F|UX)_\kappa+H(F|VX)_\kappa\geq0.
\label{eq:WM-classical}
\end{equation}
The classical nature of $X$ is essential for this particular conditional form.

\subsection{Entropy of orthogonal blocks}
Let $\left(q_j\right)_j$ be a probability distribution and define its classical state
$q_Y\coloneqq\sum_j q_j\proj{j}_Y$. Let $\left(\xi_D^j\right)_j$ be normalized
states supported on mutually orthogonal subspaces of $D$, and define
\begin{equation}
\xi_D\coloneqq\bigoplus_j q_j\xi_D^j.
\end{equation}
If $\left(r_{jk}\right)_k$ are the eigenvalues of $\xi_D^j$, then the eigenvalues of
$\xi_D$ are $q_jr_{jk}$. Thus
\begin{align}
H(D)_\xi
&=-\sum_{j,k}q_jr_{jk}\log_2(q_jr_{jk})\\
&=-\sum_jq_j\log_2 q_j\sum_kr_{jk}
  -\sum_jq_j\sum_kr_{jk}\log_2 r_{jk}\\
&=H(Y)_q+\sum_jq_jH(D)_{\xi^j}.
\label{eq:blocks}
\end{align}
The normalization $\sum_k r_{jk}=1$ gives the last line. Terms with zero
probability or zero eigenvalue are interpreted using $0\log_2 0=0$.

\subsection{Bernoulli-subset entropy inequalities}
We use the subset-entropy approach of Mamindlapally and Winter~\cite{MW23}.
Their Lemmas~3 and~4 credit earlier inequalities and arguments of Grassl,
Huber, and Winter~\cite{GHW22}. These fixed-cardinality inequalities motivate
the Bernoulli averages below. We give the full proof needed here, including
an arbitrary quantum system $F$.

Let $A^n\coloneqq A_1\cdots A_n$. For $J\subseteq[n]$, write
$A_J\coloneqq\bigotimes_{j\in J}A_j$ and let $J^c\coloneqq[n]\setminus J$.
The empty tensor product is a one-dimensional system. Let $J_r$ be a random
subset in which each index is included independently with probability $r$.
Fix a state $\omega_{XFA^n}$ that is classical on $X$, and define
\begin{equation}
\begin{aligned}
f(r)&\coloneqq\EE_{J_r}H(A_{J_r}|X)_\omega, &
g(r)&\coloneqq\EE_{J_r}H(FA_{J_r}|X)_\omega,\\
f_F&\coloneqq H(F|X)_\omega, &
H_\Sigma&\coloneqq\sum_{i=1}^nH(A_i)_\omega.
\end{aligned}
\label{eq:subset-definitions}
\end{equation}
In particular, $f_F\geq0$ because $X$ is classical. Written explicitly,
\begin{equation}
f(r)=\sum_{J\subseteq[n]}r^{|J|}(1-r)^{n-|J|}H(A_J|X)_\omega,
\end{equation}
and similarly for $g$. These functions are polynomials in $r$; no
differentiability assumption on eigenvalues is needed below.

\begin{lemma}[Subset concavity and an upper increment bound]\label{lem:subset-concavity}
Let $\omega_{XFA^n}$ be a state classical on $X$, and let $f,g,H_\Sigma$
be as in~\eqref{eq:subset-definitions}. The functions $f$ and $g$ are
concave on $[0,1]$. Moreover, for $r\in[0,1]$,
\begin{equation}
\begin{aligned}
f'(r)&\leq H_\Sigma,\\
f(v)-f(u)&\leq(v-u)H_\Sigma\qquad(0\leq u\leq v\leq1).
\end{aligned}
\label{eq:Lipschitz}
\end{equation}
This is an upper increment bound; it does not assert that $f$ is increasing.
\end{lemma}
\begin{proof}
For a real-valued set function $a(J)$, define its multilinear extension by
\begin{equation}
\mathcal Q(t_1,\ldots,t_n)
\coloneqq\sum_{J\subseteq[n]}
\prod_{i\in J}t_i\prod_{i\notin J}(1-t_i)a(J).
\end{equation}
Equivalently, let $K\subseteq[n]$ be a random subset in which index $k$
is included independently with probability $t_k$. Then
$\mathcal Q(t_1,\ldots,t_n)=\EE[a(K)]$. Fix an index $i$ and first
choose only the other indices, writing $J\coloneqq K\setminus\{i\}$.
Conditioned on $J$, the full subset is $J\cup\{i\}$ with probability
$t_i$ and $J$ with probability $1-t_i$. The law of total expectation gives
\begin{equation}
\mathcal Q(t_1,\ldots,t_n)
=\EE_{J\subseteq[n]\setminus\{i\}}
\bigl[t_i a(J\cup\{i\})+(1-t_i)a(J)\bigr].
\label{eq:subset-condition-i}
\end{equation}
Here each remaining index $k$ is included independently with probability
$t_k$. In particular, the distribution of $J$ does not depend on $t_i$.
Differentiation therefore acts only on the two displayed coefficients:
\begin{equation}
\partial_i\mathcal Q
=\EE_{J\subseteq[n]\setminus\{i\}}
\bigl[a(J\cup\{i\})-a(J)\bigr],
\label{eq:subset-first-derivative}
\end{equation}
where $\partial_i \equiv \frac{\partial}{\partial t_i}$.
This derivative is the average change caused by including index $i$.

For $j\ne i$, the expression inside the expectation in
\eqref{eq:subset-first-derivative} has no explicit $t_j$, but the
distribution of its random subset does depend on $t_j$. To display that
dependence, choose a new subset $J\subseteq[n]\setminus\{i,j\}$ of the
remaining indices and separate the cases in which $j$ is present or absent:
\begin{equation}
\begin{aligned}
\partial_i\mathcal Q
=\EE_{J\subseteq[n]\setminus\{i,j\}}\bigl[
&t_j\big(a(J\cup\{i,j\})-a(J\cup\{j\})\big)\\
&+(1-t_j)\big(a(J\cup\{i\})-a(J)\big)\bigr].
\end{aligned}
\label{eq:subset-condition-j}
\end{equation}
The first difference is the change from adding $i$ when $j$ is already
present; the second is the change when $j$ is absent. The distribution of
this smaller subset $J$ depends on neither $t_i$ nor $t_j$.
Differentiating~\eqref{eq:subset-condition-j} with respect to $t_j$ gives
\begin{equation}
\partial_j\partial_i\mathcal Q
=\EE_{J\subseteq[n]\setminus\{i,j\}}
\bigl[a(J\cup\{i,j\})-a(J\cup\{i\})-a(J\cup\{j\})+a(J)\bigr].
\label{eq:subset-mixed-derivative}
\end{equation}
The order of the two derivatives can be exchanged because $\mathcal Q$
is a polynomial.
For the choice $a(J)\coloneqq H(FA_J|X)_\omega$, the expression in
square brackets in~\eqref{eq:subset-mixed-derivative} is
\begin{align}
&H(FA_JA_iA_j|X)_\omega-H(FA_JA_i|X)_\omega-H(FA_JA_j|X)_\omega+H(FA_J|X)_\omega\nonumber \\
&\quad=H(A_i|FA_JA_jX)_\omega-H(A_i|FA_JX)_\omega\\
&\quad=-I(A_i;A_j|FA_JX)_\omega\\
&\quad\leq0.
\end{align}
The last inequality follows from strong subadditivity in~\eqref{eq:SSA}:
including $A_j$ in the conditioning cannot increase the entropy increment
from adding $A_i$. Also, $\partial_i^2\mathcal Q=0$ because $\mathcal Q$
is affine in each individual variable. Since $g(r)=\mathcal Q(r,\ldots,r)$,
the chain rule for differentiation gives
\begin{equation}
g''(r)=\left.\sum_{i,j=1}^n\partial_i\partial_j\mathcal Q
\right|_{t_1=\cdots=t_n=r}.
\label{eq:subset-diagonal-second-derivative}
\end{equation}
The terms with $i=j$ vanish, and the remaining terms give
\begin{equation}
g''(r)=-\sum_{i\ne j}
\EE_{J\subseteq[n]\setminus\{i,j\}}
I(A_i;A_j|FA_JX)_\omega\leq0.
\end{equation}
Here every index in the expectation is included with probability $r$,
and the sum is over ordered pairs $i\ne j$.
Taking $F$ to be a one-dimensional system proves $f''(r)\leq0$ by the same
argument. Thus both functions are concave.

The first-derivative formula likewise gives
\begin{align}
f'(r)
&=\sum_{i=1}^n\EE_{J\subseteq[n]\setminus\{i\}}
\bigl[H(A_iA_J|X)_\omega-H(A_J|X)_\omega\bigr]\\
&=\sum_{i=1}^n\EE_{J\subseteq[n]\setminus\{i\}}
H(A_i|A_JX)_\omega\\
&\leq\sum_{i=1}^n H(A_i)_\omega\\
&=H_\Sigma.
\end{align}
Indeed,~\eqref{eq:SSA} with trivial initial conditioning gives
$H(A_i|A_JX)_\omega\leq H(A_i)_\omega$. Integrating this derivative bound
from $u$ to $v$ proves~\eqref{eq:Lipschitz}. All endpoint statements follow
by continuity of the polynomials.
\end{proof}

\begin{lemma}[The midpoint inequality]\label{lem:subset-midpoint}
For every state $\omega_{XFA^n}$ that is classical on $X$, the functions
in~\eqref{eq:subset-definitions} satisfy
\begin{equation}
g(1/2)\geq f(1/2).
\label{eq:midpoint}
\end{equation}
\end{lemma}
\begin{proof}
The definitions give
\begin{equation}
g(1/2)-f(1/2)=2^{-n}\sum_{J\subseteq[n]}H(F|A_JX)_\omega.
\end{equation}
The complementation map $J\mapsto J^c$ permutes all subsets and preserves
their uniform weights. Consequently,
\begin{align}
2[g(1/2)-f(1/2)]
&=2^{-n}\sum_{J\subseteq[n]}
\bigl[H(F|A_JX)_\omega+H(F|A_{J^c}X)_\omega\bigr]\\
&\geq0.
\end{align}
For each $J$, the systems $A_J$ and $A_{J^c}$ are disjoint, so the summand
is nonnegative by~\eqref{eq:WM-classical}, applied to the appropriate
marginal of $\omega_{XFA^n}$. Dividing by two proves the claim.
\end{proof}

\begin{lemma}[The key interpolation bound]\label{lem:subset-interpolation}
Let $\omega_{XFA^n}$ be a state classical on $X$, and let $f,g,f_F,H_\Sigma$
be the quantities in~\eqref{eq:subset-definitions}. For $0\leq p\leq1/2$,
\begin{equation}
g(p)\geq2p f(1-p)+(1-2p)f_F-p(1-2p)H_\Sigma.
\label{eq:key}
\end{equation}
\end{lemma}
\begin{proof}
By Lemma~\ref{lem:subset-concavity}, $g$ is concave on $[0,1/2]$.
Since $g(0)=f_F$, this gives
\begin{align}
g(p)&\geq(1-2p)g(0)+2pg(1/2)\\
&\geq(1-2p)f_F+2pf(1/2),
\end{align}
where the second inequality follows from Lemma~\ref{lem:subset-midpoint}.
The increment bound in~\eqref{eq:Lipschitz} gives
\begin{align}
f(1-p)-f(1/2)&\leq(1/2-p)H_\Sigma,\\
\implies f(1/2)&\geq f(1-p)-(1/2-p)H_\Sigma.
\end{align}
Substituting into the preceding lower bound yields
\begin{align}
g(p)
&\geq(1-2p)f_F+2pf(1-p)-2p(1/2-p)H_\Sigma\\
&=2pf(1-p)+(1-2p)f_F-p(1-2p)H_\Sigma.
\end{align}
No restriction on $F$, its dimension, or its correlations with $A^nX$ has
been imposed.
\end{proof}

\subsection{Arbitrary collective instruments}
Take an arbitrary instrument on $R^n$ with maps $V_x\colon R^n\to SF$. It commutes
with the channel action on $A^n$, so it can be analyzed before the erasure channels.
Define the outcome probabilities by
$p_x\coloneqq\|(V_x\otimes I_{A^n})\ket\phi_{RA}^{\otimes n}\|^2$,
and omit outcomes of probability zero. The resulting classical--quantum
state and its normalized conditional vectors are defined as
\begin{align}
\omega_{XSFA^n}&\coloneqq\sum_xp_x\proj{x}_X\otimes\omega^x_{SFA^n},\label{eq:erased-prestate}\\
\ket{\omega^x}_{SFA^n}&
\coloneqq p_x^{-1/2}(V_x\otimes I_{A^n})\ket{\phi}_{RA}^{\otimes n}.
\end{align}
Applying the erasure isometry from \eqref{eq:erasure-isometry} defines
\begin{equation}
\sigma_{XSFB^nE^n}\coloneqq(I_{XSF}\otimes U_p^{\otimes n})\omega_{XSFA^n}
(I_{XSF}\otimes (U_p^\dagger)^{\otimes n}).
\label{eq:erased-poststate}
\end{equation}
Each conditional state on $SFA^n$ is pure, and completeness of the instrument implies
\begin{equation}
\omega_{A^n}=\sum_xp_x\omega^{x}_{A^n}=\tau^{\otimes n}.
\label{eq:fixed-input}
\end{equation}
To see the last assertion directly, for every operator $Z$ on $A^n$,
\begin{align}
\Tr[Z\omega_{A^n}]
&=\sum_x\Tr[(V_x^\dagger V_x\otimes Z)\phi^{\otimes n}]\\
&=\Tr[(I_{R^n}\otimes Z)\phi^{\otimes n}]\\
&=\Tr[Z\tau^{\otimes n}].
\end{align}
Equality of these expectation values for all $Z$ proves \eqref{eq:fixed-input}.

Use the subset polynomials $f$ and $g$ defined in \eqref{eq:subset-definitions} for the
pre-channel state $\omega_{XSFA^n}$, with $S$ traced out when necessary. Then
\begin{equation}
H_\Sigma=nh,\qquad f_F=H(F|X)_\omega\geq0.
\label{eq:Hsum}
\end{equation}
Individual conditional states $\omega^{x}_{A^n}$ are not assumed to be product.

\subsection{Erasure-pattern averaging}
For a fixed conditional input $x$, Bob receives the subsystem $A_J$ when the
surviving set is $J$. Its probability is
$(1-p)^{|J|}p^{n-|J|}$, independently of $x$. Distinct erasure patterns occupy
orthogonal flag sectors. Applying \eqref{eq:blocks} gives
\begin{align}
H(B^n|X)_{\sigma}&=nh_2(p)+f(1-p),\\
H(FE^n|X)_{\sigma}&=nh_2(p)+g(p).
\label{eq:pattern-entropies}
\end{align}
Tracing over $SFX$ leaves the original product $B^nE^n$ marginal. Therefore,
\begin{align}
H(B^nE^n)_{\sigma}&=nh,\\
H(B^n)_{\sigma}&=nh_2(p)+(1-p)nh,\\
H(E^n|B^n)_{\sigma}&=pnh-nh_2(p).
\label{eq:unconditional-block}
\end{align}

\subsection{Information quantities for the erased state}
All information quantities in the following identities are evaluated on the
post-channel state \eqref{eq:erased-poststate}. Strong subadditivity makes
the two conditional mutual information penalties nonnegative. Conditional purity of $SFB^nE^n$ gives
\begin{align}
I(S\rangle B^nX)_{\sigma}
&=H(B^n|X)_{\sigma}-H(SB^n|X)_{\sigma}\\
&=H(B^n|X)_{\sigma}-H(FE^n|X)_{\sigma}\\
&=f(1-p)-g(p).
\label{eq:c}
\end{align}
For the first-facet penalty, apply conditional entropy duality and then
\eqref{eq:pattern-entropies}--\eqref{eq:unconditional-block}:
\begin{align}
I(SX;E^n|B^n)_{\sigma}
&=H(E^n|B^n)_{\sigma}-H(E^n|SB^nX)_{\sigma}\\
&=H(E^n|B^n)_{\sigma}+H(E^n|FX)_{\sigma}\\
&=H(E^n|B^n)_{\sigma}+H(FE^n|X)_{\sigma}-H(F|X)_\omega\\
&=pnh-nh_2(p)+nh_2(p)+g(p)-f_F\\
&=pnh+g(p)-f_F.
\label{eq:d}
\end{align}
Finally, the value of the support objective for this instrument is
\begin{align}
s_{\lambda,\mu}
&\coloneqq I(S\rangle B^nX)_{\sigma}-\lambda I(X;E^n|B^n)_{\sigma}-\mu I(S;E^n|B^nX)_{\sigma}\\
&=I(S\rangle B^nX)_{\sigma}-\mu I(SX;E^n|B^n)_{\sigma}-(\lambda-\mu)I(X;E^n|B^n)_{\sigma}\\
&\leq I(S\rangle B^nX)_{\sigma}-\mu I(SX;E^n|B^n)_{\sigma}.
\label{eq:remove-lambda}
\end{align}
The last inequality only uses $\lambda\geq\mu$ and $I(X;E^n|B^n)_{\sigma}\geq0$.

\subsection{Proof of Theorem~\ref{thm:erased-support}}\label{app:erased-support}
Fix $n\in\mathbb N$ and parameters $\mu\geq0$ and
$\mu\leq\lambda\leq\mu+1$, as in Theorem~\ref{thm:erased-support}.
\begin{proof}
We first treat $0<p<1/2$. Set $k\coloneqq1-2p>0$ and
\begin{equation}
\mu_*\coloneqq\frac{k}{2p}=\frac{1}{2p}-1.
\end{equation}
For every collective instrument, the key interpolation bound (Lemma~\ref{lem:subset-interpolation}) and
\eqref{eq:Hsum} give
\begin{equation}
g(p)\geq2pf(1-p)+kf_F-pknh.
\label{eq:key-product}
\end{equation}
We establish two endpoint bounds.

\noindent\textbf{Endpoint 1: coherent information.}
Using \eqref{eq:c} and \eqref{eq:key-product},
\begin{align}
I(S\rangle B^nX)_{\sigma}
&=f(1-p)-g(p)\\
&\leq(1-2p)f(1-p)-kf_F+pknh\\
&=k[f(1-p)-f_F+pnh].
\end{align}
Because $f(0)=0$, \eqref{eq:Lipschitz} gives
\begin{equation}
f(1-p)\leq(1-p)nh.
\end{equation}
As $f_F\geq0$, it follows that
\begin{align}
I(S\rangle B^nX)_{\sigma}&\leq k[(1-p)nh-f_F+pnh]\\
&=k(nh-f_F)\\
&\leq knh.
\label{eq:coh-bound}
\end{align}

\noindent\textbf{Endpoint 2: the critical communication penalty.}
Using \eqref{eq:c} and \eqref{eq:d},
\begin{align}
I(S\rangle B^nX)_{\sigma}-\mu_*I(SX;E^n|B^n)_{\sigma}&=f(1-p)-g(p)-\mu_*[pnh+g(p)-f_F]\\
&=f(1-p)-(1+\mu_*)g(p)+\mu_*f_F-\mu_*pnh.
\end{align}
Since $2p(1+\mu_*)=1$ and $2p\mu_*=k$,
\begin{align}
2p[I(S\rangle B^nX)_{\sigma}-\mu_*I(SX;E^n|B^n)_{\sigma}]&=2pf(1-p)-g(p)+kf_F-pknh\\
&\leq0,
\label{eq:critical}
\end{align}
where the last inequality follows from~\eqref{eq:key-product}.

\noindent\textbf{Interpolation between the endpoints.}
Suppose $0\leq\mu\leq\mu_*$. Set $\theta\coloneqq\mu/\mu_*\in[0,1]$. Then
\begin{align}
&I(S\rangle B^nX)_{\sigma}-\mu I(SX;E^n|B^n)_{\sigma}\nonumber \\
&=(1-\theta)I(S\rangle B^nX)_{\sigma}+\theta[I(S\rangle B^nX)_{\sigma}-\mu_*I(SX;E^n|B^n)_{\sigma}]\\
&\leq(1-\theta)knh\\
&=\left(k-\frac{k\mu}{\mu_*}\right)nh\\
&=(1-2p-2p\mu)nh.
\end{align}
Together with \eqref{eq:remove-lambda}, this is the desired bound when its
right-hand side is nonnegative. For $\mu\geq\mu_*$, use $I(SX;E^n|B^n)_{\sigma}\geq0$ to obtain
\begin{align}
&I(S\rangle B^nX)_{\sigma}-\mu I(SX;E^n|B^n)_{\sigma}\nonumber \\
&=I(S\rangle B^nX)_{\sigma}-\mu_*I(SX;E^n|B^n)_{\sigma}-(\mu-\mu_*)I(SX;E^n|B^n)_{\sigma}\\
&\leq I(S\rangle B^nX)_{\sigma}-\mu_*I(SX;E^n|B^n)_{\sigma}\\
&\leq0.
\end{align}
This establishes the upper bound for $0<p<1/2$.

\noindent\textbf{The case $p=1/2$.}
Equations \eqref{eq:c} and \eqref{eq:midpoint} give
\begin{equation}
I(S\rangle B^nX)_{\sigma}=f(1/2)-g(1/2)\leq0.
\end{equation}
Since $I(SX;E^n|B^n)_{\sigma}\geq0$ and $I(X;E^n|B^n)_{\sigma}\geq0$ and $\lambda\geq\mu\geq0$, the support objective is at most zero.

\noindent\textbf{The case $p>1/2$.}
Bob's channel at erasure probability $p$ can be implemented by first applying
$\cE_{1/2}$, and then erasing a surviving data system with probability $2p-1$.
The existing erasure flag is left unchanged. The total survival probability is
$\tfrac12[1-(2p-1)]=1-p$. Thus the $p$-output is a local degradation of the
$1/2$-output. For completeness, denote the instrument output for erasure probability
$1/2$ by $\sigma^{1/2}_{XSFB_{1/2}^nE^n}$. Let
$W\colon B_{1/2}^n\to B_p^nG$ be a dilation of the product degradation, and let
$\zeta_{XSFB_p^nGE^n}$ be the state after applying $W$.
Its $XSB_p^n$ marginal agrees with the marginal obtained using
$\mathcal E_p^{\otimes n}$ directly, because the composed erasure
channels agree on every input.
Isometric invariance and strong subadditivity give
\begin{align}
H(S|B_p^nX)_\zeta-H(S|B_{1/2}^nX)_{\sigma^{1/2}}&=H(S|B_p^nX)_\zeta-H(S|B_p^nGX)_\zeta\\
&=I(S;G|B_p^nX)_\zeta\geq0.
\end{align}
Consequently,
$I(S\rangle B_p^nX)_\zeta\leq I(S\rangle B_{1/2}^nX)_{\sigma^{1/2}}\leq0$.
The remaining terms in \eqref{eq:support} are nonpositive, so the objective is at most zero.

\noindent\textbf{The case $p=0$.}
The environment is fixed and pure, so both conditional mutual information
penalties vanish. Conditional purity gives
\begin{align}
I(S\rangle B^nX)_{\sigma}&=H(A^n|X)_\omega-H(F|X)_\omega\\
&\leq H(A^n)_\omega\\
&=nh.
\end{align}
Here concavity of entropy gives $H(A^n|X)_\omega\leq H(A^n)_\omega$, and $H(F|X)_\omega\geq0$.

\noindent\textbf{Attainment of the upper bound.}
For the identity instrument, take $S=R^n$ with $F$ and $X$ trivial. Directly,
\begin{align}
I(R^n\rangle B^n)_{(\psi^p)^{\otimes n}}
&=H(B^n)_{(\psi^p)^{\otimes n}}-H(E^n)_{(\psi^p)^{\otimes n}}\\
&=n[h_2(p)+(1-p)h-h_2(p)-ph]\\
&=(1-2p)nh.
\end{align}
Also, by purity,
\begin{align}
I(R^n;E^n|B^n)_{(\psi^p)^{\otimes n}}
&=I(R^n;E^n)_{(\psi^p)^{\otimes n}}\\
&=H(R^n)_{(\psi^p)^{\otimes n}}+H(E^n)_{(\psi^p)^{\otimes n}}-H(B^n)_{(\psi^p)^{\otimes n}}\\
&=nh+n[h_2(p)+ph-h_2(p)-(1-p)h]\\
&=2pnh.
\end{align}
Its support value is therefore $[1-2p(1+\mu)]nh$. The instrument that assigns
$R^n$ to $F$, and outputs a fixed pure $S$ and constant $X$, has support
value zero. The better of these two instruments attains the upper bound in
all cases.
\end{proof}

\subsection[Derivation of the five-facet region]{Derivation of the facets in~\eqref{eq:five-first}--\eqref{eq:five}}
\label{app:erased-facets}
We derive the five inequalities in~\eqref{eq:five-first}--\eqref{eq:five}
and verify that they describe the erased-state region in~\eqref{eq:erased-region}.
For $0<p<1/2$ and $h>0$, define
\begin{equation}
k\coloneqq1-2p,\qquad \mu_*\coloneqq\frac{k}{2p}.
\label{eq:erased-facet-parameters}
\end{equation}
The inequality in~\eqref{eq:five-first} follows from the first information
facet~\eqref{eq:facet1} and nonnegativity of conditional mutual information.
Substituting the erased-state support value~\eqref{eq:erased-support}
into the supporting inequality in~\eqref{eq:halfspace}, with
$(\lambda,\mu)=(0,0)$ and $(1,0)$, gives
\eqref{eq:five-second} and~\eqref{eq:five-third}, respectively.
At $\mu=\mu_*$ from~\eqref{eq:erased-facet-parameters}, the support
value in~\eqref{eq:erased-support} is zero. The choice
$(\lambda,\mu)=(\mu_*,\mu_*)$ in~\eqref{eq:halfspace}, followed by
multiplication by $2p$, therefore gives
\begin{align}
2p\mu_* C+2p(1+2\mu_*)Q+2pE&\leq0,\\
kC+2(1-p)Q+2pE&\leq0.
\end{align}
The second line uses~\eqref{eq:erased-facet-parameters} and is exactly
\eqref{eq:five-fourth}. Similarly, the choice
$(\lambda,\mu)=(\mu_*+1,\mu_*)$ gives~\eqref{eq:five}, because
$2p(\mu_*+1)=1$ and $2p(1+2\mu_*)=2(1-p)$.
Both choices lie in the parameter domain~\eqref{eq:parameters}.

To verify sufficiency algebraically, recall that the region in
\eqref{eq:erased-region} is equivalently specified by
\eqref{eq:region-t-first}--\eqref{eq:region-t}. These inequalities require
\begin{equation}
\max\!\left\{0,\frac{Q+E}{kh},\frac{C+Q+E}{kh}\right\}
\leq t\leq
\min\!\left\{1,-\frac{C+2Q}{2ph}\right\}.
\end{equation}
The inequalities in~\eqref{eq:five-second}--\eqref{eq:five-third}
guarantee that the lower bound is at most $1$, and
\eqref{eq:five-first} guarantees that the upper bound is nonnegative.
The inequality in~\eqref{eq:five-fourth} gives
\begin{align}
2p(Q+E)&\leq-k(C+2Q),\\
\frac{Q+E}{kh}&\leq-\frac{C+2Q}{2ph}.
\end{align}
The inequality in~\eqref{eq:five} gives
\begin{align}
2p(C+Q+E)&\leq-k(C+2Q),\\
\frac{C+Q+E}{kh}&\leq-\frac{C+2Q}{2ph}.
\end{align}
Thus every component of the lower bound is at most every component of the
upper bound, and a feasible $t$ exists. If $h=0$, the input is unentangled and
the region is $\cU$. At $p=0$, the answer is $(0,0,h)+\cU$, as also follows
directly from \eqref{eq:region-t-first}--\eqref{eq:region-t}.

\section{An entropic proof of the Hadamard coherent-information face}
\label{app:hadamard}
We give the independent proof of Theorem~\ref{thm:had-face} announced in
Section~\ref{sec:hadamard}. We also recall the conditional-entropy
inequalities used in the penalty decomposition in~\eqref{eq:had-penalty}.
\begin{proof}
Fix the purification $\psi_{ABE}$ from~\eqref{eq:had-purification}, and let
$\sigma_{XSFBE}$ be the output of an arbitrary local instrument, as defined
in~\eqref{eq:instrument-state}. Because $\psi_{AE}$ is separable and the
instrument acts on $A$, every normalized conditional state
$\sigma^x_{SFE}$ is separable across $SF:E$. In particular, choose a
decomposition
\begin{equation}
\sigma^x_{FE}
=\sum_zq_{z|x}\alpha_F^{xz}\otimes\beta_E^{xz},
\end{equation}
where $\alpha_F^{xz}$ and $\beta_E^{xz}$ are states. Define its classical
extension by
\begin{equation}
\zeta_{XZFE}\coloneqq
\sum_{x,z}p_xq_{z|x}\proj x_X\otimes\proj z_Z
\otimes\alpha_F^{xz}\otimes\beta_E^{xz}.
\label{eq:had-separable-extension}
\end{equation}
Tracing out $Z$ gives $\sigma_{XFE}$. Strong subadditivity and the product
form conditioned on $X,Z$ imply
\begin{align}
H(F|EX)_\sigma
&\geq H(F|EXZ)_\zeta
=\sum_{x,z}p_xq_{z|x}H(F)_{\alpha^{xz}}\geq0,
\label{eq:had-sep-entropy}\\
H(E|FX)_\sigma
&\geq H(E|FXZ)_\zeta
=\sum_{x,z}p_xq_{z|x}H(E)_{\beta^{xz}}\geq0.
\label{eq:had-sep-entropy-E}
\end{align}
Conditional purity of $\sigma^x_{SFBE}$, the chain rule, and
\eqref{eq:had-sep-entropy} give
\begin{align}
I(S\rangle BX)_\sigma
&=H(B|X)_\sigma-H(SB|X)_\sigma\\
&=H(B|X)_\sigma-H(FE|X)_\sigma\\
&=H(B|X)_\sigma-H(E|X)_\sigma-H(F|EX)_\sigma\\
&\leq H(B)_\psi-H(E)_\psi
-\big[I(X;B)_\sigma-I(X;E)_\sigma\big]\\
&\leq H(B)_\psi-H(E)_\psi
=I(A\rangle B)_\rho.
\label{eq:had-coh}
\end{align}
The penultimate line uses the fixed marginal $\sigma_{BE}=\psi_{BE}$.
The final inequality follows from data processing under the degrading channel, which
reproduces $\sigma_{XE}$ from $\sigma_{XB}$. The last equality follows
from purity of $\psi_{ABE}$ and $\psi_{AB}=\rho_{AB}$.

At $\mu=0$, the objective in~\eqref{eq:support} is
\begin{equation}
I(S\rangle BX)_\sigma-\lambda I(X;E|B)_\sigma
\leq I(A\rangle B)_\rho,
\end{equation}
because $\lambda\geq0$ and conditional mutual information is nonnegative.
The identity instrument, with $S=A$ and $X,F$ trivial, attains this bound.
For $n$ copies, $\psi_{AE}^{\otimes n}$ is separable and the
product degrading channel gives the same data-processing inequality. Thus
the preceding upper bound applies to arbitrary collective instruments and
equals
\begin{equation}
I(A^n\rangle B^n)_{\rho^{\otimes n}}
=nI(A\rangle B)_\rho.
\end{equation}
The product identity instrument attains it, proving~\eqref{eq:had-mu0}.
\end{proof}

\section{Proofs for maximally correlated states}\label{app:mc}
This appendix gives the proofs underlying Section~\ref{sec:mc}.
It establishes the canonical instrument form, the exact matrix
optimization and cardinality bounds, the efficient-instrument and
commuting-memory restrictions, and the strict quantum-memory advantage.
It then proves the bottleneck identity and additivity of the
unrestricted outer bound, while distinguishing the remaining
memory-reduction question.

\subsection{Canonical instrument normal form}\label{app:mc-normal-form}
We prove Theorem~\ref{thm:normal-form} in two steps. First, retaining
an orthogonal copy of the input basis label improves all three information
inequalities. Second, a suitable refinement makes the discarded states
pure without losing this improvement. We begin with entropy identities
that allow the two comparisons to use the same quantities.

\subsubsection{Three entropy identities}
For an arbitrary instrument, conditional purity expresses the three
information quantities in terms of the systems $S,F,E$.
These identities make it possible to compare instruments while preserving
their discarded-system and environment marginal.

For an instrument with operators $V_x\colon A\to SF$ as in~\eqref{eq:instrument},
define its normalized conditional output vectors and its full output state by
\begin{align}
\ket{\sigma_x}_{SFBE}
&\coloneqq w_x^{-1/2}(V_x\otimes I_{BE})\ket\psi_{ABE},\\
\sigma_{XSFBE}&\coloneqq\sum_x w_x\proj x_X\otimes\proj{\sigma_x}_{SFBE},
\end{align}
where $w_x$ is the probability of outcome $x$; zero-probability outcomes are
omitted. Conditional purity and~\eqref{eq:EgivenBzero} imply
\begin{align}
I(S\rangle BX)_\sigma
&=H(B|X)_\sigma-H(SB|X)_\sigma\\
&=H(SFE|X)_\sigma-H(FE|X)_\sigma,\label{eq:c-purity}\\
I(SX;E|B)_\sigma
&=H(E|B)_\psi-H(E|SBX)_\sigma\\
&=H(E|FX)_\sigma,\label{eq:ell-purity}\\
I(X;E|B)_\sigma
&=H(E|B)_\psi-H(E|BX)_\sigma\\
&=H(B|X)_\sigma-H(BE|X)_\sigma\\
&=H(SFE|X)_\sigma-H(SF|X)_\sigma.
\label{eq:j-purity}
\end{align}
For the second identity, conditional entropy duality gives
$H(E|SBX)_\sigma=-H(E|FX)_\sigma$. For the final identity, the complements
of $B$ and $BE$ in each pure conditional state are $SFE$ and $SF$, respectively.

\subsubsection{Retaining an orthogonal copy of the basis label}
We first construct a new instrument that retains an orthogonal record
of the distinguished input basis label. Its $FXE$ marginal agrees with
the original one. The identities above will show that this construction
can only enlarge the information region.

For every $x,i$, define the following quantities, omitting outcomes with
$w_x=0$ and using the vector formula when $a(x|i)>0$:
\begin{align}
a(x|i)&\coloneqq\|V_x\ket i\|^2,\\
w_x&\coloneqq\sum_i p_i a(x|i),\\
q_i^x&\coloneqq\frac{p_i a(x|i)}{w_x},\\
\ket{v_{xi}}_{SF}&\coloneqq\frac{V_x\ket i}{\sqrt{a(x|i)}}.
\end{align}
If $a(x|i)=0$, choose $\ket{v_{xi}}_{SF}$ to be any fixed unit vector
and define its discarded marginal and pure-state decomposition accordingly.
The factor $\sqrt{a(x|i)}=0$ makes these choices irrelevant to the instrument
output, while the tuple distribution in~\eqref{eq:independent-refinement}
remains defined for every $i$. Completeness gives
\begin{equation}
\sum_xa(x|i)=\bra i\Big(\sum_xV_x^\dagger V_x\Big)\ket i=1.
\end{equation}
Define $q_Y^x\coloneqq\sum_iq_i^x\proj i_Y$ and the classical state
$\omega_{YX}\coloneqq\sum_{i,x}p_i a(x|i)\proj i_Y\otimes\proj x_X$.
This extends the fixed marginal $\omega_Y$ of~\eqref{eq:cqsource}.
Tracing out $B$ after the original instrument gives
\begin{equation}
(\sigma_x)_{SFE}=\sum_iq_i^x\proj{v_{xi}}_{SF}\otimes e_{i,E}.
\label{eq:original-SFE}
\end{equation}
Every state in this ensemble is pure. Thus the entropy-of-mixture bound gives
\begin{equation}
H(SFE)_{\sigma_x}\leq H(Y)_{q^x}.
\label{eq:mixture-bound}
\end{equation}
One proof uses the Gram matrix of the pure ensemble. It has the same nonzero
spectrum as the mixture and has diagonal $q^x$. Dephasing this Gram matrix
cannot decrease entropy, which proves~\eqref{eq:mixture-bound}.

Define a new instrument by
\begin{equation}
W_x\ket i_A\coloneqq\sqrt{a(x|i)}\ket i_L\ket{v_{xi}}_{SF},
\label{eq:copied-instrument}
\end{equation}
with retained system $LS$ and discarded system $F$. Its output is denoted by
$\widetilde\sigma_{XLSFBE}$, obtained from the preceding definition of $\sigma$
by replacing $V_x$ with $W_x$. More explicitly,
\begin{align}
\ket{\widetilde\sigma_x}_{LSFBE}
&=\sum_i\sqrt{q_i^x}\ket i_L\ket{v_{xi}}_{SF}\ket i_B\ket{e_i}_E,\\
\widetilde\sigma_{XLSFBE}
&=\sum_xw_x\proj x_X\otimes\proj{\widetilde\sigma_x}_{LSFBE}.
\end{align}
It is a legitimate instrument because
\begin{align}
\bra iW_x^\dagger W_x\ket j
&=\sqrt{a(x|i)a(x|j)}\langle i|j\rangle\langle v_{xi}|v_{xj}\rangle\\
&=\delta_{ij}a(x|i),\\
\sum_xW_x^\dagger W_x&=\sum_i\Big(\sum_xa(x|i)\Big)\proj i=I_A.
\end{align}
The $FXE$ marginal is unchanged. In $\widetilde\sigma$, the $LSFE$ and $LSF$
marginals conditional on $X$ are block diagonal in $L$, with pure conditional
states. Consequently,
\begin{equation}
H(LSFE|X)_{\widetilde\sigma}
=H(LSF|X)_{\widetilde\sigma}=H(Y|X)_\omega.
\end{equation}
The classical register $Y$ is used to evaluate these entropies; introducing it
does not prescribe a measurement of the original coherent state.

Using~\eqref{eq:c-purity}--\eqref{eq:j-purity}, we obtain
\begin{align}
I(LS\rangle BX)_{\widetilde\sigma}
&=H(Y|X)_\omega-H(FE|X)_\sigma
\geq I(S\rangle BX)_\sigma,\\
I(LSX;E|B)_{\widetilde\sigma}
&=H(E|FX)_\sigma=I(SX;E|B)_\sigma,\\
I(X;E|B)_{\widetilde\sigma}
&=0\leq I(X;E|B)_\sigma.
\label{eq:copy-dominance}
\end{align}
Consequently, all three right sides of~\eqref{eq:facet1}--\eqref{eq:facet3}
increase or stay fixed. The new instrument dominates the original information
region, simultaneously for every support parameter.

\subsubsection{Pure conditional states in the discarded system}
The preceding construction may still leave mixed conditional states
on $F$. We replace them by a pure-state refinement whose additional
classical outcome is independent of the input label conditional on $X$.
This independence preserves the required classical entropy while strong
subadditivity controls the other two quantities.

Set $\tau_{xi,F}\coloneqq\Tr_S\proj{v_{xi}}_{SF}$. These states may be mixed.
For each $x,i$, choose a pure-state decomposition
\begin{equation}
\tau_{xi,F}=\sum_k b_{xi,k}\proj{f_{xi,k}}_F.
\end{equation}
For fixed $x$, introduce a tuple $Z\coloneqq(k_1,\ldots,k_d)$ with distribution
\begin{equation}
\pi(z|x)\coloneqq\prod_{i=1}^d b_{xi,k_i}.
\label{eq:independent-refinement}
\end{equation}
The tuple $Z$ is independent of $Y$ conditional on $X$. Define a further instrument
\begin{equation}
\widehat W_{xz}\ket i_A
\coloneqq\sqrt{a(x|i)\pi(z|x)}\ket i_{\widehat S}\ket{f_{xi,k_i}}_F.
\label{eq:pure-refinement}
\end{equation}
Its output state is
\begin{align}
\ket{\widehat\sigma_{xz}}_{\widehat SFBE}
&\coloneqq\sum_i\sqrt{q_i^x}\ket i_{\widehat S}
  \ket{f_{xi,k_i}}_F\ket i_B\ket{e_i}_E,\\
\widehat\sigma_{XZ\widehat SFBE}
&\coloneqq\sum_{x,z}w_x\pi(z|x)\proj x_X\otimes\proj z_Z
  \otimes\proj{\widehat\sigma_{xz}}_{\widehat SFBE}.
\end{align}
Completeness follows from
\begin{equation}
\sum_{x,z}\widehat W_{xz}^\dagger\widehat W_{xz}
=\sum_i\Big[\sum_xa(x|i)\sum_z\pi(z|x)\Big]\proj i=I_A.
\end{equation}
To compare the reduced states, introduce the common classical--quantum extension
\begin{equation}
\Omega_{YXZFE}\coloneqq\sum_{i,x,z}p_i a(x|i)\pi(z|x)
\proj i_Y\otimes\proj x_X\otimes\proj z_Z
\otimes\proj{f_{xi,k_i}}_F\otimes e_{i,E}.
\end{equation}
Averaging over $Z$ recovers the previous $FXE$ marginal because
\begin{align}
\sum_z\pi(z|x)\proj{f_{xi,k_i}}_F
&=\sum_{k_i}b_{xi,k_i}\proj{f_{xi,k_i}}_F
\prod_{j\ne i}\Big(\sum_{k_j}b_{xj,k_j}\Big)\\
&=\tau_{xi,F}.
\end{align}
Thus $\Omega_{FXE}=\widetilde\sigma_{FXE}$,
$\Omega_{FXZE}=\widehat\sigma_{FXZE}$, and
$H(Y|XZ)_\Omega=H(Y|X)_\Omega=H(Y|X)_\omega$. The desired comparisons are
\begin{align}
I(\widehat S\rangle BXZ)_{\widehat\sigma}
-I(LS\rangle BX)_{\widetilde\sigma}&=[H(Y|XZ)_\Omega-H(FE|XZ)_\Omega]
-[H(Y|X)_\Omega-H(FE|X)_\Omega]\\
&=H(FE|X)_\Omega-H(FE|XZ)_\Omega\\
&=I(FE;Z|X)_\Omega\geq0,\\
&I(LSX;E|B)_{\widetilde\sigma}
-I(\widehat SXZ;E|B)_{\widehat\sigma}\nonumber \\
&=H(E|FX)_\Omega-H(E|FXZ)_\Omega\\
&=I(E;Z|FX)_\Omega\geq0,\\
&I(XZ;E|B)_{\widehat\sigma}=0.
\end{align}
All three information facets are again improved.

\begin{proof}
The two constructions above dominate an arbitrary instrument and end in the
form~\eqref{eq:canonical}, after relabeling the retained system and the classical
outcome. The reverse inclusion holds because these are legitimate instruments.
For a fixed outcome $x$, the at most $d$ vectors $f_{xi}$ span a space of dimension
at most $d$. Isometrically representing each such span in the same fixed
$d$-dimensional space preserves its Gram matrix and every relevant entropy.
This representation can depend on $x$ because the instrument has separate
operators $V_x$; no direct sum over outcomes is required in $F$.
The same argument applies to $d^n$ basis labels.
\end{proof}

\subsection{Proof of the exact matrix formula and cardinality bound}\label{app:mc-matrix-formula}
We now prove Theorem~\ref{thm:matrix-formula}. The entropy calculation
in~\eqref{eq:matrix-c}--\eqref{eq:matrix-ell} maps each canonical instrument
to a feasible matrix ensemble. We establish the reverse realization and
then bound the number of outcomes by preserving the diagonal constraint
and one scalar objective.

\begin{proof}
Theorem~\ref{thm:normal-form} and~\eqref{eq:matrix-c}--\eqref{eq:matrix-ell}
give one direction. For the converse, start with an arbitrary ensemble satisfying
\eqref{eq:diagonal-constraint}. Let $q_i^x\coloneqq(\theta_x)_{ii}$ and choose a purification
\begin{equation}
\ket{\vartheta_x}_{\widetilde AF}
\coloneqq\sum_i\sqrt{q_i^x}\ket i_{\widetilde A}\ket{f_{xi}}_F.
\end{equation}
Expand a purification in the distinguished basis of $\widetilde A$ and normalize its
coefficient vectors to obtain this representation. Coefficients of zero norm
can be assigned arbitrary unit vectors because their weights vanish.
The reduced state on $\widetilde A$ is exactly~\eqref{eq:theta}. Set
\begin{equation}
a(x|i)\coloneqq\frac{w_xq_i^x}{p_i}.
\end{equation}
The diagonal constraint gives
\begin{equation}
\sum_xa(x|i)=\frac{\sum_xw_xq_i^x}{p_i}=1.
\end{equation}
The canonical instrument~\eqref{eq:canonical} therefore realizes the desired
ensemble and support value.

For the cardinality bound, we use the following connected-set form of the
Fenchel--Eggleston--Carath\'eodory theorem~\cite{EGK11}:
every point in the convex hull of a connected subset of $\mathbb R^m$
is a convex combination of at most $m$ points of that subset.
Consider the continuous map from density operators on $\widetilde A$ to
\begin{equation}
\big(\theta_{11},\ldots,\theta_{d-1,d-1},L_\mu^G(\theta)\big)
\in\mathbb R^d.
\label{eq:cardinality-support-map}
\end{equation}
The density-operator set is compact and convex, hence path-connected.
Entropy is continuous, including at density operators with zero eigenvalues,
so the image in~\eqref{eq:cardinality-support-map} is compact and
path-connected. Its convex hull is compact. Fixing the first $d-1$
coordinates at $p_1,\ldots,p_{d-1}$ therefore leaves a nonempty compact
set on which the last coordinate attains its maximum.
The connected-set theorem represents a maximizing point using at most
$d$ image points. The remaining diagonal coordinate is preserved by
the trace-one condition. The reverse construction above realizes the
resulting ensemble as an instrument, proving $|X|\leq d$.
For all $n\in\mathbb N$, the same argument on $\widetilde A^n$ gives
$|X|\leq d^n$ for the support optimization on $n$ copies.
\end{proof}

\subsection{The complete one-copy region}\label{app:mc-region}
We now identify the full one-copy information region and bound the
number of outcomes needed to preserve an arbitrary primitive rate point.
The canonical primitive point
\begin{equation}
\left(0,-\tfrac12 I(SX;E|B)_\sigma,
I(S\rangle BX)_\sigma+\tfrac12I(SX;E|B)_\sigma\right)
\end{equation}
saturates all three canonical facets. Substituting~\eqref{eq:matrix-c} and
\eqref{eq:matrix-ell} gives~\eqref{eq:primitive-matrix-Q}--\eqref{eq:primitive-matrix};
Theorem~\ref{thm:normal-form} shows that arbitrary instruments do not enlarge the union.

To preserve an arbitrary primitive rate point, consider the map
\begin{equation}
\theta\longmapsto
\big(\theta_{11},\ldots,\theta_{d-1,d-1},Q(\theta),E(\theta)\big)
\in\mathbb R^{d+1},
\label{eq:cardinality-region-map}
\end{equation}
where $Q(\theta)$ and $E(\theta)$ are defined in
\eqref{eq:primitive-matrix-Q}--\eqref{eq:primitive-matrix}.
As in the preceding proof, its image is compact and path-connected.
The Fenchel--Eggleston--Carath\'eodory theorem therefore preserves the
average diagonal and both average rates using at most $d+1$ outcomes.
The reverse construction in Appendix~\ref{app:mc-matrix-formula}
again gives a valid instrument. Thus $|X|\leq d+1$ suffices for
the complete one-copy region, and $|X|\leq d^n+1$ suffices on
$n$ copies for all $n\in\mathbb N$. These are sufficient bounds;
no claim of optimality of the cardinalities is needed.

Both information quantities are nonnegative for canonical instruments.
For coherent information this follows because dephasing $\cN_G(\theta)$ gives
$\Delta(\theta)$ and cannot decrease entropy. For a single conditional state,
define $\omega_{YFE}^{\theta}\coloneqq\sum_iq_i\proj i_Y\otimes\proj{f_i}_F
\otimes e_{i,E}$, where $q_i\coloneqq\theta_{ii}$ and the vectors $f_i$ purify $\theta$
as above. Its $FE$ marginal is separable, and strong subadditivity gives
\begin{align}
H(E|F)_{\omega^\theta}&\geq H(E|FY)_{\omega^\theta}\\
&=\sum_iq_iH(E)_{e_i}=0.
\end{align}
Together with the spectral identities, this proves $Q(\theta)\leq0$ and
$E(\theta)\geq0$.

\subsection{Additivity of the efficient restriction}\label{app:mc-efficient}
We prove Proposition~\ref{prop:efficient} by splitting the objective
of a joint efficient instrument into two admissible one-copy objectives.
The essential step is data processing from the first classical input
label to its environment output. Product instruments give the matching
lower bound.

\begin{proof}
Write $\alpha\coloneqq1+\mu$. A joint efficient ensemble determines the state
\begin{equation}
\eta_{XY_1Y_2E_1E_2}\coloneqq
\sum_{x,i_1,i_2}w_xq^x_{i_1i_2}\proj x_X\otimes
\proj{i_1}_{Y_1}\otimes\proj{i_2}_{Y_2}
\otimes e_{i_1,E_1}\otimes e_{i_2,E_2},
\end{equation}
where $\sum_xw_xq^x_{i_1i_2}=p_{1,i_1}p_{2,i_2}$.
Its objective decomposes as
\begin{align}
&H(Y_1Y_2|X)_\eta-\alpha H(E_1E_2|X)_\eta\nonumber \\
&=H(Y_1|X)_\eta-\alpha H(E_1|X)_\eta+H(Y_2|Y_1X)_\eta-\alpha H(E_2|E_1X)_\eta.
\end{align}
The state on $E_1E_2X$ is obtained from that on $Y_1E_2X$ by preparing
$e_{i_1,E_1}$ from $Y_1$. Conditional data processing therefore gives
\begin{align}
I(E_2;E_1|X)_\eta&\leq I(E_2;Y_1|X)_\eta,\\
H(E_2|E_1X)_\eta&\geq H(E_2|Y_1X)_\eta.
\end{align}
Since $\alpha\geq0$, the joint objective is at most
\begin{align}
&\big[H(Y_1|X)_\eta-\alpha H(E_1|X)_\eta\big]+\big[H(Y_2|Y_1X)_\eta-\alpha H(E_2|Y_1X)_\eta\big].
\end{align}
These are valid one-copy objectives with classical auxiliaries $X$ and $XY_1$,
respectively. Their unconditional input distributions are $p_1$ and $p_2$, so
each is bounded by its one-copy optimum. Product efficient instruments attain
the sum of the two optima, proving equality.
\end{proof}

\subsection{Commuting discarded memories}\label{app:mc-commuting}
We now justify the commuting-memory reduction stated in Section~\ref{sec:mc}.
The canonical refinement from Appendix~\ref{app:mc-normal-form} can be
chosen in a common eigenbasis for each outcome. Its eigenvector label
can then be included in the classical outcome of an efficient instrument.

Suppose the conditional states $\left(\tau_{xi,F}\right)_i$ defined
in~\eqref{eq:commuting-conditional-states} commute for each fixed $x$.
In the pure-state refinement, choose a common eigenbasis for each fixed
$x$; this basis may depend on $x$. After absorbing
the independent tuple into the outcome $X$, every refined discarded vector is
a basis vector $\ket{z(x,i)}_F$. Denote this classical memory label by $Z$;
it is deterministic given $Y,X$. The associated state is
\begin{equation}
\eta_{XYZ E}\coloneqq\sum_{i,x}p_i a(x|i)\proj x_X\otimes\proj i_Y
\otimes\proj{z(x,i)}_Z\otimes e_{i,E}.
\end{equation}
Using the matrix entropy expression with $F$ identified with $Z$, its objective is
\begin{align}
&H(Y|X)_\eta+\mu H(Z|X)_\eta-(1+\mu)H(ZE|X)_\eta\nonumber \\
&=H(Y|X)_\eta+\mu H(Z|X)_\eta-(1+\mu)\big[H(Z|X)_\eta+H(E|XZ)_\eta\big]\\
&=H(Y|X)_\eta-H(Z|X)_\eta-(1+\mu)H(E|XZ)_\eta\\
&=H(Y|XZ)_\eta-(1+\mu)H(E|XZ)_\eta.
\end{align}
The last equality follows from
$H(YZ|X)_\eta=H(Y|X)_\eta+H(Z|YX)_\eta=H(Y|X)_\eta$.
The resulting expression is an efficient objective with outcome $XZ$.
The preceding dominance constructions can only increase the original support
value, so commuting conditional memories cannot improve the efficient optimum.
The same construction applies to an instrument on $A^n$, treating the
entire block as its input. Proposition~\ref{prop:efficient} then bounds
the resulting efficient objective by $n$ times its one-copy value.

\subsection{Proof of the quantum-memory advantage}\label{app:mc-quantum-memory}
We prove Theorem~\ref{thm:quantum-memory} by evaluating the efficient
optimum for its two-qubit example and constructing a strictly better
instrument with a noncommuting discarded memory. Exact rational entropy
bounds certify the comparison; the preceding additivity and commuting-memory
results extend it to the stated collective comparison classes.

Consider the two-qubit maximally correlated state with coefficient matrix
\begin{equation}
r\coloneqq\begin{pmatrix}1/2&2/5\\2/5&1/2\end{pmatrix}.
\label{eq:qubit-r}
\end{equation}
Its eigenvalues are \(9/10\) and \(1/10\). Here
\begin{equation}
p=(1/2,1/2),\qquad
G=\begin{pmatrix}1&4/5\\4/5&1\end{pmatrix}.
\end{equation}
Take \(\mu=1\); then any \(1\leq\lambda\leq2\) is admissible.
We first prove the exact efficient optimum, rather than merely comparing against the identity instrument.

\subsubsection{Concavity of the efficient conditional objective}
The efficient optimum in~\eqref{eq:efficient} is a concave-envelope
optimization. For the parameters of Theorem~\ref{thm:quantum-memory},
we show that its conditional objective is already concave, so its value
at the prescribed mean gives the exact optimum.

Let \(\gamma\coloneqq4/5\), and write \(q\in[0,1]\) for the first input probability. The state $\eta_E^q=q e_{0,E}+(1-q)e_{1,E}$, with $\langle e_1|e_0\rangle=\gamma$, has eigenvalues \((1\pm u(q))/2\), where
\begin{equation}
u(q)\coloneqq\sqrt{1-4(1-\gamma^2)q(1-q)}.
\end{equation}
Indeed its trace is one and its determinant is \(q(1-q)(1-\gamma^2)\). The efficient score is therefore
\begin{equation}
f(q)\coloneqq h_2(q)-2h_2\left(\frac{1+u(q)}2\right).
\label{eq:f-qubit}
\end{equation}
We prove that \(f\) is concave.

Set \(z\coloneqq1-2q\), \(k\coloneqq1-\gamma^2\), and \(u\coloneqq\sqrt{\gamma^2+kz^2}\). Direct differentiation gives
\begin{align}
\frac{du}{dq}&=-\frac{2kz}{u},\\
\frac{d^2u}{dq^2}&=\frac{4k\gamma^2}{u^3}.
\end{align}
For \(h(u)\coloneqq h_2((1+u)/2)\),
\begin{align}
h'(u)&=-\frac{\atanh u}{\ln2},\\
h''(u)&=-\frac1{(1-u^2)\ln2}.
\end{align}
It follows that
\begin{align}
\frac{d^2}{dq^2}h(u(q))
=-\frac4{\ln2}\left[
\frac{k^2z^2}{u^2(1-u^2)}+\frac{k\gamma^2\atanh u}{u^3}\right].
\label{eq:second-h}
\end{align}
Also
\begin{equation}
h_2''(q)=-\frac1{q(1-q)\ln2}=-\frac4{(1-z^2)\ln2}.
\end{equation}
Divide \eqref{eq:second-h} by this negative number, use
\(1-u^2=k(1-z^2)\) and \(kz^2=u^2-\gamma^2\), and obtain
\begin{align}
R(u)
&\coloneqq\frac{\frac{d^2}{dq^2}h(u(q))}{h_2''(q)}\\
&=\frac{kz^2}{u^2}
 +\frac{k\gamma^2(1-z^2)\atanh u}{u^3}\\
&=1-\frac{\gamma^2}{u^2}
 +\frac{\gamma^2(1-u^2)\atanh u}{u^3}\\
&=1-\gamma^2 A(u),
\end{align}
where
\begin{equation}
A(u)\coloneqq\frac{u-(1-u^2)\atanh u}{u^3}.
\end{equation}
Expanding \(\atanh u=\sum_{m\geq0}u^{2m+1}/(2m+1)\) yields
\begin{equation}
A(u)=\sum_{m=0}^{\infty}\frac{2u^{2m}}{(2m+1)(2m+3)}.
\end{equation}
Thus \(A\) is increasing on \((0,1)\). Since \(u(q)\geq\gamma\),
\begin{align}
R(u(q))&\leq R(\gamma)\\
&=\frac{1-\gamma^2}{\gamma}\atanh\gamma\\
&=\frac9{20}\ln3.
\end{align}
For an entirely explicit bound,
\begin{align}
\ln3
&=2\sum_{m\geq0}\frac{(1/2)^{2m+1}}{2m+1}\\
&\leq1+\frac1{12}+\frac25\sum_{m\geq2}(1/2)^{2m+1}\\
&=1+\frac1{12}+\frac25\frac{1/32}{1-1/4}\\
&=\frac{11}{10}.
\end{align}
Consequently \(2R(u(q))\leq99/100<1\), and
\begin{equation}
f''(q)=h_2''(q)[1-2R(u(q))]<0\qquad(0<q<1).
\end{equation}
Continuity extends concavity to \([0,1]\). Jensen's inequality now implies, for every ensemble of mean \(q=1/2\),
\begin{equation}
\sum_xw_xf(q_x)\leq f(1/2)=1-2h_2(9/10).
\end{equation}
The single conditional distribution \(q=1/2\) attains equality. Thus
\begin{equation}
\sS^{\eff}_{1}(p)=1-2h_2(9/10)\approx0.0620088128.
\label{eq:efficient-exact-example}
\end{equation}
By the preceding efficient-instrument analysis, this also bounds all commuting conditional memories and, by additivity, all collective efficient instruments per copy.

\subsubsection{A strictly better quantum-memory instrument}
Having evaluated the efficient optimum, we compute the score of the
instrument with nonorthogonal discarded states. We first certify the
strict separation and then verify the improved rational example
in~\eqref{eq:refined-witness}.

Set $t\coloneqq9/10$ and choose unit vectors $f_0,f_1$ with
$\langle f_1|f_0\rangle=t$. Define the single-outcome isometry
$V\colon A\to SF$ by
\begin{equation}
V\ket i_A\coloneqq\ket i_S\ket{f_i}_F.
\end{equation}
It is an isometry because the retained labels are orthogonal. The matrix in \eqref{eq:theta} is
\begin{equation}
\theta_t\coloneqq\frac12\begin{pmatrix}1&t\\t&1\end{pmatrix},\qquad
\cN_G(\theta_t)=\frac12\begin{pmatrix}1&\gamma t\\\gamma t&1\end{pmatrix}.
\end{equation}
Their eigenvalues give
\begin{align}
H(\widetilde A)_{\theta_t}&=h_2(19/20)\approx0.2863969571,\\
H(\widetilde A)_{\cN_G(\theta_t)}&=h_2(43/50)\approx0.5842388116.
\end{align}
Therefore
\begin{equation}
L_1^G(\theta_t)
=1+h_2(19/20)-2h_2(43/50)
\approx0.1179193338>\sS_1^{\eff}(p).
\label{eq:strict-example}
\end{equation}
The strict comparison can be certified without relying on floating-point
rounding. For $x>0$, set $v\coloneqq(x-1)/(x+1)$ and use
\begin{equation}
\left|\ln x-2\sum_{k=0}^{N-1}\frac{v^{2k+1}}{2k+1}\right|
\leq\frac{2|v|^{2N+1}}{(2N+1)(1-v^2)}.
\end{equation}
This is the geometric bound on the tail of the power series for
$2\operatorname{atanh}v$. Taking $N=128$ in the rational arguments involved,
and dividing by the positive corresponding interval for $\ln2$, gives
\begin{align}
0.46899&<h_2(9/10)<0.46900,\\
0.28639&<h_2(19/20)<0.28640,\\
0.58423&<h_2(43/50)<0.58424.
\end{align}
Thus the efficient value is less than $0.06202$, while the displayed quantum
memory value is greater than $0.11791$.
The same $N=128$ logarithm bounds also certify the refined witness
in~\eqref{eq:refined-witness}:
\begin{equation}
0.118034
<1+h_2(473/500)-2h_2(1071/1250)
<0.118035.
\label{eq:refined-witness-certificate}
\end{equation}
Applying the same $N=128$ bounds to the original $t=9/10$ score gives
\begin{equation}
0.117919<1+h_2(19/20)-2h_2(43/50)<0.117920.
\label{eq:original-witness-certificate}
\end{equation}
Thus the refinement is a strict improvement between two explicitly
achievable instruments. The ancillary script
\path{anc/explore_two_outcome_dephasing.py} checks the intervals
in~\eqref{eq:refined-witness-certificate} and
\eqref{eq:original-witness-certificate} using exact rational arithmetic.

The corresponding achievable primitive is
\begin{align}
C_0&=0,\\
Q_0&=-\tfrac12[h_2(43/50)-h_2(19/20)]\approx-0.1489209273,\\
E_0&=1-\tfrac12[h_2(19/20)+h_2(43/50)]\approx0.5646821156.
\end{align}
In the supporting direction at \(\mu=1\),
\begin{equation}
\lambda C_0+3Q_0+E_0=0.1179193338\ldots,
\end{equation}
which is strictly above the efficient bound. To include the unit protocols,
set $w\coloneqq(\lambda,3,1)$, where $1\leq\lambda\leq2$.
For the generating rays in~\eqref{eq:unit-rays},
\begin{equation}
w\cdot u_{\rm TP}=2-2\lambda\leq0,\qquad
w\cdot u_{\rm SD}=2\lambda-4\leq0,\qquad
w\cdot u_{\rm ED}=-2.
\end{equation}
Thus adding the unit-resource cone cannot increase this support, so the
advantage persists after allowing all unit protocols.

\begin{remark}[Logical content of the example]
The memory $F$ is a local output of Alice's instrument that is discarded,
not an additional supplied communication or entanglement resource.
This is a counterexample to efficient-instrument sufficiency, even for a two-qubit maximally correlated state. It is also a counterexample to the sufficiency of commuting conditional discarded memories. It is \emph{not} a counterexample to single-letterization: the improvement itself uses a single-copy instrument. No claim that \(t=9/10\) is globally optimal is needed.
\end{remark}

\subsection{The restricted bottleneck identity}\label{app:mc-bottleneck-identity}
We prove the identity in~\eqref{eq:restricted-exact} by expanding the
information quantities of the canonical memory $T=XF$.
Purity of $F$ conditional on the input label and the classical outcome
is the property that makes this expansion agree with the matrix score.
Throughout this proof, set $\alpha\coloneqq1+\mu$, as in
\eqref{eq:restricted-exact}.

For $T=XF$ from a canonical instrument, the extension~\eqref{eq:memory-source}
is the previously defined $\omega_{YXFE}$, and its $XFE$ marginal agrees with
$\sigma_{XFE}$. Since $F$ is pure conditional on $Y,X$,
\begin{align}
I(T;Y)_\omega
&=H(XF)_\omega-H(XF|Y)_\omega\\
&=H(X)_\omega+H(F|X)_\omega-H(X|Y)_\omega-H(F|YX)_\omega\\
&=I(X;Y)_\omega+H(F|X)_\omega.
\label{eq:memory-cost}
\end{align}
The chain rule also gives
\begin{equation}
I(T;E)_\omega=I(X;E)_\omega+I(F;E|X)_\omega.
\end{equation}
Substituting these identities, we find
\begin{align}
&H(Y)_\omega-\alpha H(E)_\omega+
\alpha I(T;E)_\omega-I(T;Y)_\omega\nonumber \\
&=H(Y)_\omega-\alpha H(E)_\omega+
\alpha I(X;E)_\omega+\alpha I(F;E|X)_\omega-I(X;Y)_\omega-H(F|X)_\omega\\
&=H(Y|X)_\omega-\alpha H(E|X)_\omega+\alpha[H(F|X)_\omega+H(E|X)_\omega-H(FE|X)_\omega]
-H(F|X)_\omega\\
&=H(Y|X)_\omega+(\alpha-1)H(F|X)_\omega-\alpha H(FE|X)_\omega\\
&=I(S\rangle BX)_\sigma-\mu I(SX;E|B)_\sigma.
\label{eq:bottleneck-identity}
\end{align}
The last equality follows from~\eqref{eq:matrix-c}--\eqref{eq:matrix-ell},
or directly from
\begin{align}
I(S\rangle BX)_\sigma&=H(Y|X)_\omega-H(FE|X)_\omega,\\
I(SX;E|B)_\sigma&=H(FE|X)_\omega-H(F|X)_\omega.
\end{align}

\subsection{Additivity of the unrestricted bottleneck functional}\label{app:mc-bottleneck-additivity}
We prove Theorem~\ref{thm:Jadd} for arbitrary classical--quantum
sources. Independence of the two source marginals and data processing
split a joint objective into valid one-source objectives. The proof
allows mixed conditional environment states and does not assume a
dimension bound on the memory.

\begin{proof}
Write the two sources as
$(\omega_k)_{Y_kE_k}=\sum_i p_{k,i}\proj i_{Y_k}\otimes e_{k,i,E_k}$
for $k\in\{1,2\}$. For a joint memory family $\left(\tau_{i_1i_2,T}\right)_{i_1,i_2}$, define
\begin{equation}
\Omega_{Y_1Y_2TE_1E_2}\coloneqq
\sum_{i_1,i_2}p_{1,i_1}p_{2,i_2}\proj{i_1}_{Y_1}\otimes\proj{i_2}_{Y_2}
\otimes\tau_{i_1i_2,T}\otimes e_{1,i_1,E_1}\otimes e_{2,i_2,E_2}.
\label{eq:joint-memory}
\end{equation}
Its unconditional source marginals satisfy
\begin{equation}
I(E_1;E_2)_\Omega=0,\qquad I(Y_1;Y_2)_\Omega=0.
\end{equation}
The chain rule gives
\begin{align}
\alpha I(T;E_1E_2)_\Omega-I(T;Y_1Y_2)_\Omega&=\alpha I(T;E_1)_\Omega-I(T;Y_1)_\Omega+\alpha I(T;E_2|E_1)_\Omega-I(T;Y_2|Y_1)_\Omega\\
&=\big[\alpha I(T;E_1)_\Omega-I(T;Y_1)_\Omega\big]+\big[\alpha I(TE_1;E_2)_\Omega-I(TY_1;Y_2)_\Omega\big].
\label{eq:J-chain}
\end{align}
For the last equality, independence gives
\begin{align}
I(TE_1;E_2)_\Omega
&=I(E_1;E_2)_\Omega+I(T;E_2|E_1)_\Omega
=I(T;E_2|E_1)_\Omega,\\
I(TY_1;Y_2)_\Omega
&=I(Y_1;Y_2)_\Omega+I(T;Y_2|Y_1)_\Omega
=I(T;Y_2|Y_1)_\Omega.
\end{align}
The marginal on $TE_1E_2$ is obtained from that on $TY_1E_2$ by preparing
$e_{1,i_1,E_1}$ from $Y_1$. Data processing yields
\begin{equation}
I(TE_1;E_2)_\Omega\leq I(TY_1;E_2)_\Omega.
\end{equation}
Since $\alpha\geq0$, the expression in~\eqref{eq:J-chain} is at most
\begin{align}
&\big[\alpha I(T;E_1)_\Omega-I(T;Y_1)_\Omega\big]+\big[\alpha I(TY_1;E_2)_\Omega-I(TY_1;Y_2)_\Omega\big].
\label{eq:J-split}
\end{align}
The first term is a valid one-source objective for $\omega_1$, with conditional
memory states
\begin{equation}
\overline\tau_{i_1,T}\coloneqq\sum_{i_2}p_{2,i_2}\tau_{i_1i_2,T}.
\end{equation}
The second term is a valid one-source objective for $\omega_2$, with memory
$TY_1$ and conditional states
\begin{equation}
\widetilde\tau_{i_2,TY_1}\coloneqq
\sum_{i_1}p_{1,i_1}\tau_{i_1i_2,T}\otimes\proj{i_1}_{Y_1}.
\end{equation}
Thus~\eqref{eq:J-split} is bounded by
$\sJ_\alpha(\omega_1)+\sJ_\alpha(\omega_2)$.
Taking the supremum over the joint memory proves subadditivity.

For the reverse inequality, choose separate memory families within $\varepsilon$
of their respective suprema and prepare their product conditional on
$(i_1,i_2)$. Mutual information is additive on the resulting product state.
The joint value is therefore at least the sum of the two suprema minus
$2\varepsilon$. Letting $\varepsilon\downarrow0$ proves the assertion.
\end{proof}

\subsection{The memory-class obstruction and the dynamic comparison}\label{app:mc-obstruction}
The preceding additivity proof applies to unrestricted memories.
We now identify why it does not directly establish additivity of the
restricted functional in~\eqref{eq:J-restricted}, and compare this
obstruction with the dynamic Hadamard-channel argument.

\subsubsection{Why the additive proof does not automatically apply to the restriction}
The proof above changes the conditional memory states by averaging
over one source label and by adjoining another label to the memory.
To reuse that proof for the exact static formula, these operations would
need to preserve the common pure-block form in~\eqref{eq:flagged-pure}.

In the proof of Theorem~\ref{thm:Jadd}, the derived memory
\begin{equation}
\overline\tau_{i_1,T}=\sum_{i_2}p_{2,i_2}\tau_{i_1i_2,T}
\end{equation}
can be a noncommuting mixed family without a common classical decomposition into
rank-one blocks. Even if every joint state $\tau_{i_1i_2,T}$ is pure, its average
need not be pure or have such a decomposition. The first term
in~\eqref{eq:J-split} is therefore bounded by the unrestricted functional,
but the argument does not bound it by the flagged-pure functional.

Appending a flag to purify or spectrally refine the memory need not improve
the objective. For an admissible extension of the source in
\eqref{eq:memory-source}, define
\begin{equation}
\Omega_{YTZE}\coloneqq
\sum_i p_i\proj i_Y\otimes\tau_{i,TZ}\otimes e_{i,E},
\qquad \Tr_Z[\tau_{i,TZ}]=\tau_{i,T}.
\label{eq:flagged-memory-extension}
\end{equation}
The change in the objective is
\begin{align}
&\big[\alpha I(TZ;E)_\Omega-I(TZ;Y)_\Omega\big]
-\big[\alpha I(T;E)_\Omega-I(T;Y)_\Omega\big]\nonumber \\
&=\alpha[I(TZ;E)_\Omega-I(T;E)_\Omega]
-[I(TZ;Y)_\Omega-I(T;Y)_\Omega]\\
&=\alpha I(Z;E|T)_\Omega-I(Z;Y|T)_\Omega.
\label{eq:flag-change}
\end{align}
Both conditional mutual informations are nonnegative. Conditional data
processing for the channel preparing $E$ from $Y$ also gives
$I(Z;E|T)_\Omega\leq I(Z;Y|T)_\Omega$. Thus at $\alpha=1$ the change
in~\eqref{eq:flag-change} is nonpositive. For $\alpha>1$, its sign is
not fixed in general. Conjecture~\ref{conj:memory} therefore requires
more than a generic claim that purification improves the objective.

\subsubsection{Contrast with the dynamic Hadamard proof}
The dynamic argument from \cite{WH12} uses a different optimization, in which the
input ensemble can be chosen. Comparing its entropy expression with
the fixed-diagonal static score explains why the channel additivity
result does not resolve~\eqref{eq:target-tensorization}.

For comparison, consider a channel isometry $U\colon A'\to BE$ and a pure-state
input ensemble $\left(p_x,\varphi^x_{AA'}\right)_x$. Define
\begin{equation}
\xi_{XABE}\coloneqq\sum_xp_x\proj x_X\otimes
(I_A\otimes U)\varphi^x_{AA'}(I_A\otimes U^\dagger).
\end{equation}
Write the evaluated dynamic support objective from~\cite{WH12} in entropy form as
\begin{align}
D_{a,b}(\xi)
&\coloneqq(1+b)H(B)_\xi+H(BE|X)_\xi+aH(B|X)_\xi-(1+a+b)H(E|X)_\xi.
\end{align}
Here $a,b\geq0$ are dynamic parameters, unrelated to the stochastic kernel
$a(x|i)$. The Hadamard degrading measurement supplies conditional entropy
inequalities with the required directions for this expression.

In the maximally correlated static problem, the matrix expression is instead
\begin{equation}
H(\widetilde A)_{\Delta(\theta)}+\mu H(\widetilde A)_\theta
-(1+\mu)H(\widetilde A)_{\cN_G(\theta)}.
\end{equation}
The auxiliary state $\theta$ is generally mixed. Its positive entropy term rewards
a nontrivial discarded quantum system, whereas the channel-output entropy has a
negative coefficient. Section~\ref{sec:qubit} gives an explicit example in which
restricting $\theta$ to be pure loses achievable rates. The dynamic additivity
theorem therefore does not directly establish~\eqref{eq:target-tensorization}.

The bottleneck formulation isolates the unresolved issue: the unrestricted
classical-input functional is additive, but the static instrument imposes a
special restriction on its memory states.

\Needspace{12\baselineskip}
\section{Proof of the complete dephasing contraction inequality}
\label{app:contraction}
This appendix proves Theorem~\ref{thm:contraction} for the coherence functional
in~\eqref{eq:partial-coherence}. The additional system $R$ is arbitrary,
finite dimensional, and untouched by the dephasing channel. We first prove
the binary pinching estimate used to control entropy production. We then
combine it with a relative-entropy inequality for successive dephasings
and integrate the resulting differential inequality. This order separates
the one-qubit estimate from its extension to arbitrarily correlated inputs.

\subsection{Entropy production for binary pinching}
The following lemma compares the coherence removed by one binary pinching
with the entropy production along its interpolation. Its proof uses an
even-power expansion with nonnegative coefficients and does not assume
that the pinched state commutes with the removed operator.

\begin{lemma}[Binary pinching entropy production]\label{lem:pinching}
Let $B$ be a quantum system, let $J=J^\dagger=J^{-1}$ be a unitary
involution on $B$, and define
$\Delta(\theta_B)\coloneqq(\theta_B+J\theta_BJ)/2$.
For a full-rank density operator $\theta_B$, set
$\zeta\coloneqq\Delta(\theta_B)$ and $\Xi\coloneqq\theta_B-\zeta$.
Here $\zeta$ is a density operator and $\Xi$ is a traceless Hermitian
operator on $B$. For $-1\leq t\leq1$ define
\begin{equation}
f(t)\coloneqq H(B)_\zeta-H(B)_{\zeta+t\Xi}.
\end{equation}
Then
\begin{equation}
f'(1)\geq2f(1).
\label{eq:entropy-production}
\end{equation}
\end{lemma}
\begin{proof}
We have $J\zeta J=\zeta$, $J\Xi J=-\Xi$, and $\Tr \Xi=0$.
The operators $\zeta+t\Xi$ are
positive definite density operators for $-1\leq t\leq1$, since their
endpoints are $\theta_B$ and $J\theta_BJ$. Positive definiteness also
holds on an open neighborhood of this interval, so the entropy
derivatives extend continuously to both endpoints. Thus $f(0)=0$ and
$f'(0)=\Tr[\Xi\ln \zeta]/\ln2=0$. The last equality follows from
$J\zeta J=\zeta$ and $J\Xi J=-\Xi$ by conjugating the trace by $J$.
The Fr\'echet derivative of the logarithm gives
\begin{equation}
f''(t)=\frac1{\ln2}\int_0^\infty
\Tr[\Xi(\zeta+t\Xi+sI)^{-1}\Xi(\zeta+t\Xi+sI)^{-1}]\,ds.
\label{eq:pinching-hessian}
\end{equation}
This identity follows by differentiating
$\ln T=\int_0^\infty[(1+s)^{-1}I-(T+sI)^{-1}]\,ds$.
Define $K_s\coloneqq(\zeta+sI)^{-1/2}\Xi(\zeta+sI)^{-1/2}$. The congruence
\begin{equation}
\zeta+t\Xi+sI=(\zeta+sI)^{1/2}(I+tK_s)(\zeta+sI)^{1/2}
\end{equation}
and cyclicity of the trace rewrite the integrand as
$\Tr[K_s(I+tK_s)^{-1}K_s(I+tK_s)^{-1}]$.
Since $K_s$ commutes with its own resolvent, this equals
$\Tr[K_s^2(I+tK_s)^{-2}]$. No commutation of $\zeta$ and $\Xi$ is assumed.
Since $\zeta\pm \Xi$ are positive definite,
$r\coloneqq\|\zeta^{-1/2}\Xi\zeta^{-1/2}\|<1$. The inequalities
$-r\zeta\leq \Xi\leq r\zeta$ imply $\|K_s\|\leq r$ for every $s\geq0$.
Thus the inverse-square power series converges uniformly for $0\leq t\leq1$.
Also $JK_sJ=-K_s$, so $\Tr K_s^{2m+1}=0$.
Expanding the inverse square therefore gives
\begin{equation}
f''(t)=\frac1{\ln2}\int_0^\infty\sum_{m=0}^\infty
(2m+1)t^{2m}\Tr K_s^{2m+2}\,ds,
\qquad0\leq t\leq1.
\end{equation}
All displayed summands are nonnegative, so Tonelli's theorem permits
interchanging their sum and integral. Near $s=0$, full rank bounds the
integrand; as $s\to\infty$, $\Tr K_s^2=O(s^{-2})$, and the remaining
series is bounded by a constant depending only on $r$.
Integrating twice from zero and using monotone convergence gives
\begin{equation}
f(t)=\sum_{m=1}^\infty a_m t^{2m}, \text{ where }
 a_m\coloneqq\frac1{2m\ln2}\int_0^\infty\Tr K_s^{2m}\,ds\geq0.
\label{eq:positive-series}
\end{equation}
Integrating the Hessian once instead gives
\begin{equation}
f'(t)=\sum_{m=1}^{\infty}2m a_m t^{2m-1},\qquad 0\leq t<1.
\end{equation}
The derivative is finite and continuous at $t=1$ by positive
definiteness. Monotone convergence as $t\uparrow1$ therefore gives
\begin{equation}
f'(1)=\sum_{m=1}^{\infty}2m a_m
\geq2\sum_{m=1}^{\infty}a_m=2f(1),
\end{equation}
as asserted.
\end{proof}

\subsection{From binary pinching to product dephasing}
We now prove Theorem~\ref{thm:contraction} using
Lemma~\ref{lem:pinching}. First, successive pinching and data processing
bound the total coherence by a sum of one-qubit contributions.
Applying the lemma to each contribution along the dephasing semigroup
then gives the claimed contraction with a constant independent of both
the number of qubits and the auxiliary system.
\begin{proof}
For $0\leq i\leq n$, define the density operators
$\theta^{(i)}_{\widetilde A^nR}\coloneqq(\Delta_1\circ\cdots\circ\Delta_i)(\theta_{\widetilde A^nR})$,
with $\theta^{(0)}=\theta$. Empty compositions below denote the identity channel.
We first establish the relative-entropy
inequality
\begin{align}
\mathcal C(\theta_{\widetilde A^nR})
&=\sum_{i=1}^n\left[H(\widetilde A^nR)_{\theta^{(i)}}
-H(\widetilde A^nR)_{\theta^{(i-1)}}\right]\\
&=\sum_{i=1}^nD\!\left(\theta^{(i-1)}_{\widetilde A^nR}\middle\|
(\Delta_1\circ\cdots\circ\Delta_{i-1})(\Delta_i(\theta_{\widetilde A^nR}))\right)\\
&\leq\sum_{i=1}^nD\!\left(\theta_{\widetilde A^nR}\middle\|
\Delta_i(\theta_{\widetilde A^nR})\right).
\label{eq:coherence-tensor}
\end{align}
The first equality follows by telescoping, and the second follows from the
pinching identity. The last inequality follows from data processing under $\Delta_1\circ\cdots\circ\Delta_{i-1}$, using commutation of
the pinching maps.

For full-rank $\theta_{\widetilde A^nR}$, define
$\theta_{t,\widetilde A^nR}\coloneqq
(\mathcal Z_{e^{-t}}^{\otimes n}\otimes\id_R)(\theta_{\widetilde A^nR})$
for $t\geq0$. This state remains full rank. The channel on site $i$ is
$\Delta_i+e^{-t}(\id-\Delta_i)$, and therefore
\begin{equation}
\frac{d\theta_t}{dt}=\sum_i(\Delta_i(\theta_t)-\theta_t).
\end{equation}
Its fully dephased state is constant:
$\Delta_{[n]}(\theta_t)=\widehat\theta$.
At a fixed $t$, apply Lemma~\ref{lem:pinching} to the joint system
$\widetilde A^nR$, with the involution $J_i$ acting as $Z$ on
$\widetilde A_i$ and as the identity on every other system. Define
\begin{align}
\zeta_{i,t}&\coloneqq\Delta_i(\theta_t),\qquad
\Xi_{i,t}\coloneqq\theta_t-\zeta_{i,t},\label{eq:pinching-time-states}\\
f_{i,t}(s)&\coloneqq H(\widetilde A^nR)_{\zeta_{i,t}}
-H(\widetilde A^nR)_{\zeta_{i,t}+s\Xi_{i,t}}.
\label{eq:pinching-time-function}
\end{align}
Here $\zeta_{i,t}$ is a density operator and $\Xi_{i,t}$ is a
traceless Hermitian operator on the joint system. The derivative
$f'_{i,t}(1)$ is with respect to $s$ at fixed $i,t$, as in the lemma.
It follows that
\begin{align}
\frac{d}{dt}\mathcal C(\theta_t)
&=\frac1{\ln2}\Tr\!\left[\frac{d\theta_t}{dt}\ln\theta_t\right]\\
&=-\sum_i\frac1{\ln2}\Tr[\Xi_{i,t}\ln\theta_t]\\
&=-\sum_i f'_{i,t}(1)\\
&\leq-2\sum_iD(\theta_t\|\Delta_i(\theta_t))\\
&\leq-2\mathcal C(\theta_t).
\end{align}
The first line differentiates
$H(\widetilde A^nR)_{\widehat\theta}-H(\widetilde A^nR)_{\theta_t}$,
using $\Tr[d\theta_t/dt]=0$. The second substitutes the evolution equation
and the definition of $\Xi_{i,t}$ in~\eqref{eq:pinching-time-states}.
The next equality follows by differentiating~\eqref{eq:pinching-time-function},
and the first inequality follows from Lemma~\ref{lem:pinching}.
The last inequality follows from~\eqref{eq:coherence-tensor} applied to $\theta_t$.
Hence $e^{2t}\mathcal C(\theta_t)$ is nonincreasing and
$\mathcal C(\theta_t)\leq e^{-2t}\mathcal C(\theta)$.
Set $\gamma\coloneqq e^{-t}$. The endpoints follow by continuity.
Finally, if $\theta_{\widetilde A^nR}$ is not full rank, set
$d\coloneqq\dim(\widetilde A^nR)$ and
$\theta^{[\eta]}_{\widetilde A^nR}\coloneqq(1-\eta)\theta_{\widetilde A^nR}+\eta I_{\widetilde A^nR}/d$.
Apply the full-rank result to this state and let $\eta\downarrow0$.
Entropy continuity completes the proof.
\end{proof}

\section{Construction of the capacity plots}\label{app:figures}
This appendix describes the construction of
Figures~\ref{fig:erased-three-dimensional},~\ref{fig:erased-slices},
and~\ref{fig:dephasing-slices}. The erased-state figures follow from the
exact region, whereas the dephasing curves combine explicit achievable
instruments with analytic outer bounds. The latter construction gives
no numerical optimality claim for the unresolved part of the region.

The three-dimensional drawing in Figure~\ref{fig:erased-three-dimensional}
uses the exact five-facet description in~\eqref{eq:five-first}--\eqref{eq:five}, with $p=1/4$
and $h=1$. Each face is intersected with a common plane to obtain a bounded
polygon for display. No additional face is drawn on this plane: dashed
segments mark the display boundary, and arrows indicate the unbounded
edges. The two vertices and the three coordinate axes use the same
orthographic projection. The accompanying script checks that every polygon
satisfies all five inequalities and lies on its designated supporting plane.

The erased-state curves in Figure~\ref{fig:erased-slices} are the exact
piecewise-linear functions~\eqref{eq:erased-Q-slice} and
\eqref{eq:erased-C-slice}, with $h=1$. Each corner is included explicitly
in the plotting grid. No optimization is needed for that figure.

For Figure~\ref{fig:dephasing-slices}, fix $p=1/2$ and $\gamma=0.8$.
For each $0\leq q\leq1/2$ and $0\leq t\leq1$, take two canonical
outcomes of equal probability, specified by the density operators
$\theta_{\widetilde A}(q,t)$ and $\theta_{\widetilde A}(1-q,t)$ from
\eqref{eq:qubit-test-state}. Their mean diagonal is $(1/2,1/2)$, so they
satisfy the exact instrument constraint in~\eqref{eq:diagonal-constraint}.
Both outcomes have the same entropies. With $b$ as in~\eqref{eq:qubit-entropy},
define their quantum communication consumption and entanglement generation rates by
\begin{align}
Q_c(q,t)&\coloneqq\tfrac12[b(q,\gamma t)-b(q,t)],\label{eq:plot-primitive-Q}\\
E_g(q,t)&\coloneqq h_2(q)-\tfrac12[b(q,t)+b(q,\gamma t)].
\label{eq:plot-primitive}
\end{align}
Equation~\eqref{eq:qubit-point} proves that $(0,-Q_c(q,t),E_g(q,t))$
is achievable. The difference $E_g(q,t)-Q_c(q,t)$ equals
$h_2(q)-b(q,\gamma t)$ and is nonnegative. The cost is nonnegative because
dephasing increases the entropy of a qubit.

The supplied script samples 101 equally spaced $q$ values on $[0,1/2]$.
For $t$, it takes the union of 801 equally spaced values on $[0,1]$ and
121 logarithmically spaced values on $[10^{-5},0.1]$, giving 920 distinct
values. These 92,920 pairs define feasible instruments, irrespective of
whether an optimization over the continuous family could do better.
The refinement adds balanced instruments with $q=1/2$ selected by
continuous optimization in $t$. For fixed $\mu$, the interior stationary
points of $F_\mu(1/2,t)$ satisfy
\begin{equation}
\mu\atanh t=(1+\mu)\gamma\atanh(\gamma t).
\label{eq:plot-stationary}
\end{equation}
The code compares these candidates with the endpoints $t=0,1$ for a
grid of support parameters below $\mu_\gamma$, and also includes
the rational witness $t=223/250$. Every selected value specifies a
feasible instrument, independently of whether it is a global maximizer.
The parameter grid and all selected values are recorded in the ancillary
code and data. Compared with the uniform and logarithmic grid alone, this refinement
raises the inner curves by approximately $8.7\times10^{-5}$ ebits per copy at their largest
separation; it leaves the analytic outer bounds unchanged.
For a sample indexed by $j$, write $Q_{c,j}$ and $E_{g,j}$ for the rates
in~\eqref{eq:plot-primitive-Q}--\eqref{eq:plot-primitive}. For $v\geq0$, solve the finite linear program
\begin{equation}
\begin{split}
F(v)\coloneqq\max_{\left(w_j\right)_j}\;&\sum_jw_j(E_{g,j}-Q_{c,j})\\
\text{subject to }\;&w_j\geq0,\quad\sum_jw_j\leq1,\quad
\sum_jw_jQ_{c,j}\leq v.
\end{split}
\label{eq:plot-LP}
\end{equation}
The unused weight is assigned to the trivial instrument. It is enough to
retain the selected points on the upper convex hull in the
$(Q_c,E_g-Q_c)$ plane; deleting other samples cannot compromise
achievability.

Teleporting the consumed qubits converts each primitive point into
$(-2Q_{c,j},0,E_{g,j}-Q_{c,j})$. Thus the $Q=0$ slice has the achievable
boundary $E=F((-C)/2)$. For the $C=0$ slice, unused supplied qubits
can instead generate ebits by entanglement distribution. Its achievable
boundary is $E=-Q+F(-Q)$. These arguments establish feasibility directly
from the primal weights in~\eqref{eq:plot-LP}.
The parameter grid includes the pure endpoint $(q,t)=(1/2,1)$ exactly.
Accordingly the achievable and outer bounds coincide once
$-C\geq h_2(0.9)$ or $-Q\geq h_2(0.9)/2$, respectively.

The script evaluates entropies in double precision and reconstructs each
objective from its primal weights. It scales those weights down whenever
needed to satisfy the two resource constraints numerically. All retained
hull vertices and analytic outer-bound corners are inserted in the final
grid. Line segments between achievable points are achievable by time
sharing. The figure is a numerical rendering of this explicit feasible
construction; no numerical optimality claim for the continuous or
regularized problem is used. The ancillary directory \texttt{anc/}
contains the plotting scripts, TikZ sources, input data, retained
instruments, and verification records. The command
\texttt{python anc/reproduce\_all.py} reproduces all five figure PDFs;
the dependencies and instructions are given in \texttt{anc/README.md}.

\end{document}